\documentclass[letterpaper]{article} %
\newif\iflong
\longtrue
\iflong
\usepackage{aaai2027_hacked}  %
\else
\usepackage[submission]{aaai2027}  %
\fi
\usepackage[hyphens]{url}  %
\usepackage{graphicx} %
\usepackage{natbib}  %
\usepackage{caption} %
\usepackage{algorithm}
\usepackage{algorithmic}

\usepackage{newfloat}
\usepackage{listings}
\DeclareCaptionStyle{ruled}{labelfont=normalfont,labelsep=colon,strut=off} %
\floatstyle{ruled}
\newfloat{listing}{tb}{lst}{}
\floatname{listing}{Listing}

\usepackage{booktabs}

\title{Nearly Group-Separable Elections}
\author {
    Piotr Faliszewski\textsuperscript{\rm 1}, 
    Jan Jabrocki\textsuperscript{\rm 1}, 
    Stanisław Kaźmierowski\textsuperscript{\rm 1,2},
    Kristýna Pekárková\textsuperscript{\rm 3},\\
    Šimon Schierreich\textsuperscript{\rm 1,4}, 
    Ildikó Schlotter\textsuperscript{\rm 5,6}
}
\affiliations {
    \textsuperscript{\rm 1} AGH University of Krakow, Poland\\
    \textsuperscript{\rm 2} University of Warsaw, Poland\\
    \textsuperscript{\rm 3} Charles University, Czechia\\
    \textsuperscript{\rm 4} Czech Technical University in Prague, Czechia\\
    \textsuperscript{\rm 5}
    ELTE Centre for Economic and Regional Studies, Hungary\\
    \textsuperscript{\rm 6} Budapest University of Technology and Economics, Hungary\\
    faliszew@agh.edu.pl, jjabrocki@student.agh.edu.pl, s.kazmierowski@uw.edu.pl,
    kristyna.pekarkova@mail.muni.cz,\\ schiesim@fit.cvut.cz, schlotter.ildiko@krtk.elte.hu
}

\usepackage{subcaption}

\usepackage{amsmath}
\usepackage{amsfonts}
\usepackage{amssymb}
\usepackage{amsthm}
\usepackage{nicefrac}
\usepackage{bm}

\iflong
\usepackage{hyperref}
\hypersetup{
    pdfencoding=auto, 
    psdextra,
    colorlinks=true,
    citecolor=green!40!black,
    linkcolor=red!50!black,
    urlcolor=blue!80!black
}
\usepackage{stfloats}
\usepackage{cuted}
\input{insbox}
\fi
\usepackage{cleveref}

\usepackage{enumitem}
\setlist[description]{%
    font={\normalfont\itshape},
}
\usepackage{xspace}
\usepackage{comment}

\usepackage{tabularx}
\usepackage{booktabs}
\usepackage{multirow}
\newcolumntype{Y}{>{\centering\arraybackslash}X}

\usepackage{todonotes}

\usepackage{tikz}
\usetikzlibrary{decorations.pathreplacing,calc,arrows.meta}
\newcommand{\mysetminus}{\hbox{\tikz{\draw[line width=0.6pt,line cap=round] (3pt,0) -- (0,6pt);}}}

\usepackage{thm-restate}
\newcommand{\linkproof}[1]{%
    \iflong
    \hyperref[app:proof:#1]{$\star$}%
    \else
    {$\star$}%
    \fi 
}
\newcommand{\toappendix}[1]{%
  \gappto{\appendixtext}{
    {#1}
   }
}
\newcommand{\prooftoappendix}[3]{%
  \gappto{\appendixtext}{
    \subsection{Proof of \Cref{#1}}\label{app:proof:#1}
    #2
    \begin{proof}
    #3\end{proof}
  }
}

\newcommand{\prooftoappendixdivided}[3]{%
  \gappto{\appendixtext}{
    \subsection{Proof of \Cref{#1}}\label{app:proof:#1}
    #2
    
    #3
  }
}

\newcommand{\appendixsection}[1]{%
  \gappto{\appendixtext}{
    \section{Additional material for \Cref{#1}}
    \label{app:#1}
  }
}

\newcommand{\appendixsubsection}[1]{%
  \gappto{\appendixtext}{
    \subsection{Additional material for \Cref{#1}}
    \label{app:#1}
  }
}

\usepackage{pict2e}
\newcommand{\opentriangle}{%
  \raisebox{0.2pt}{\makebox[0.77778em]{%
    \setlength{\unitlength}{0.6em}%
    \linethickness{0.4pt}\roundjoin
    \begin{picture}(1,1)
    \polygon(0,0)(1,0)(1,1)
    \end{picture}%
  }}%
}\newenvironment{proofsketch}{\par
  \pushQED{\hfill \opentriangle}%
  \normalfont \topsep 12pt\relax
  \trivlist
  \item[\hskip\labelsep
        {\itshape Proof sketch.}]\ignorespaces
}{%
  \popQED\endtrivlist
}

\usepackage{algorithm}
\usepackage{algorithmic}
\usepackage{newfloat}
\usepackage{listings}
\DeclareCaptionStyle{ruled}{labelfont=normalfont,labelsep=colon,strut=off} %
\floatstyle{ruled}
\newfloat{listing}{tb}{lst}{}
\floatname{listing}{Listing}

\newcommand{\probName}[1]{\textsc{#1}\xspace}

\newcommand{\swapNE}[1]{\probName{Nearly \allowbreak $#1$ \allowbreak Election}}
\newcommand{\KemenyOA}{\probName{Kemeny \allowbreak Aggregation}}

\newcommand{\candidates}{C}
\newcommand{\numCandidates}{m}
\newcommand{\voters}{V}
\newcommand{\vote}{v}
\newcommand{\numVoters}{n}
\newcommand{\budget}{d}
\newcommand{\mappinglist}[1]{(#1)}
\def\mapping{\varphi}
\def\arcset{A}
\newcommand{\dist}{\operatorname{dist}}

\newcommand{\axis}{\mathrel{\triangleright}}

\newcommand{\N}{\mathbb{N}}

\newcommand{\rank}{{\mathrm{pos}}}
\newcommand{\cost}{{\mathrm{cost}}}

\newcommand{\swap}{{\mathrm{swap}}}

\newcommand{\calL}{{\mathcal{L}}}
\newcommand{\calD}{{\mathcal{D}}}
\newcommand{\calX}{{\mathcal{X}}}

\newcommand{\ora}[1]{\overrightarrow{#1}}
\newcommand{\ola}[1]{\overleftarrow{#1}}

\DeclareMathOperator{\topc}{top}

\newcommand{\restr}[2]{#1_{\restriction_{#2}}}
\newcommand{\Inversions}{\ensuremath{\mathrm{Inv}}}

\newcommand{\opt}{\ensuremath{\mathrm{opt}}}
\newcommand{\splitcost}{\ensuremath{\mathrm{sc}}}
\newcommand{\treeheight}{\ensuremath{t}}
\newcommand{\Tcat}{T^{\mathrm{cat}}}

\def\typePQF{\ensuremath{\mathrm{PQF}}}
\def\typeP{\ensuremath{\mathrm{P}}}
\def\typeQ{\ensuremath{\mathrm{Q}}}
\def\typeF{\ensuremath{\mathrm{F}}}

\def\domain{\mathcal{D}}
\newcommand{\SPeak}{\ensuremath{\mathrm{Single}\text{-\allowbreak}\mathrm{Peaked}}}
\newcommand{\SCros}{\ensuremath{\mathrm{Single}\text{-}\allowbreak\mathrm{Crossing}}}
\newcommand{\GS}{\ensuremath{\mathrm{GS}}}
\newcommand{\GSbal}{\ensuremath{\mathrm{GS}\text{/}\allowbreak\mathrm{bal}}}
\newcommand{\GScat}{\ensuremath{\mathrm{GS}\text{/}\allowbreak\mathrm{cat}}}
\newcommand{\GSbin}{\GS}
\newcommand{\Strat}[1][k]{\ensuremath{#1\text{-}\allowbreak\mathrm{Stratification}}}
\newcommand{\Ident}{\mathrm{Identity}}
\newcommand{\Antag}{\mathrm{Antagonism}}
\newcommand{\Separ}{\mathrm{Separation}}

\newcommand{\Oh}[1]{{\mathcal{O}\mathopen{}\left(#1\right)}}
\newcommand{\Ohstar}[1]{{\mathcal{O}^*\mathopen{}\left(#1\right)}}

\NewDocumentCommand{\cc}{ O{} O{} m }{\mbox{%
    \expandafter\ifx\expandafter\relax\detokenize{#2}\relax\else{#2-}\fi%
    \ensuremath{{\mathrm{#3}}}%
    \expandafter\ifx\expandafter\relax\detokenize{#1}\relax\else{-#1}\fi%
    }\xspace}
\newcommand{\DeclareComplexityLowerBound}[1]{%
  \expandafter\newcommand\csname #1\endcsname{\cc{#1}}%
  \expandafter\newcommand\csname #1h\endcsname{\cc[hard]{#1}}%
  \expandafter\newcommand\csname #1hness\endcsname{\cc[hardness]{#1}}%
  \expandafter\newcommand\csname #1c\endcsname{\cc[complete]{#1}}%
  \expandafter\newcommand\csname #1cness\endcsname{\cc[completeness]{#1}}%
}
\newcommand{\DeclareComplexityLowerBoundParam}[2]{%
  \expandafter\newcommand\csname #1\endcsname[1][#2]{\cc{#1[##1]}}%
  \expandafter\newcommand\csname #1h\endcsname[1][#2]{\cc[hard]{#1[##1]}}%
  \expandafter\newcommand\csname #1hness\endcsname[1][#2]{\cc[hardness]{#1[##1]}}%
  \expandafter\newcommand\csname #1c\endcsname[1][#2]{\cc[complete]{#1[##1]}}%
  \expandafter\newcommand\csname #1cness\endcsname[1][#2]{\cc[completeness]{#1[##1]}}%
}
\newcommand{\DeclareComplexityUpperBound}[1]{
  \expandafter\newcommand\csname #1\endcsname{\cc{#1}}%
}

\DeclareComplexityLowerBound{NP}
\DeclareComplexityLowerBound{coNP}
\DeclareComplexityLowerBound{paraNP}
\DeclareComplexityLowerBoundParam{W}{1}
\DeclareComplexityLowerBoundParam{coW}{1}
\DeclareComplexityUpperBound{FPT}
\DeclareComplexityUpperBound{XP}
\DeclareMathOperator{\poly}{poly}

\newtheorem{theorem}{Theorem}[section]
\newtheorem{definition}{Definition}[section]
\newtheorem{lemma}[theorem]{Lemma}
\newtheorem{corollary}[theorem]{Corollary}
\newtheorem{proposition}[theorem]{Proposition}

\newtheorem{claim}[theorem]{Claim}
\Crefname{claim}{Claim}{Claims}
\newenvironment{claimproof}{\begin{proof}}{\end{proof}}

\newcommand{\proofsubparagraph}[1]{\smallskip\emph{#1}\hspace{0.15cm}}

\newcolumntype{M}[1]{>{\centering\arraybackslash}m{#1}}

\newcounter{taskcntr}
\stepcounter{taskcntr}

\newcommand{\lca}{\operatorname{lca}}
\newcommand{\Tbal}{T^{\mathrm{bal}}}

\begin{document}

\maketitle

\begin{abstract}
    We study the problem of computing how close a given election is to being group-separable, measuring proximity by swaps of adjacent candidates in the votes. We also consider several other domains, including caterpillar group-separable, balanced group-separable, single-peaked, and single-crossing ones. Our problem is generally intractable, but we find practical \FPT algorithms parameterized by the number of candidates or swaps. For the latter case, our algorithm applies to all domains characterized by finite forbidden subelections, resolving a well-established open problem. We supplement our theoretical findings with experimental analysis. 
\end{abstract}

\section{Introduction}

An ordinal election consists of a set of candidates and a collection of voters ranking these candidates from the most to the least appealing one: We are interested in computing how close such elections are to being group-separable ($\GS$), balanced group-separable ($\GSbal$), or caterpillar group-separable ($\GScat$), measuring proximity in terms of the number of swaps of adjacent candidates in the votes. So far, this line of work on election proximity to having structure mostly focused on the single-peaked domain~\citep{fal-hem-hem:j:nearly-sp,cor-gal-spa:c:sp-width,erd-lac-pfa:j:nearly-sp}, the single-crossing one~\citep{lak-pet-elk:c:nearly-single-crossing,jae-pet-elk:c:nearly-single-crossing}, or both of them jointly~\citep{elk-lac:c:nearly,cor-gal-spa:c:spsc-width}. Group-separability also received some attention, but for different notions of distance~\citep{bre-che-woe:j:distance,kra-elk:c:explaining-prefs}.

Group-separability was introduced by \citet{ina:j:group-separable,ina:j:simple-majority}, and is currently receiving increased attention; see, e.g., the works of \citet{kar:j:group-separable,fal-kar-obr:c:group-separable,jan-lan-lis-szu:c:consistent-subelections,szu-boe-bre-fal-nie-sko-sli-tal:j:map,boe-bre-elk-fal-szu:c:map-pref-learning}. Inada's original definition says that an election is group-separable if every subset of its candidates (of size at least two) can be split into two, such that each voter prefers members of one of them to members of the other.  Recently, \citet{kar:j:group-separable} provided an alternative definition based on PQ-trees and, using this approach, \citet{fal-kar-obr:c:group-separable} identified its balanced and caterpillar subdomains. The main difference between $\GSbal$ and $\GScat$ is that the former uses minimum-height trees, whereas the latter uses maximum-height ones. Based on this, \citet{fal-kar-obr:c:group-separable} showed that many problems that are \NPh in general, become polynomial-time solvable for $\GSbal$ elections but remain intractable for $\GScat$ ones. As tractability under structured domains is the more common behavior, often observed, e.g., for single-peaked and single-crossing domains, this is our first bit of evidence that $\GScat$ is quite a special domain; for examples of tractability under structured domains, see the survey of \citet{elk-lac-pet:t:restricted-domains-survey}.

The second bit of evidence was provided by \citet{szu-boe-bre-fal-nie-sko-sli-tal:j:map}, who have shown that computing Dodgson scores for $\GScat$ elections using an ILP solver is notably slower than for the other distributions they considered (including those over $\GSbal$, single-peaked, and single-crossing domains).

The third bit was recently provided by \citet{fal-sor-szu-was:c:kemeny-inner-diversity,fal-sor-szu-was:c:outer-diversity}, who demonstrated that $\GScat$ is a
surprisingly diverse domain. In particular, they showed that if we draw a vote uniformly at random, then its expected swap distance to the closest $\GScat$ domain is significantly smaller than to $\GSbal$, single-peaked, or single-crossing domains.
This result strongly relies on the assumption that random votes are completely independent, but since votes in elections are correlated, a priori, it is not clear whether various elections would be closer to being $\GScat$, $\GSbal$, single-peaked, or single-crossing. We want to answer this question, and to do so, we need efficient algorithms to %
answer such questions.
\begin{table*}[bt]
    \centering
    \renewcommand*{\arraystretch}{1.2}
    \scalebox{1}{
    \begin{tabularx}{\linewidth}{r|Y|c|Y|Y}
        \toprule
         & --- & $\numVoters$ & $\numCandidates$ & $\budget$ \\
         \midrule
         $\GS$ 
            & ? & ? & $\mathcal{O}^*(3^\numCandidates)$ \, {\small [T\ref{thm:NE:swap:GSbin:FPT:numCandidates}]}& 
            $\mathcal{O}^*(6^\budget)$ \, {\small [C\ref{thm:NE:swap:GS:FPT:budget}]}\\
        $\GScat$ 
            & \NP-c \, {\small{[T\ref{thm:NE:swap:GScat:NPc}]}} 
            & para-\NP-c \, {\small{[T\ref{thm:NE:swap:GScat:NPc}]}}
            & $\mathcal{O}^*(2^\numCandidates)$ \, {\small [T\ref{thm:NE:swap:GScat:FPT:numCandidates}]}
            & %
            $\mathcal{O}^*(8^{\budget})$ \, {\small [C\ref{thm:NE:swap:GS:FPT:budget}]} \\
        $\GSbal$ 
            & \NP-c \, {\small{[T\ref{thm:NE:swap:GSbal:NPc}]}} 
            & ? 
            & $\mathcal{O}^*((2\sqrt{2})^\numCandidates)$ \, {\small [T\ref{thm:NE:swap:GSbal:FPT:numCandidates}]} 
            & $\mathcal{O}^*(2^{2\budget^2})$ \, {\small [T\ref{thm:NE:swap:GSbal:GScat:FPT:budget}]} \\
         \midrule
         $\Ident$ 
            & \NP-c \, [\small BTT] 
            & para-\NP-c \, {\small [DKNS]} 
            & $\mathcal{O}^*(2^{\numCandidates})$ \, {\small [BFGNR]}
            & $\mathcal{O}^*(1.53^{\budget})$ \, {\small [BFGNR]} \\
         $\Antag$ 
            & \NP-c \, {\small [T\ref{thm:NE:swap:antag:NPc}]}
            & para-\NP-c  \, {\small [T\ref{thm:NE:swap:antag:NPc}]}
            & $\mathcal{O}^*(2^\Oh{\numCandidates\log\numCandidates})$  \, {\small [T\ref{thm:NE:swap:any:FPT:numCandidates}]} 
            & $\mathcal{O}^*(3.06^{\budget})$  \, {\small [T\ref{thm:NE:swap:FPT:budget:collapsed}]} \\
         $\Separ$ 
            & \NP-c \, {\small [T\ref{thm:NE:swap:threeStrat:NPc}]} 
            & $\mathcal{O}^*(2^{\numVoters})$ \, {\small [T\ref{thm:NE:swap:twoStrat:poly+FPT:voters}]}
            & $\mathcal{O}^*(2^{\numCandidates})$ \, {\small [T\ref{thm:NE:swap:GScat:FPT:numCandidates}]}
            & $\mathcal{O}^*(5^\budget)$ \, {\small [T\ref{thm:NE:swap:FPT:budget:collapsed}]}
             \\
        \multirow{2}{*}{
          $\Strat$} 
            & \cc{P} {\small($k\leq 2$) \, [T\ref{thm:NE:swap:twoStrat:poly+FPT:voters}]}
             & \multirow{2}{*}{para-\NP-c {\small [T\ref{thm:NE:swap:mhalfStrat:NPc}]}} & 
             \multirow{2}{*}{$\mathcal{O}^*(3^{\numCandidates})$ \, {\small [T\ref{thm:NE:swap:GScat:FPT:numCandidates}]}}
             & \multirow{2}{*}{$\mathcal{O}^*(3^\budget)$ \, {\small [T\ref{thm:NE:swap:FPT:budget:collapsed}]}} 
             \\
           & \NP-c {\small($k\geq 3$) \, [T\ref{thm:NE:swap:threeStrat:NPc}]}
           & & &
           \\
        \Strat[\nicefrac{\numCandidates}{2}]
            & \NP-c \, {\small [T\ref{thm:NE:swap:mhalfStrat:NPc}]} & para-\NP-c {\small [T\ref{thm:NE:swap:mhalfStrat:NPc}]} & $\mathcal{O}^*(2^{\numCandidates})$ \, {\small [T\ref{thm:NE:swap:GScat:FPT:numCandidates}]} &
            $\mathcal{O}^*(2.84^\budget)$ \, {\small [C\ref{thm:NE:swap:FPT:budget:collapsed}]}
            \\
         \midrule 
         \SPeak 
            & \NP-c \, {\small [ELP]}& para-\NP-c \, {\small [ELP]} & $\mathcal{O}^*(2^{\numCandidates\log\numCandidates})$ \, {\small [O\ref{thm:NE:swap:singlePeak:FPT:numCandidates}]} & $\mathcal{O}^*(6^\budget)$ \, {\small [C\ref{thm:NE:swap:singlpeaked:FPT:budget}]}
            \\
         $\SCros$ 
            & \NP-c {\small \iflong
            [LPE,\,P\ref{thm:NE:swap:singlecrossing:NPh}]
            \fi 
            } & ? & $\mathcal{O}^*(2^{\numCandidates^2\log\numCandidates})$ \, {\small [T\ref{thm:NE:swap:singleCrossing:FPT:numCandidates}]} & $\mathcal{O}^*(9^\budget)$ \, {\small [C\ref{thm:NE:swap:singlcrossing:FPT:budget}]}
            \\
\bottomrule
    \end{tabularx}
    }
    \caption{An overview of our complexity and algorithmic results.
    Citations for known results are abbreviated as follows: 
    {\small BFGNR} refers to~\cite{bet-fel-guo-nie-ros:j:fpt-kemeny-aaim},
    {\small BTT} to~\cite{bar-tov-tri:j:manipulating},
    {\small DKNS} to~\cite{dwo-kum-nao-siv:c:rank-aggregation}, 
    {\small ELP} to~\cite{erd-lac-pfa:j:nearly-sp},
    and 
    {\small LPE} to~\cite{lak-pet-elk:c:nearly-single-crossing}.
    }
    \label{tab:summary}
\end{table*}

\paragraph{Contributions.}
We show that computing the distance from an election to the closest $\GScat$ or $\GSbal$ one is $\NP$-complete, for $\GScat$ even for the case of four voters. For $\GSbal$, such a strong hardness result is %
elusive, but we give an intuitive explanation as to  why this is the case by considering %
several related domains for which  %
the complexity %
varies significantly. We %
circumvent these hardness results by  (a)~showing reductions to ILP, and giving  (b)~$\FPT$ algorithms parameterized by the number of candidates or %
the distance from respective domains (see Table~\ref{tab:summary}): %
\begin{enumerate}
    \item We evaluate running times of our algorithms experimentally and, somewhat surprisingly, we find the ILP formulations to be completely impractical. On the positive side, combinatorial $\FPT$ algorithms perform very well. %
    \item One of our algorithms, parameterized by the distance from a domain, applies to all domains characterized by forbidden subelections (such as $\GS$, $\GScat$, single-peaked and single-crossing ones), resolving the ``global swaps'' variant of Open Problem~5 of \citet{elk-lac-pet:t:restricted-domains-survey}, dating back to \citet{erd-lac-pfa:j:nearly-sp} %
    and %
    \citet{lak-pet-elk:c:nearly-single-crossing}. %
\end{enumerate}

Experiments show that various elections, including those from the map of elections~\citep{szu-boe-bre-fal-nie-sko-sli-tal:j:map,fal-kac-sor-szu-was:c:div-agr-pol-map}, are quite close to being $\GScat$. %
Yet, transforming an election to being $\GScat$ is often more likely to change its winners than transforming it to being $\GSbal$.
Apparently, the former conversions introduce fewer distortions overall, but
distort in
those parts %
that matter for winner determination.

\section{Preliminaries}
\label{sec:prelim}

For $k \in \N$, we let $[k]=\{1,\ldots,k\}$.
For a finite set~$X$, we write $\calL(X)$ to
denote the set of all linear orders over~$X$. For a linear order
${\succ}\in\calL(X)$ and a subset $Y\subseteq X$, the
\emph{restriction} of $\succ$ to $Y$, denoted $\restr{\succ}{Y}$, is
the linear order in $\calL(Y)$ in which $a$ precedes $b$ if and only if
$a\succ b$.

\paragraph{Candidates, Votes, and Elections.}

An \emph{election} is a pair $E=(\candidates,\voters)$, where
$\candidates$ is a set of $\numCandidates$ \emph{candidates} and
$\voters$ is a set of $\numVoters$ \emph{voters}, where each voter
$v \in \voters$ has a \emph{vote} $\succ_v \in \calL(\candidates)$.
To simplify notation, we often write $v$ to refer both to the voter and
his or her vote, but the meaning will always be clear.
For a vote $v \in\calL(\candidates)$ and a
candidate~$c\in\candidates$, we write $\rank_v(c)$ for the position of
$c$ in $v$; the top-ranked candidate has position~$1$ and the
bottom-ranked one has position~$\numCandidates$. By
$\topc(v)=\rank_v^{-1}(1)$ we mean the top-ranked candidate
in~$v$. %

For a set~$V$ of voters and a subset
$\candidates'\subseteq\candidates$ of candidates,
by~$\restr{V}{\candidates'}$ we mean
restricting the votes in $V$
to candidates in~$\candidates'$.  A
\emph{subelection} of~$E=(\candidates,\voters)$ \emph{induced}  \emph{by}~$\candidates' \subseteq \candidates$ and a subset~$\voters'$ of voters
is the election $(\candidates',\restr{\voters'}{\candidates'})$,
whereas the \emph{restriction} of~$E$ to~$C'$, written
as~$\restr{E}{\candidates'}$, is the subelection induced
by~$\candidates'$ and~$\voters$.
Elections~$E=(\candidates,\voters)$ and $\hat{E}=(\hat\candidates,\hat\voters)$
are \emph{isomorphic} if they can be made identical by renaming their candidates and voters.
An election $\hat{E}$ is \emph{present} in~$E$ if
there exists a subelection~$E'$ of~$E$ that is isomorphic
with~$\hat{E}$.

A \emph{voting rule} is a function $f$ that given an election
outputs a nonempty subset of winning candidates. For example,
\emph{Plurality} outputs the candidates with the highest number of top
positions and
\emph{Borda} outputs those with the highest average position.
Given an election $E = (\candidates, \voters)$, its \emph{majority relation}
is $M(E) = \{ (a,b) \in \candidates \times \candidates \colon$ strict
majority of voters prefer~$a$ to~$b\}$.
\emph{Copeland rule} assigns each candidate $c$
one point for each candidate $d$ such that $(c,d) \in M(E)$ and
half a point for each candidate $d$ such that exactly half of
the voters prefer~$c$ to~$d$; then it outputs the candidates with the highest score.
Copeland is \emph{Condorcet consistent}, i.e., whenever there is a
candidate $c$  preferred to every other one by a (possibly
different) strict majority of the voters, $c$ is the unique winner.

\paragraph{Preference Domains.}

A \emph{(preference) domain} $\domain$ over candidate set
$\candidates$ is a non-empty subset of $\calL(C)$. %
The single-peaked domain is among the best-known
ones~\citep{bla:j:rationale-of-group-decision-making}.
\begin{definition}
  Let axis $\axis$ be a linear order over candidate set~$\candidates$. A vote $v$ is
  \emph{single-peaked on~$\axis$}
  if for every $t \in [|\candidates|]$, its $t$ top-ranked candidates
  form an interval within $\axis$.
  A domain~$\domain$ is \emph{single-peaked} if there exists an axis
  ${\axis}\in\calL(\candidates)$ %
  such that $\domain$ consists of all the votes that are single-peaked on
  $\axis$.
\end{definition}

Note that %
each set of
candidates and each axis yields %
a different single-peaked
domain. Consequently, we will often consider families of
domains, %
typically also referring to them  as
domains.
For a (family of) domain(s) $\calX$, $E$ is an $\calX$ election
if all its votes belong to (a domain from)~$\calX$.
By $\calX(\candidates)$ we mean the
restriction of $\calX$ to domains over candidate set~$\candidates$.

We primarily consider domains defined using extensions of
PQ-trees~\citep{boo-lue:j:consecutive-ones-property}, following the
way in which \citet{kar:j:group-separable} characterized
\emph{group-separability}.  To this end, a $\typePQF$-tree
over a candidate set $\candidates = \{c_1, \ldots, c_m\}$ is an
ordered rooted tree~$T$ with exactly~$\numCandidates$ leaves
$\ell_1,\ldots,\ell_\numCandidates$, listed from left to right. Each
internal node has at least two children and a type that constrains how
its children may be permuted: A node of type~$\typeP$ allows arbitrary
permutations, a node of type~$\typeQ$ only allows reversals, and a
node of type~$\typeF$ fixes the order of the children as is.
A \emph{mapping}~$\mapping$ is a bijection between~$\candidates$ and
the leaves of~$T$; $\mapping^{-1}(\ell_i)$ is the
candidate placed at leaf~$\ell_i$. We also write
$\mapping=(\mapping_1,\ldots,\mapping_{\numCandidates})$ for the
sequence of candidates that~$\mapping$ places on the leaves from left
to right, so that~$\mapping_i=\mapping^{-1}(\ell_i)$. For a
subtree~$T'$ of~$T$ with leaf set~$L'$, let $\candidates'$ be the set
of candidates that $\mapping$ maps into~$L'$. The \emph{restricted
  mapping} $\restr{\mapping}{\candidates'}$ is the bijection between
$\candidates'$ and the leaves of~$T'$ obtained by
restricting~$\mapping$ to~$\candidates'$; read as a sequence, it is
the subsequence of $(\mapping_1,\ldots,\mapping_{\numCandidates})$
retaining exactly the candidates in $\candidates'$.

A vote $v$ is consistent with tree $T$ if it is possible
to permute the internal nodes' children---as allowed by their parents'
types---so that reading the leaves from left to right and mapping them
according to $\mapping$ gives exactly this vote.
The \emph{domain represented by~$T$
  and~$\mapping$}, denoted
$\domain(T,\mapping)$, is the set of all votes consistent with~$T$.
A $\typeQ$-tree is a $\typePQF$-tree whose all nodes are of type
$\typeQ$. 

Let $\numCandidates$ be the number of candidates. We say that $\domain(T,\mapping)$~is:
\newcommand{\myemph}[1]{{\bf\em #1}}
\begin{itemize}
\item \myemph{group-separable (GS)} if $T$ is a binary $\typeQ$-tree;
  this notion is due to
  \citet{ina:j:group-separable,ina:j:simple-majority} but we use the definition of \citet{kar:j:group-separable}, see also  \citep{kra-elk:c:explaining-prefs};
\item \myemph{balanced GS ($\boldsymbol\GSbal$)} if $T$ is a perfect binary $\typeQ$-tree, i.e.,
  each leaf is at the same distance from the root;\footnote{Some
    authors allow
    leaves' distances from the root to differ by up to $1$; this does
    not affect our experiments as we focus on 8 candidates, but
    simplifies presentation of our complexity results.}
\item \myemph{caterpillar GS ($\boldsymbol\GScat$)} if $T$ is a binary caterpillar
  $\typeQ$-tree, i.e., every internal node has at least one leaf child;
\item \myemph{identity} if $T$ has a single internal node of
  type~$\typeF$;
\item \myemph{antagonism} if $T$ has a single internal node of
  type~$\typeQ$;
\item \myemph{separation} if $\numCandidates$ is even and $T$'s root is
  of type~$\typeQ$ with exactly two children, each of type~$\typeP$
  and with $\numCandidates/2$ leaves.
\item \myemph{$\boldsymbol{k}$-stratification}, for $k\in[\numCandidates]$ with
  $k\mid\numCandidates$, if $T$'s root is of type~$\typeF$ and has $k$
  type-$\typeP$ children, each with $\numCandidates/k$~leaves.
\end{itemize}
Note that the identity domain contains a single preference order and the antagonism one contains two (which are reversals of each other).
We focus on group-separable domains, including its balanced and
caterpillar variants. We consider the identity, antagonism, and
$2$-stratification ones because they are naturally featured on maps of
elections~\citep{szu-boe-bre-fal-nie-sko-sli-tal:j:map}. Finally, we
consider $k$-stratification domains for larger values of $k$ (in
particular, $k = \nicefrac{m}{2}$), as well as the separation domain,
because they can be viewed as simplified variants of the balanced GS
domain (intuitively, $\Strat[\nicefrac{m}{2}]$ conflates all but the
final two levels of the balanced tree, whereas $\Separ$ conflates all but the first two).
Additionally, we are also interested in single-crossing domains.
\begin{definition}[\citet{mir:j:single-crossing,rob:j:tax}]
  A domain %
  over candidate set $\candidates$ is \emph{single-crossing}
  if there is an ordering of its votes so that for each two candidates
  $a, b \in \candidates$, votes that prefer $a$ to $b$ either form a
  prefix or a suffix of this ordering.
\end{definition}

We let $\GS$, %
$\GSbal$, $\GScat$, $\Ident$, $\Antag$,  $\Strat$, $\Separ$,  $\SPeak$,
and $\SCros$ denote families of  domains of given types.
$\calD$ is a Condorcet domain if %
every $\calD$ election with an odd number of
voters has a transitive majority relation. All our domains are Condorcet, except for 
$\Strat[k]$ with $k \neq \nicefrac{m}{2}$ and $\Separ$.

\paragraph{Election Distance From a Domain.}
Let $u$ and $v$ be two votes. Their \emph{swap distance},
denoted $\dist(u,v)$,
is %
the minimum number of swaps of adjacent candidates needed to transform
one into the other (equivalently, it is the number of inversions
between them).
Let $E=(\candidates,\voters)$ be an election and let $\domain$ be a
domain over~$\candidates$. Their distance is:
\[
    \dist(E,\domain)
    =
    \textstyle \sum_{v\in\voters}\, \min_{{u}\in\domain} \dist(u,v),
\]
In other words, it is the number of swaps %
needed to ensure that all votes from $E$ belong to~$\calD$.  The
distance of an election $E$ to a family of domains $\calX$, denoted
$\dist(E,\mathcal{X})$, is its smallest distance to a
member~of~$\calX$.  Computing this value for various domains is the
central problem of this paper.

\begin{center}%
    \begin{tabularx}{\linewidth}{@{\hspace{0.2em}}p{1.3cm}X}
        \toprule
        \multicolumn{2}{c}{$\swapNE{\mathcal{X}}$}\\
        \midrule
        \textbf{Given:}
        & An election $E=(\candidates,\voters)$ and a budget $\budget\in\N_0$.\\[4pt]
        \textbf{Question:}
        & Is $\dist(E,\mathcal{X})\leq \budget$?
        \\
        \bottomrule
    \end{tabularx}
\end{center}
\noindent{}For families of domains defined via $\typePQF$-trees, the
problem asks whether there is a tree~$T$ and a mapping~$\mapping$ such
that $\domain(T,\mapping) \in \mathcal{X}$ and
$\dist(E,\domain(T,\mapping)) \leq \budget$; for $\mathcal{X}=\SPeak$,
the analogous reformulation quantifies over axes rather than trees.
In case of $\SCros$, such a reinterpretation is less useful.

For the $\Ident$ domain, our problem asks for the smallest number of swaps
of adjacent candidates that ensures that all votes are equal. This is
exactly the well-known \KemenyOA problem, i.e., the problem of computing a ranking that minimizes the total swap distance (Kendall-tau distance) to the input rankings~\cite{kem:j:no-numbers}.

\begin{center}
	\begin{tabularx}{\linewidth}{@{}p{1.2cm}X}
		\toprule
		\multicolumn{2}{c}{\KemenyOA}\\
		\midrule
		\textbf{Given:}
			& An election $(\candidates,\voters)$ and a budget $\budget\in\N_0$.\\[2pt]
		\textbf{Question:}
			& Does there exist a vote $u^*$ over~$\candidates$ such that $\sum_{v\in\voters}\dist(v,u^*)\leq \budget$?\\
		\bottomrule
	\end{tabularx}
\end{center}

\citet{bar-tov-tri:j:manipulating} showed that the \KemenyOA problem is NP-hard, while~\citet{hem-spa-vog:j:kemeny} established its exact complexity. This intractability has motivated approximation algorithms~\cite{ail-cha-new:j:kemeny-approx}, parameterized algorithms~\cite{bet-dor:j:possible-winner-dichotomy, bet-bre-nie:j:kemeny}, and heuristic approaches~\cite{con-dav-kal:c:kemeny}. Moreover, even verifying whether a given ranking is Kemeny-optimal is \coNPh~\cite{fit-hem:c:kemeny-consensus-complexity}.

\toappendix{
\section{Additional Notation}
Here we provide additional notation that we subsequently use in our proofs. 

When defining votes in an election, we may write a vote~$v=(c_1,\dots,c_m)$ in the form
\[v:c_1,\dots,c_m.\]
Such preference lists can also include sets; the interpretation of this is that the candidates in these sets are ordered in an arbitrarily fixed way.
Often, we implicitly assume some arbitrary, fixed ordering~$\succ$ over the candidate set; then, for a set~$S$ of candidates, we denote by $\ora{S}$ the sequence of the candidates in~$S$ ordered according to~$\succ$, and we denote by~$\ola{S}$ the reverse of this sequence. 

We say that an election $E=(C,V)$ is \emph{caterpillar-consistent} if there exists a caterpillar tree $T$ and a mapping~$\mapping$ such that every vote $v \in V$ belongs to $\domain(T, \mapping)$.
If the number~$\numCandidates$ of candidates is clear from the context, then we will denote by~$\Tcat$ the caterpillar tree that has $\numCandidates$ leaves, each a left child.

Some of our algorithms will use a branching method where we \emph{guess} certain values or some property of the solution; the interpretation of such a guessing step is that the algorithm  exhaustively tries all possibilities.

Finally, we also adopt the conventions $a+\infty=\infty$ and ${\min\emptyset=\infty}$.

}

\section{Hardness of Finding Distances to Domains}
\label{sec:classiccomplexity}
\appendixsection{sec:classiccomplexity}
\toappendix{In this section, we present in detail our \NPcness proofs for the problems at hand. 

We start with the argument showing containment in \NP.}

Our goal in this section is to establish the hardness of
$\swapNE{\calX}$ for our domains. Note that it %
is in $\NP$ for all of the above-defined domains (see 
\iflong
\Cref{thm:NPmembership}).
\fi

\toappendix{
We observe that given a \typePQF-tree
$T$ with a mapping $\mapping$ from a candidate set~$C$ to the leaves of~$T$, we can verify in polynomial
time if a given vote belongs to the domain~$\domain(T,\mapping)$.

\begin{proposition}
  \label{prop:check-compatibility}
  There is a polynomial-time algorithm that given a vote $v$ and a domain
  $\domain(T,\mapping)$, defined via a \typePQF-tree~$T$ and a mapping $\varphi$,
  decides whether $v$ belongs to~$\domain(T,\mapping)$.  
\end{proposition}
\begin{proof}
    Let $T_1,\ldots,T_k$ be the subtrees rooted at the children of the root
    of~$T$. Given $T$ and~$\mapping$, we can compute in polynomial time the
    sets $\candidates_1,\ldots,\candidates_k$ such that, for each $i\in[k]$,
    the set $\candidates_i$ consists exactly of those candidates that are
    mapped to leaves of~$T_i$.

    If vote~$v$ belongs to $\domain(T,\mapping)$, every set $\candidates_i$ must occupy a contiguous block of vote~$v$. Hence, we first check that each~$\candidates_i$ is contiguous
    in $v$; if some $\candidates_i$ is not, then
    ${v}\notin\domain(T,\mapping)$. The contiguous blocks then appear in
    $v$ in a unique order, which we accept if and only if it is permitted
    by the type of the root: any order for type~$\typeP$, the left-to-right
    order or its reversal for type~$\typeQ$, and exactly the left-to-right
    order for type~$\typeF$. 

    Upon acceptance, for each $i\in[k]$ we proceed recursively on the subtree~$T_i$ with the restricted vote $\restr{v}{\candidates_i}$ and the restricted mapping~$\restr{\mapping}{\candidates_i}$. Vote $v$ belongs to $\domain(T,\mapping)$ if and only if all these checks succeed. Each node of $T$ is processed once, with polynomial work per node, and $T$ has $\Oh{\numCandidates}$ nodes; hence, the verification runs in polynomial time.
\end{proof}

This immediately implies that $\swapNE{\calX}$ is in $\NP$ for all our tree-defined domains. The same holds easily for the $\SPeak$ and $\SCros$ domains; hence, in the following proofs we omit $\NP$-membership arguments.

\begin{proposition}
\label{thm:NPmembership}
\swapNE{\mathcal{X}} is in~\NP for any
$\mathcal{X} \in \{\GS,\allowbreak\GSbal,\allowbreak\GScat,\allowbreak\Ident,\allowbreak\Antag,\allowbreak\Strat,\allowbreak\Separ\}$.
\end{proposition}
\begin{proof}
    Let $E = (C, V)$ be the input election of the considered problem.
    The certificate is a tree~$T$, a mapping~$\mapping$, and an election $E' = (C, V')$, with the same number of voters as the input election $E$. All of those objects are of size $\poly(\numCandidates,\numVoters)$. We check that $T$ has the form required by~$\mathcal{X}$ by inspecting its shape and node types in $\Oh{\numCandidates}$ time; if so, then $\domain(T, \mapping)\in\mathcal{X}$ by definition. By \Cref{prop:check-compatibility}, we verify that each vote $v'$ in $V'$ satisfies ${v'}\in\domain(T,\mapping)$, and finally that $\sum_{i\in[n]}\dist(v_i,v_i')\leq\budget$ where $V=(v_1,\dots,v_n)$ and~$V'=(v'_1,\dots,v'_n)$.
\end{proof}

}

\paragraph{Intractability for GS/cat and GS/bal.}
We first show that $\swapNE{\GScat}$ is $\NP$-complete, via a
reduction from \KemenyOA.  Our proof relies on 
the observation that if two votes $u$ and $v$ belong to the same
$\GScat$ domain and rank the same two candidates at the bottom, but in
opposite order, then they must rank all the preceding candidates
identically. Together with another technical trick, this
ensures that finding the closest $\GScat$ election
boils down to making (large parts of) all the votes identical,
as in \KemenyOA.
As our reduction preserves the number of voters, we inherit 
hardness for the case of four voters.
\begin{restatable}[\linkproof{thm:NE:swap:GScat:NPc}]{theorem}{thmNPGSCat}
	\label{thm:NE:swap:GScat:NPc}
	$\swapNE{\GScat}$ is \NPc, even for elections with $4$ voters. %
\end{restatable}

\prooftoappendixdivided{thm:NE:swap:GScat:NPc}{\thmNPGSCat*}{

To prove \Cref{thm:NE:swap:GScat:NPc}, we use the fact that the $\GScat$ domain is hereditary, i.e., caterpillar-consistency is closed under candidate deletion~\cite{kra-elk:c:explaining-prefs}.
We also use the following auxiliary claim.

\begin{lemma}
	\label{lem:GS-cat_fixed}
	Let $\succ_1,\succ_2$ be votes over~$\candidates$ whose two bottom positions are occupied by the same pair $\{x,y\}$ in opposite orders. %
    Then $(\succ_1,\succ_2)$ is caterpillar-consistent if and only if $\succ_1$ and $\succ_2$ agree in every position except the last two.
\end{lemma}

\begin{proof}
	For the forward direction, suppose that $(\succ_1,\succ_2)$ is caterpillar-consistent w.r.t.\ a mapping $\mapping=(\mapping_1,\dots,\mapping_\numCandidates)$. We show by induction on $|\candidates|$ that $\succ_1$ and $\succ_2$ agree in every position but the last two. For $|\candidates|=2$ there is nothing to prove, so let $|\candidates|>2$. By the recursive characterization, $\mapping_1$ is first or last in both votes. As both $x$ and~$y$ are second-to-last in exactly one of the votes $\succ_1$ and $\succ_2$, neither is at an extreme in both votes, so $\mapping_1\notin\{x,y\}$. Since the last position of each vote is occupied by $x$ or $y$, the candidate $\mapping_1$ is not last in either, hence is first in both. Deleting $\mapping_1$ from the top of both votes leaves two votes that still have $x,y$ in their bottom two positions in opposite orders; by the induction hypothesis they agree in every position but their last two. Together with the shared first candidate $\mapping_1$, the votes $\succ_1$ and $\succ_2$ agree in every position but the last two.

	For the converse, let $(p_1,\ldots,p_{\numCandidates-2})$ be the common prefix of $\succ_1$ and~$\succ_2$, and set $\mapping=(p_1,\ldots,p_{\numCandidates-2},x,y)$. After deleting $p_1,\ldots,p_{i-1}$, the candidate $p_i$ is first in both remaining votes, so the recursive condition holds for the first $\numCandidates-2$ steps. The remaining votes are then $(x,y)$ and $(y,x)$, both caterpillar with respect to $(x,y)$. Hence, both votes lie in~$\domain(\Tcat,\mapping)$, so $(\succ_1,\succ_2)$ is caterpillar-consistent.
\end{proof}

		Now, we present an \NPhness reduction from \KemenyOA. Given a \KemenyOA instance with candidates $\candidates=\{c_1,\ldots,c_{\numCandidates}\}$, voters ${\voters}=(\vote_1,\ldots,\vote_{\numVoters})$, and budget~$\budget$, we construct an election $E' = (\candidates',\voters')$ with
	\[
		\candidates' = \candidates \cup\{f_1,\ldots,f_{\budget+1}\} \cup\{g_1,\ldots,g_{\budget+2}\}.
	\]
    Next, we define the votes $\voters'$ as follows.
	\begin{itemize}
		\item %
        voter~$v_1'$ with vote 
        \[v'_1 =(v_1, f_1, \dots, f_{\budget+1}, g_1,\dots,g_{\budget+2});\]
		\item voter $\vote_i'$ for $i\in \{2,\dots,\numVoters\}$, with vote
        \[v'_i =(v_i, f_1, \dots, f_{\budget+1}, g_{\budget+2},\dots,g_1).\]
	\end{itemize}
	We claim $\dist(E',\GScat)\leq \budget$ if and only if the \KemenyOA instance is a yes-instance.

	\proofsubparagraph{Direction ``$\Rightarrow$''.}
	Suppose $\dist(E',\GScat)\leq \budget$. Then 
    there exists a  caterpillar-consistent election $E''=(\candidates', \voters'')$ whose votes $\voters''=(\vote_1'',\ldots,\vote_{\numVoters}'')$ are obtained from $\voters'$ by at most $\budget$ swaps in total. Because at most $\budget$ swaps were used, we know that
	\begin{enumerate}
		\item $c\succ_{v_i''} g_1$ and $c\succ_{v_i''} g_{\budget+2}$ for each $i \in [\numVoters]$ and $c\in\candidates$;
        
		\item $g_{\budget+2}\succ_{v_i''} g_1$ for each $i\in[2,\numVoters]$;
		\item $g_1\succ_{v_1''} g_{\budget+2}$.
	\end{enumerate}
	Indeed, violating~(1) would require moving $g_1$ or $g_{\budget+2}$ across the block $f_1,\ldots,f_{\budget+1}$, and violating~(2) or~(3) would require reversing two candidates with $\budget$ candidates between them; each needs more than $\budget$ swaps.

	Since $\GScat$ is hereditary, restricting $E''$ to ${\candidates''=\candidates\cup\{g_1,g_{\budget+2}\}}$ leaves a caterpillar-consistent election. By properties (1)--(3) above, in every vote of $\restr{\voters''}{\candidates''}$ the candidates $g_1$ and $g_{\budget+2}$ occupy the bottom two positions, ordered $g_1,g_{\budget+2}$ in $v_1''$ and $g_{\budget+2},g_1$ in every other vote. Applying \Cref{lem:GS-cat_fixed} to $\vote_1''$ and each $\vote_i''$ with $i\geq2$, all votes of~$\restr{\voters''}{\candidates''}$ agree on their top $\numCandidates$ positions, which are exactly the candidates of $\candidates$; let $\vote^*$ be this common order over~$\candidates$.

	Since restricting the candidate set $\candidates$ does not increase swap distance, and votes $\vote_i'$ and $\vote_i''$ restricted to~$\candidates$ yield votes $\vote_i$ and $\vote^*$, respectively, we get
	\[
		\sum_{i=1}^{\numVoters}\dist(\vote_i,\vote^*) \;\leq\; \sum_{i=1}^{\numVoters}\dist(\vote_i',\vote_i'') \;\leq\; \budget,
	\]
	so the \KemenyOA instance is a yes-instance.

	\proofsubparagraph{Direction ``$\Leftarrow$''.}
	Let $\vote^*=(q_1,\ldots,q_{\numCandidates})$ be an ordering over $\candidates$ with $\sum_{i}\dist(\vote_i,\vote^*)\leq \budget$. Take
	\[
		\mapping=(q_1,\ldots,q_{\numCandidates}, f_1,\ldots,f_{\budget+1},g_1,\ldots,g_{\budget+2}).
	\]
	Using a total of at most $\budget$ swaps among the candidates of~$\candidates$ in all votes of $\voters'$, we bring them into the order $\vote^*$ while leaving all other candidates in place. The resulting votes have one of two forms:
	\[
		(q_1, \dots,  q_{\numCandidates},  f_1, \dots,  f_{\budget+1},  g_1, \dots,  g_{\budget+2}); 
	\]
	\[
		(q_1, \dots,  q_{\numCandidates},  f_1, \dots,  f_{\budget+1},  g_{\budget+2}, \dots,  g_1). %
	\]
	The first type of vote ranks the candidates exactly in as~$\mapping$, so it lies in $\domain(\Tcat,\mapping)$. For the second form, its prefix through $f_{\budget+1}$ coincides with that of $\mapping$, so the recursive characterization removes $q_1,\ldots,q_{\numCandidates},f_1,\ldots,f_{\budget+1}$ from the top; the remaining vote $(g_{\budget+2},\dots, g_1)$ and remaining mapping $(g_1,\ldots,g_{\budget+2})$ are then matched by removing $g_1,\ldots,g_{\budget+2}$ from the bottom. Hence the second type of vote also lies in~$\domain(\Tcat,\mapping)$, and $\dist(E',\GScat)\leq \budget$.

	Since \KemenyOA is already \NPh{} for four voters~\cite{dwo-kum-nao-siv:c:rank-aggregation} and the reduction preserves the number of voters, the hardness holds for $\numVoters=4$.
}

For the case of $\GSbal$, we also obtain $\NP$-completeness, but in
this case the proof is far more involved, and follows by a reduction
from the \probName{Maximum Vertex Coverage} problem on bipartite
graphs, shown to be \NPh independently by
\citet{apo-sim:j:max-vertex-coverage-bipartite} and
\citet{jor-vet:j:rank-of-matroid}.

\begin{restatable}[\linkproof{thm:NE:swap:GSbal:NPc}]{theorem}{thmNEswapGSbalNPc}   
\label{thm:NE:swap:GSbal:NPc}
  \swapNE{\GSbal} is \NPc.
\end{restatable}
\prooftoappendix{thm:NE:swap:GSbal:NPc}{\thmNEswapGSbalNPc*}{
    Before presenting our reduction, we introduce some notation. 
    For an integer~$k \in \N$ that is a power of~$2$, let 
    $T_k$ denote the binary, balanced \typeQ-tree with~$k$ leaves.
    For a given tree~$T_k$ and a list~$\ora{\candidates}$ of candidates, we let $\mappinglist{\ora{\candidates}}$ denote the bijection from~$\candidates$ to the leaves of~$T_k$ for which the left-to-right traversing of the leaves yields the ordering~$\ora{\candidates}$ of the candidates.
    
    Let~$E_k=(\candidates_k,\voters_k)$ be an election with $\candidates_k=\{c_1,\ldots,c_k\}$
    and $\voters_k=(\vote_1,\ldots,\vote_k)$ where vote~$\vote_i$ is obtained from the  ordering $c_1,\ldots,c_k$ of candidates by moving $c_i$ to the top position, i.e., $\vote_i=(c_i,c_1,\ldots,c_{i-1},c_{i+1},\ldots,c_k)$. 
    We use the following property of such elections.

    \begin{claim}
    \label{clm:downshift-elections}
        Let $k$ be a power of 2, and define the function $g(k)=k(\nicefrac{k}{2}-1)$.
        Then $(E_k,\budget)$ is a yes-instance of \swapNE{\GSbal} if and only if $\budget \geq g(k)$.
        Moreover, $E_k$ has swap distance~$g(k)$ from an election in~$\domain(T_k,\mappinglist{c_1,\dots,c_k})$. 
    \end{claim}

    \begin{claimproof}
        We use induction on $\log k$. The claim is easily seen to hold for $k=2$ with $g(2)=0$, and it is also straightforward to check that for $k=4$ we obtain $g(4)=4$.\footnote{In fact, for any mapping~$\mapping$ from $\{c_i:i \in [4]\}$ to the leaves of the balanced binary tree~$T_4$, the minimal swap distance of $E_4$ from some election in $\domain(T_4,\mapping)$ is exactly~$4$.} Assume that it holds for $\nicefrac{k}{2}$.
        
        Suppose now that there is a  GS-balanced election $(\candidates_k,\voters'_k)$ with $\voters'_k=(\vote'_1,\ldots,\vote'_k)$
        and a  mapping~$\mapping$ from~$\candidates_k$ to the leaves of~$T_k$
        such that each $\vote'_i$ is in $ \domain(T_k,\mapping)$. Let us define $s^\star= \sum_{i \in [k]}\swap(\vote_i,\vote'_i)$. We aim to show $s^\star \geq g(k)$.
        
        Let $T^1$ and~$T^2$ be the two subtrees rooted at the two children of the root of~$T$.
        Let $S_i$ be the set of candidates mapped by~$\mapping$ to leaves in~$T^i$ for $i \in [2]$.
        Observe that we can obtain $\voters'_k$ with at most~$s^\star$ swaps by first performing only those swaps that are between candidates of~$S_1$ and~$S_2$, obtaining some election~$(\candidates_k,\voters''_k)$, and only then 
        performing the swaps between candidates belonging to the same set~$S_1$ or~$S_2$.
        Note that the ordering of the candidates in~$S_i$ for some $i \in [2]$ as appearing in the votes of~$\voters''_k$ is as follows: in~$\nicefrac{k}{2}$ votes from~$\voters''_k$, they appear in their ``default'' order (in increasing order of their subscripts), and the remaining $\nicefrac{k}{2}$ votes in~$\voters''_k$ can be obtained by taking each candidate~$c$ in~$S_i$ and moving it to the top of this default ordering of~$S_i$.
        Hence, these $\nicefrac{k}{2}$ votes form an election that is isomorphic to~$E_{\nicefrac{k}{2}}$. 
        By our induction hypothesis, this shows that the minimum number of swaps necessary to reach $\voters'_k$ from~$\voters''_k$ is at least twice the swap distance of~$E_{\nicefrac{k}{2}}$ from a GS-balanced election, i.e.,  $2g(\nicefrac{k}{2})$.
        
        We will now show that the swap distance of~$\voters$ from $\voters''$ is at least $\nicefrac{k^2}{4}$. Let $\voters''=(\vote''_1,\dots,\vote''_k)$ where $\vote''_i$ is the vote in~$V''$ with minimum swap distance to~$v_i$, for each $i \in [k]$.
        Consider any two votes~$\vote_a$ and~$\vote_b$ such that $c_a \in S_1$ and $c_b \in S_2$. We will show that $\swap(\vote_a,\vote''_a)+\swap(\vote_b,\vote''_b)\geq \nicefrac{k}{2}$.
        First, suppose that all candidates in~$S_1$ precede all candidates in~$S_2$ in both~$\vote''_a$ and~$\vote''_b$. 
        In this case, 
        candidate~$c_b$, ranked at the first position in~$\vote_b$, needs to be swapped with all candidates in~$S_1$ to reach its position in~$\vote''_b$, which means $\nicefrac{k}{2}$ swaps. 
        Similarly, if all candidates in~$S_2$ precede all candidates in~$S_1$ in both~$\vote''_a$ and~$\vote''_b$, then $c_a$ needs $\nicefrac{k}{2}$ swaps in~$\vote_a$ to reach its position in~$\vote''_a$. 
        Assume now that candidates in~$S_1$ precede those in~$S_2$ in exactly one of~$\vote''_a$ and~$\vote''_b$. Then every pair of candidates from $(S_1 \setminus  \{c_a\}) \times (S_2 \setminus \{c_b\})$ needs to be swapped either when turning $\vote_a$ into~$\vote''_a$, or when turning $\vote_b$ into~$\vote''_b$. This yields $\swap(\vote_a,\vote''_a)+\swap(\vote_b,\vote''_b)\geq (\nicefrac{k}{2}-1)^2>\nicefrac{k}{2}$ which holds for each $k \geq 8$. 

        Pairing up the votes in~$\voters$ so that for each pair, the candidates on the top positions are in different sets~$S_1$ or~$S_2$, we obtain that $\sum_{i} \swap (\vote_i,\vote''_i)\geq (\nicefrac{k}{2})^2$, as promised. It follows that $s^\star \geq 2g(\nicefrac{k}{2})+\nicefrac{k^2}{4}=\nicefrac{k^2}{4}-k+\nicefrac{k^2}{4}=k(\nicefrac{k}{2}-1)$, proving the first direction of the claim. 
        
        It remains to observe that we can obtain a GS-balanced election~$E'_k \in \domain(T_k,\mappinglist{c_1,\dots,c_k})$ from~$E_k$ as follows. %
        First, we move all candidates in~$\{c_1,\dots,c_{\nicefrac{k}{2}}\}$ to the first $\nicefrac{k}{2}$ positions in each vote; this requires $0$ swaps in each vote $\vote_i$, $i \in [\nicefrac{k}{2}]$, and exactly $\nicefrac{k}{2}$ swaps in all remaining votes. Hence, this yields $\nicefrac{k^2}{4}$ swaps. Note that the subelection restricted to~$C_{\nicefrac{k}{2}}=\{c_1,\dots,c_{\nicefrac{k}{2}}\}$ contains the vote $(c_1,\dots,c_{\nicefrac{k}{2}})$ exactly~$\nicefrac{k}{2}$ times, and additionally, a copy of $E_{\nicefrac{k}{2}}$. Thus, we can turn this subelection into a GS-balanced election using $g(\nicefrac{k}{2})$ swaps. Dealing with the subelection restricted to~$C_k \setminus C_{\nicefrac{k}{2}}$ in the same way, we obtain that the swap distance of $E_k$ from~$E'_k$ is exactly $2g(\nicefrac{k}{2})+\nicefrac{k^2}{4}=g(k)$. 
        This finishes our proof.
        \end{claimproof}

    \def\vertexbudget{k}
    \def\threshold{t}
    \def\graphsize{n}

    \proofsubparagraph{The Input.}
    We now present a reduction from a variant of the \probName{Balanced Maximum Vertex Coverage (BMVC)} problem on bipartite graphs. %
    The input of this problem is a bipartite graph~$G=(U,W;F)$ with two integers~$\vertexbudget$ and~$\threshold$, and the question is %
    whether $G$ contains $\vertexbudget$ vertices from~$U$ and $\vertexbudget$ vertices from~$W$ such that together they cover at least~$\threshold$ edges.
    \begin{claim}
    \label{clm:BMVC:NPc}
        BMVC is \NPc even if $\vertexbudget$ is a power of~$2$, and $|U|=|W|=\graphsize+\vertexbudget$ for some integer~$\graphsize> 2 \vertexbudget$ that is a power of~$2$.
    \end{claim}
    \begin{claimproof}
        Containment in \NP is trivial. We give a reduction from the \probName{Constraint Bipartite Vertex Cover (CBVC)} problem which, given a bipartite graph~$H$ with partite sets~$A$ and~$B$ together with two integers~$k_A$ and~$k_B$, asks whether $H$ contains a vertex cover of the form~$A' \cup B'$ such that $A' \subseteq A$ has size at most $k_A$ and $B' \subseteq B$ has size at most~$k_B$. 
        This problem is \NPc \cite{kuo-fuc:c:constrained-vertex-cover}.
        Let $\vertexbudget$ be power of~$2$ at least $\max\{k_A,k_B\}$. 
        We create a graph~$H'$ as follows: we first add~$(\vertexbudget-k_A)$ new vertices to~$A$ and $(\vertexbudget-k_B)$ new vertices to~$B$, and then for each of the newly added $2\vertexbudget-k_A-k_B$ vertices, we create $2\vertexbudget+1$ neighbors on the other side, each of them having degree~1.  Notice that the newly added vertices must be contained in every vertex cover of~$H'$ of size at most~$2\vertexbudget$. Thus, $(H,k_A,k_B)$ is a yes-instance of \probName{CBVC} if and only if 
        there are $\vertexbudget$ vertices on each side of~$H'$ such that together they cover all edges of~$H'$, or in other words,  
        if $(H',\vertexbudget,m_{H'})$ is a yes-instance of \probName{Balanced Maximum Vertex Coverage} where $m_{H'}$ is the number of edges in~$H'$. This shows the \NPhness of the BMVC problem with the desired requirement on~$\vertexbudget$.

        As for the requirement on~$|U|$ and~$|W|$, we can pick any integer~$\graphsize> 2 \vertexbudget$ that is a power of~$2$, and then ensure the condition        $|U|=|W|=\graphsize+\vertexbudget$ by adding the necessary number of isolated vertices to the input graph.        
    \end{claimproof}
    
    \proofsubparagraph{The Construction.}
    Given our input graph $G=(U,W;F)$ with $|U|=|W|=\graphsize+\vertexbudget$ and integers~$\vertexbudget$ and~$\threshold$ for BMVC with the properties as in~\Cref{clm:BMVC:NPc}, we construct an election $(\candidates,\voters)$ as follows.
    Let $\alpha=6|U|^4$. 
    The set of candidates consists of the vertices in~$U \cup W$ as well as two sets $B_1$ and~$B_2$ of dummies with $|B_1|=2\graphsize-\vertexbudget$ and $|B_2|=4\graphsize-\vertexbudget$.
    We create two sets of  \emph{anchor voters}, $A$ and~$A'$ with $|A|=|A'|=\alpha$. Their preferences are as follows: 
    \begin{align*}
        a &: \ora{B_1}, \ora{U}, \ora{W}, \ora{B_2} & \qquad \qquad \forall a \in A; \\
        a' &: \ora{B_1}, \ola{U}, \ola{W}, \ora{B_2} & \qquad \qquad \forall a' \in A'.
    \end{align*}
    We further create \emph{incidence voters} $\vote_{u,w}$ and~$\vote'_{u,w}$ for each $u \in U$ and $w \in W$; %
    their preferences are as below:
    \begin{align*}
        \vote_{u,w} &: \ora{B_1}, \ora{W \setminus \{w\}}, u, w, \ora{U \setminus \{u\}}, \ora{B_2} 
        & \text{ if $\{u,w\} \in F$}; \\
        \vote_{u,w} &: \ora{B_1}, \ora{W \setminus \{w\}}, w, u, \ora{U \setminus \{u\}}, \ora{B_2} 
        & \text{ if $\{u,w\} \notin F$}; \\
        \vote'_{u,w} &: \ora{B_1}, w, \ola{W \setminus \{w\}}, \ola{U \setminus \{u\}}, u, \ora{B_2}.
    \end{align*}

    We finish the construction by setting our budget as
    \begin{align*}
    \budget = \alpha \cdot 2 \graphsize  \vertexbudget 
    & + 2 
    \vertexbudget(3\graphsize+\vertexbudget)(\graphsize+\vertexbudget)^2 \\
    & + 4(\graphsize+\vertexbudget) (g(\graphsize)+g(\vertexbudget))
    + |F| - 2\threshold.
    \end{align*}

    We are going to show that $(G,\vertexbudget,\threshold)$ is a yes-instance of \probName{Balanced Maximum Vertex Coverage} if and only if the constructed instance $((\candidates,\voters),\budget)$ of \swapNE{\GSbal} is a yes-instance.
    
    \proofsubparagraph{Direction ``$\Rightarrow$''.}
    First assume that vertex sets $U^\star \subseteq U$ and $W^\star \subseteq W$ form a solution to our instance of \probName{Balanced Maximum Vertex Coverage}. Let $T_{8\graphsize}$ be the balanced binary tree with~$8\graphsize$ leaves, and consider the mapping~$\mapping^\star=\mappinglist{\ora{B_1}, \ora{U^\star},\ora{U \setminus U^\star}, \ora{W \setminus W^\star}, \ora{W^\star}, \ora{B_2}}$. To turn the constructed election~$(\candidates,\voters)$ into an election in~$\domain(T_{8\graphsize},\mapping^\star)$, we perform the following swaps:
    
    First, we move all candidates in~$U^\star$ to the $\vertexbudget$ positions right after the last dummy in~$B_1$
    and all candidates in~$W^\star$ to the $\vertexbudget$ positions right before the first dummy in~$B_2$
    in all votes.
    In a pair of anchor votes from $A \times A'$, this step requires exactly~$2\vertexbudget \cdot \graphsize$ swaps in total. 
    In a pair $\vote_{u,w}$, $\vote'_{u,w}$ of incidence  voters, 
    we need to swap 
    \begin{itemize}
        \item all pairs in 
        \[(U \times W) \setminus \left((U \setminus U^\star) \times (W \setminus W^\star)\right)
        \]
        in both votes,
        except for $(u,w)$ in~$\vote_{u,w}$ if $\{u,w\} \in F$;
        \item all pairs in
        \[((U \setminus U^\star) \times U^\star) \cup ((W \setminus W^\star) \times W^\star)\] exactly once in total in the two votes.  
    \end{itemize} 
    This requires
    $2(3\graphsize+\vertexbudget)\vertexbudget$ swaps in these two votes, unless (i) $u \in U^\star$ or $w \in W^\star$, and (ii) $\{u,w\} \in F$, in which case we need one swap less. Let $F^\star$ denote the set of edges in~$G$ incident to some vertex in~$U^\star \cup W^\star$.
    Then in this step, we perform  
    $2\vertexbudget (3\graphsize+\vertexbudget) (\graphsize+\vertexbudget)^2-|F^\star|$ swaps in all incidence votes.
    Taking into account the swaps in anchor votes, we obtain a total of 
    \begin{equation}
    \label{eqn:swaps-in-first-step}
    \alpha \cdot 2kn + 2\vertexbudget (3\graphsize+\vertexbudget) (\graphsize+\vertexbudget)^2-|F^\star|
    \end{equation}
    swaps in this step.
    
    Second, we also swap $u$ with~$w$ in each incidence vote~$\vote_{u,w}$ where $u \in U \setminus U^\star$, $w \in W \setminus W^\star$, and $\{u,w\} \in F$. 
    This means 
    \begin{equation}
    \label{eqn:swaps-in-second-step}
    |F \setminus F^\star|
    \end{equation}
    swaps. 
    Let $E'$ denote the election obtained at this point. Notice that for each anchor voter~$a$ in~$E'$ we know 
    \[B_1 \succ_{a} U^\star \succ_{a} U \setminus U^\star \succ_{a} W \setminus W^\star \succ_{a}  W^\star \succ_{a} B_2,\] whereas for each incidence voter~$\vote$ we know 
    \[B_1 \succ_{\vote} U^\star \succ_{\vote} W \setminus W^\star \succ_{\vote} U \setminus U^\star \succ_{\vote}  W^\star \succ_{\vote} B_2.\] 
    Recall also that %
    each of the sizes $|U^\star|=|W^\star|$, $|B_1 \cup U^\star|$, $|U \setminus U^\star|=|W \setminus W^\star|$, and $|B_2 \cup W^\star|$ are 
    powers of~$2$.

\begin{figure}
    \centering

\begin{tikzpicture}[
  every node/.style={font=\small},
  intnode/.style={circle,fill=black,inner sep=0pt,minimum size=3pt},
  treeline/.style={black,thin}
]

\def\u{0.47cm} %
\def\wBone{3.2}     %
\def\wUst{1.2}    %
\def\wUmU{2.2}      %
\def\wWmW{2.2}      %
\def\wWst{1.2}    %
\def\wBtwo{7.6}   %

\pgfmathsetmacro{\total}{\wBone+\wUst+\wUmU+\wWmW+\wWst+\wBtwo}

\pgfmathsetmacro{\xa}{0}
\pgfmathsetmacro{\xb}{\xa+\wBone}
\pgfmathsetmacro{\xc}{\xb+\wUst}
\pgfmathsetmacro{\xd}{\xc+\wUmU}
\pgfmathsetmacro{\xe}{\xd+\wWmW}
\pgfmathsetmacro{\xf}{\xe+\wWst}
\pgfmathsetmacro{\xg}{\xf+\wBtwo}

\definecolor{cB1}{RGB}{225,225,225}
\definecolor{cUst}{RGB}{255,214,153}
\definecolor{cUmU}{RGB}{255,236,204}
\definecolor{cWmW}{RGB}{204,229,255}
\definecolor{cWst}{RGB}{153,204,255}
\definecolor{cB2}{RGB}{225,225,225}

\def\hbar{0.55cm}

\draw[fill=cB1]  ({\xa*\u},0) rectangle ({\xb*\u},\hbar);
\draw[fill=cUst] ({\xb*\u},0) rectangle ({\xc*\u},\hbar);
\draw[fill=cUmU] ({\xc*\u},0) rectangle ({\xd*\u},\hbar);
\draw[fill=cWmW] ({\xd*\u},0) rectangle ({\xe*\u},\hbar);
\draw[fill=cWst] ({\xe*\u},0) rectangle ({\xf*\u},\hbar);
\draw[fill=cB2]  ({\xf*\u},0) rectangle ({\xg*\u},\hbar);

\node at ({(\xa+\xb)/2*\u},\hbar/2) {$B_1$};
\node at ({(\xb+\xc)/2*\u},\hbar/2) {$U^\star$};
\node at ({(\xc+\xd)/2*\u},\hbar/2) {$U \mysetminus U^\star$};
\node at ({(\xd+\xe)/2*\u},\hbar/2) {$W\mysetminus W^\star$};
\node at ({(\xe+\xf)/2*\u},\hbar/2) {$W^\star$};
\node at ({(\xf+\xg)/2*\u},\hbar/2) {$B_2$};

\def\braceY{-0.15cm}
\def\gap{0.06cm}

\draw[decorate,decoration={brace,amplitude=4pt,mirror}]
  ({\xa*\u+\gap},\braceY) -- ({\xb*\u-\gap},\braceY)
  node[midway,yshift=-0.55cm] {$2n-k$};
\draw[decorate,decoration={brace,amplitude=4pt,mirror}]
  ({\xb*\u+\gap},\braceY) -- ({\xc*\u-\gap},\braceY)
  node[midway,yshift=-0.55cm] {$k$};
\draw[decorate,decoration={brace,amplitude=4pt,mirror}]
  ({\xc*\u+\gap},\braceY) -- ({\xd*\u-\gap},\braceY)
  node[midway,yshift=-0.55cm] {$n$};
\draw[decorate,decoration={brace,amplitude=4pt,mirror}]
  ({\xd*\u+\gap},\braceY) -- ({\xe*\u-\gap},\braceY)
  node[midway,yshift=-0.55cm] {$n$};
\draw[decorate,decoration={brace,amplitude=4pt,mirror}]
  ({\xe*\u+\gap},\braceY) -- ({\xf*\u-\gap},\braceY)
  node[midway,yshift=-0.55cm] {$k$};
\draw[decorate,decoration={brace,amplitude=4pt,mirror}]
  ({\xf*\u+\gap},\braceY) -- ({\xg*\u-\gap},\braceY)
  node[midway,yshift=-0.55cm] {$4n-k$};

\coordinate (baseA) at ({\xa*\u},\hbar);
\coordinate (baseB) at ({\xb*\u},\hbar);
\coordinate (baseC) at ({\xc*\u},\hbar);
\coordinate (baseD) at ({\xd*\u},\hbar);
\coordinate (baseE) at ({\xe*\u},\hbar);
\coordinate (baseF) at ({\xf*\u},\hbar);
\coordinate (baseG) at ({\xg*\u},\hbar);

\coordinate (root)  at ({(\xa+\xg)/2*\u},{\hbar+4.45cm});
\coordinate (nodeL) at ({(\xc)*\u},{\hbar+3.1cm});
\coordinate (apexUW) at ({(\xd)*\u},{\hbar+2.0cm});

\coordinate (apexBU)  at ({(\xa+\xc)/2*\u},{\hbar+2.0cm});   %
\coordinate (apexUmU) at ({(\xc+\xd)/2*\u},{\hbar+1.0cm});   %
\coordinate (apexWmW) at ({(\xd+\xe)/2*\u},{\hbar+1.0cm});   %
\coordinate (apexB2W) at ({(\xe+\xg)/2*\u},{\hbar+3.1cm});   %

\coordinate (apexU) at ($(baseC)!{(1.0)/(4)}!(apexBU)$);
\coordinate (apexW) at ($(baseE)!{(1.0)/(6.5)}!(apexB2W)$);

\draw[treeline] (root)  -- (nodeL);
\draw[treeline] (apexUmU)  -- (apexUW);
\draw[treeline] (nodeL) -- (apexBU);
\draw[treeline] (nodeL) -- (apexUW);
\draw[treeline] (apexUW) -- (apexWmW);
\draw[treeline] (root) -- (apexB2W);

\draw[thick,black,fill=black!8,fill opacity=0.7]
  (apexBU) -- (baseA) -- (baseC) -- cycle;
\draw[thick,black,fill=cUst,fill opacity=0.9]
  (apexU) -- (baseB) -- (baseC) -- cycle;

\draw[thick,black,fill=cUmU,fill opacity=0.85]
  (apexUmU) -- (baseC) -- (baseD) -- cycle;

\draw[thick,black,fill=cWmW,fill opacity=0.85]
  (apexWmW) -- (baseD) -- (baseE) -- cycle;

\draw[thick,black,fill=black!8,fill opacity=0.7]
  (apexB2W) -- (baseE) -- (baseG) -- cycle;
\draw[thick,black,fill=cWst,fill opacity=0.9]
  (apexW) -- (baseE) -- (baseF) -- cycle;

\node[above=1pt,font=\scriptsize] at (nodeL) {$t_1$};
\node[above=1pt,font=\scriptsize] at (apexB2W) {$t_2$};
\node[above left,font=\scriptsize] at (apexBU) {$t_{11}$};
\node[above right,font=\scriptsize] at (apexUW) {$t_{12}$};
\node[above left, xshift=2pt,font=\scriptsize] at (apexUmU) {$t_{121}$};
\node[above right,font=\scriptsize] at (apexWmW) {$t_{122}$};

\foreach \p in {root,nodeL,apexUW,apexBU,apexUmU,apexWmW,apexB2W,apexU,apexW}{
  \node[intnode] at (\p) {};
}

\end{tikzpicture}
    \caption{Illustration for the proof of \Cref{thm:NE:swap:GSbal:NPc}. The figure depicts the prefect binary tree~$T_{8\graphsize}$. 
    The candidate set displayed in the rectangle at the bottom of a given subtree contains those candidates that are mapped by~$\mapping^\star$ to the leaves of the given subtree, with the size of this set written underneath.}
    \label{fig:GSbal-tree}
\end{figure}

    Third, let us now perform all swaps necessary to turn election~$E'$ into an election in the domain~$\domain(T_{8\graphsize},\mapping^\star)$.
    Let us now calculate the swaps necessary for this. 

    \begin{description}
        \item[Swaps involving candidates in~$U^\star$.] 
        Notice that $E'$ restricted to candidates in~$U^\star$ contains
        \begin{itemize}
            \item  $|W|$ copies of an election~$E'_{U^\star}$ that is isomorphic to~$E_{\vertexbudget}$ with the default ordering of the candidates being~$\ora{U^\star}$,
            \item $|W|$ copies of the reverse of~$E'_{U^\star}$ (obtained from~$E'_{U^\star}$ by reversing all votes), and
            \item  additional votes of the form $\ora{U^\star}$ or $\ola{U^\star}$.
        \end{itemize}  
        By \Cref{clm:downshift-elections}, we know that we need $g(\vertexbudget)$ swaps to turn one copy of~$E'_{U^\star}$ into a GS-balanced election in~$\domain(T_{\vertexbudget},\mappinglist{\ora{U^\star}})$.
        Since reversing all votes does not change the distance of an election from the GS-balanced domain, and moreover, $\domain(T_{\vertexbudget},\mappinglist{\ora{U^\star}})=\domain(T_{\vertexbudget},\mappinglist{\ola{U^\star}})$, we get that the distance of the restriction of~$E'$ to~$U^\star$ from the GS-balanced domain is exactly $2(\graphsize+\vertexbudget)\cdot g(\vertexbudget)$. 
        \item[Swaps involving candidates in~$U \setminus U^\star$.] 
        Using the same arguments as in the previous case, we get that the restriction of~$E'$ to candidates in~$U \setminus U^\star$ has swap distance $2(\graphsize+\vertexbudget)\cdot g(\graphsize)$ from  $\domain(T_{\graphsize},\mappinglist{\ora{U \setminus U^\star}})$.
        \item[Swaps involving candidates of $W$.]
        By symmetry, the total number of swaps necessary to turn the restriction of~$E'$ to candidates in~$W^\star$, as well the restriction of~$E'$ to candidates in~$W \setminus W^\star$ into a GS-balanced election in~$\domain(T_{\vertexbudget},\mappinglist{\ora{W^\star}})$ and~$\domain(T_{\graphsize},\mappinglist{\ora{W \setminus W^\star}})$, respectively, is 
        $2(\graphsize+\vertexbudget)(g(\vertexbudget)+g(\graphsize))$. 
    \end{description}
    These swaps are sufficient, because they turn the election into one in~$\domain(T_{8\graphsize},\mapping^\star)$. Indeed, notice that $T_{8\graphsize}$ contains a \typeQ-node whose two children are the roots of the subtrees, each of size~$\graphsize$, that contain exactly those leaves to which candidates in~$U \setminus U^\star$ and in~$W \setminus W^\star$, respectively, are mapped---enabling the reversed ordering of these two subtrees; see \Cref{fig:GSbal-tree} for an illustration.
    
    Altogether, the total number of swaps performed in this step is 
    \begin{equation}
        \label{eqn:swaps-in-third-step}
        4(\graphsize+\vertexbudget)(g(\vertexbudget)+g(\graphsize)).
    \end{equation}
    
    Summing up \eqref{eqn:swaps-in-first-step}--\eqref{eqn:swaps-in-third-step}, we obtain that the total number of swaps performed is
    $\budget+2\threshold-2|F^\star|$. Since $U^\star \cup W^\star$ is a solution to our input instance, the set~$F^\star$ of edges covered by the vertices in~$U^\star \cup W^\star$ must satisfy $|F^\star|\geq \threshold$. Thus, we get that the swap-distance of~$(\candidates,\voters)$ from  the GS-balanced domain is at most~$\budget+2\threshold-2|F^\star|\leq \budget$, which proves the first direction of our claim.

    \proofsubparagraph{Direction ``$\Leftarrow$''.}
    Suppose now that $(\candidates,\voters)$ has swap distance at most~$\budget$ from some GS-balanced election~$(\candidates,\voters')$.
    For each vote~$\vote$ in~$\voters$, let $\pi(\vote)$ denote the corresponding vote in~$\voters'$ that belongs to the domain~$\domain(T_{8\graphsize},\mapping)$ for some~$\mapping$.
    
    Since all votes in~$A$ are the same, we can assume that $\pi(a)$ is the same vote for each $a \in A$; the analogous claim holds for votes in~$A'$. 
    Consider now an anchor vote~$a \in A$, and let $C_a^1$ and $C_a^2$ denote the candidates in the first~$2\graphsize$ and the last~$4\graphsize$ positions in~$\pi(a)$. Consider now a pair of anchor votes~$(a,a') \in A \times A'$.
    By the definition of the allowed permutations in a \typeQ-tree, we know that either (i) $C_a^i=C_{a'}^i$ or (ii) $C_a^i \cap C_{a'}^i=\emptyset$ for both~$i \in [2]$. 
    We claim that only the former case can happen, for both $i \in [2]$. 
    To see this, notice that in case (ii), 
    each dummy $b \in B_i$ has to be moved away from the outermost $|B_i|+\vertexbudget$ positions (i.e., the first $2\graphsize$ positions for $i=1$ and the last $4\graphsize$ positions for $i=2$) either in~$a$ or in~$a'$, which yields $2\graphsize \cdot |B_i|>2\graphsize$ swaps in total for~$(a,a')$. Hence, this would yield at least $\alpha \cdot 2\graphsize \cdot (2\graphsize -\vertexbudget)\geq 
    \alpha \cdot (2\graphsize\vertexbudget + \graphsize)>\budget$ where the last inequality holds due to  %
    our choice of~$\alpha$.

    Thus, we have that (i) holds for both $i \in [2]$. Then candidates in $S_i=C_a^i \setminus B_i$ need to be moved into the outermost~$|B_i|+\vertexbudget$ positions; by the ordering of candidates in~$a$ and in~$a'$, the number of swaps necessary for this can be calculated as follows:
    \begin{itemize}
        \item Moving a candidate~$u \in U \cap S_1$ to within the first $2\graphsize$ positions requires $\graphsize$ swaps within a pair of votes in~$(a,a') \in A \times A'$, since $u$ needs to be swapped with every candidate in~$U \setminus S_1$ exactly once. 
        \item Moving a candidate~$c$ in~$S_1 \setminus U$ to the first~$\vertexbudget$ position following candidates in~$B_1$ requires at least $2\graphsize$ swaps within~$a$ and~$a'$, because~$c$ must be swapped with all candidates in~$U \cup B_1 \setminus C^a_1$ in both votes~$a$ and~$a'$.
        \item Similarly, moving a candidate~$c$ to within the last~$4\graphsize$ position requires $\graphsize$ swaps if $c \in S_2 \cap W$ and $2\graphsize$ swaps if $c \in S_2 \setminus W$.
    \end{itemize}
    This implies that $S_1 \subseteq U$ with $|S_1| \leq \vertexbudget$ and $S_2 \subseteq W$ with $|S_2| \leq \vertexbudget$, because otherwise we would have to perform at least $\alpha(2\graphsize \vertexbudget+\graphsize)>\budget$ swaps in the anchor votes. 
    In particular, it follows that there exist subsets $U^\star\subseteq U$ and $W^\star\subseteq W$ of candidates with $|U^\star|=|W^\star|=\vertexbudget$ such that $C_a^1=B_1 \cup U^\star$ and $C_a^2=B_2 \cup W^\star$ for all anchor votes $a \in A \cup A'$.
    The number of swaps necessary to move these candidates to their positions in~$\pi(a)$ in each anchor vote~$a \in A \cup A'$ is $\alpha \cdot 2\graphsize \vertexbudget$ in total.
    
    Let us consider the number of swaps necessary to move these candidates to their positions in the incidence votes. 
    It is easy to see that moving $C_a^1$ to the first $2\graphsize$ positions in some incidence vote requires fewer swaps than moving them to the second set of $2\graphsize$ positions (or even further), because the former requires at most $\vertexbudget (2\graphsize+\vertexbudget)$ swaps, while the latter requires at least $|B_1| \cdot 2\graphsize=(2\graphsize-\vertexbudget)\cdot 2\graphsize>\vertexbudget (2\graphsize+\vertexbudget)$ swaps; here we used the assumption $\graphsize>2\vertexbudget$.
    Hence, in all votes of~$\voters'$, the first $2\graphsize$ positions are occupied by candidates in~$B_1 \cup U^\star$. Similarly, in all votes of~$\voters'$, the last $4\graphsize$ positions are occupied by candidates in~$B_2 \cup W^\star$.
    
    This implies that the mapping~$\mapping$ from~$\candidates$ to the leaves of~$T_{8\graphsize}$ has the following properties: Let us denote the children of the root by~$t_1$ and~$t_2$, and let $t_{11}$ and~$t_{12}$ denote the children of~$t_1$. Then the candidates mapped to the leaves in the subtree rooted at~$t_{11}$, $t_{12}$, and $t_2$, are $B_1 \cup U^\star$, $(U \setminus U^\star) \cup (W \setminus W^\star)$, and $B_2 \cup W^\star$, respectively.
    Let us denote the two children of~$t_{12}$ as $t_{121}$ and $t_{122}$. See again \Cref{fig:GSbal-tree}.

    Let us call the set of swaps required to move the candidates in $B_1 \cup U^\star$ into the first~$2\graphsize$ positions and the candidates in~$B_2 \cup W^\star$ into the last~$4\graphsize$ positions, without swapping any two candidates belonging to the same set $B_1 \cup U^\star$ or $B_2 \cup W^\star$, \emph{phase-1 swaps}. The number of such swaps is exactly the value in~\eqref{eqn:swaps-in-first-step} where $F^\star$ denotes the set of edges incident to at least one vertex in~$U^\star \cup W^\star$; the same arguments presented in the paragraph preceding~\eqref{eqn:swaps-in-first-step} are valid here as well.
    
    After performing all phase-1 swaps, the $\graphsize$ positions after candidates in~$B_1 \cup U^\star$ contain the candidates in~$U \setminus U^\star$ in all anchor votes. 
    Hence, the candidates mapped by~$\mapping$ to the leaves in the subtree rooted at $t_{121}$ and at $t_{122}$ must be those is~$U \setminus U^\star$ and in~$W \setminus W^\star$, respectively, since otherwise we would need at least one additional swap in each anchor vote; however adding~$\alpha$ to the value in~\eqref{eqn:swaps-in-first-step} already exceeds the budget~$\budget$, due to $\alpha \cdot 2\graphsize\vertexbudget + \alpha>\budget$.
    Thus, in order to obtain $\pi(\vote_{u,w})$ from some incidence vote~$\vote_{u,w}$ where $u \in U \setminus U^\star$, $w \in W \setminus W^\star$, and $\{u,w\} \in F$, we also need to swap $u$ and~$w$. Let us call such swaps \emph{phase-2 swaps}; their number is exactly $|F \setminus F^\star|$.

    After performing all phase-1 and phase-2 swaps, we obtain an election that is isomorphic with the election~$E'$ described in the first direction of this proof; note that this holds irrespective of the choice of~$U^\star$ and~$W^\star$. We have already proved that turning~$E'$ into a GS-balanced election---specifically, to an election in $\domain(T_{8\graphsize},\mapping^\star)$ for $\mapping^\star=\mappinglist{\ora{B_1}, \ora{U^\star},\ora{U \setminus U^\star}, \ora{W \setminus W^\star}, \ora{W^\star}, \ora{B_2}}$ ---requires ($4(\graphsize+\vertexbudget)(g(\graphsize)+g(\vertexbudget))$) swaps; see the calculations leading to~\eqref{eqn:swaps-in-third-step}. Hence, the total number of swaps necessary to turn the original election into a GS-balanced election is at least the sum of the values~\eqref{eqn:swaps-in-first-step}--\eqref{eqn:swaps-in-third-step}, that is, at least $\budget + 2\threshold-2|F^\star|$. Recall that we assumed this value not to exceed the budget~$\budget$, which is only possible if $|F^\star| \geq \threshold$, i.e., the vertices in~$U^\star \cup W^\star$ together cover at least~$\threshold$ edges in the input graph.
    This shows the correctness of our reduction.
}

Unfortunately, in this case we do not get hardness for a constant number of voters and seeking
such a result, or an XP algorithm, 
is an interesting challenge.
We justify this view by considering %
several
domains that capture different aspects of $\GSbal$ and yield quite
different complexity results.

\paragraph{Stratification and Separation.}

$\Strat[2]$ consists of votes over even-sized candidate sets, where half of the candidates always precede those in the other half, but without
any other restrictions. This captures the basic feature of $\GSbal$
that candidates are split into two groups, but (a)~it only includes
the top-level split, and (b)~$\GSbal$ allows reversing the order of
the two groups whereas $\Strat[2]$ does not. If we do allow this
reversing, then we obtain the $\Separ$ domain. While these domains are
rather basic, they do have some interesting algorithmic features.

\begin{restatable}[\linkproof{thm:NE:swap:twoStrat:poly+FPT:voters}]{theorem}{thmNEswaptwoStratpolyFPTvoters}
\label{thm:NE:swap:twoStrat:poly+FPT:voters}    
  There is a polynomial-time algorithm for $\swapNE{\Strat[2]}$ and an
  $\FPT$ algorithm for $\swapNE{\Separ}$, parameterized by $\numVoters$.%
\end{restatable}

\prooftoappendixdivided{thm:NE:swap:twoStrat:poly+FPT:voters}{\thmNEswaptwoStratpolyFPTvoters*}{

We prove the two statements of the theorem separately. We start with the polynomial-time solvability of $\swapNE{\Strat[2]}$.

\begin{theorem}
\label{thm:NE:swap:twoStrat:poly}
    \swapNE{\Strat[2]} is in~\cc{P}.
\end{theorem}
\begin{proof}
    Let election~$E=(\candidates,\voters)$ with $\voters=(\vote_1,\dots,\vote_\numVoters)$ and budget~$\budget$ be our input instance, with $|\candidates|=\numCandidates$. 
    Let $T$ be a \typePQF-tree whose root is of type~\typeF\  and has two children, $p_1$ and~$p_2$, both of type~\typeP\ and with $\numCandidates/2$ leaves.

    Our algorithm computes an optimal mapping of the candidates to the leaves of~$T$ in a surprisingly simple way: it computes the Borda score of each candidate, and maps the candidate with the $i$-th highest Borda score to the $i$-th leaf of~$T$, resolving ties arbitrarily. Let~$\candidates_1$ and~$\candidates_2$ denote the set of candidates mapped to children of~$p_1$ and~$p_2$, respectively.  
    We claim that $E$ has swap distance at most~$\budget$ to the \Strat[2] domain if and only if it has swap distance at most~$\budget$ from $\domain(T,\mapping)$. 
    Note that the latter can be checked in polynomial time: we just need to sum up for each vote~$\vote$ the number of candidate pairs $(c_1,c_2)$ for which $c_1 \in \candidates_1$ and $c_2 \in \candidates_2$ but $c_2 \succ_\vote c_1$. Since computing the Borda scores can also be done in polynomial time, we obtain a polynomial-time algorithm. It remains to prove its correctness.  

    Let $S_1$ and~$S_2$ be a partitioning of~$C$ with $|S_1|=|S_2|$ such that the swap distance of~$E$ to $\domain(T,\mapping_S)$ is minimal where $\mapping_S$ maps candidates in~$S_i$ to the children of~$p_i$ for $i \in [2]$.
    Assume for the sake of contradiction that $S_1 \neq C_1$. Then there are candidates~$a_1 \in S_1$ and~$a_2 \in S_2$ such that $a_2$ has higher Borda score than~$a_1$, that is,
    \begin{equation}
        \label{eq:rank_diff}
        \sum_{\vote \in \voters} \rank_\vote(a_2)<\sum_{\vote \in \voters} \rank_\vote(a_1).
    \end{equation}
    Create the partitioning~$S'_1$ and~$S'_2$ by switching~$a_1$ and~$a_2$, 
    so that 
    $S'_1=S_1\setminus \{a_1\} \cup \{a_2\}$
    and~$S'_2=S_2 \setminus \{a_2\} \cup \{a_1\}$.
    Let $\budget_S$ and~$\budget_{S'}$
    denote the swap distance of~$E$ from~$\domain(T,\mapping_S)$ and from~$\domain(T,\mapping_{S'})$, respectively,
    where $\mapping_{S'}$ maps candidates in~$S'_i$ to the children of~$p_i$, $i \in [2]$.
    Note that for each vote in~$\voters$, the number of swaps between candidates in $\candidates \setminus \{a_1,a_2\}$ is the same  in both cases. The number of swaps involving either~$a_1$ or~$a_2$ in some vote~$\vote \in \voters$, denoted by~$\budget_S^{(\vote)}$ and $\budget_{S'}^{(\vote)}$ to obtain a vote in $\domain(T,\mapping_S)$ and in~$\domain(T,\mapping_{S'})$, respectively, can be computed as follows.
    If $a_1 \succ_{\vote} a_2$, then we get
        \begin{align*}
        \budget_S^{(\vote)} \!
        &=
        |\{b \colon b \in S_2, b \succ_\vote a_1\}|
        +
        |\{b \colon b \in S_1, a_2 \succ_\vote b\}|,
        \\
        \budget_{S'}^{(\vote)} \!
        &=
        |\{b \colon b \in S_1, a_1 \succ_\vote b\}|
        +
        |\{b \colon b \in S_2, b \succ_\vote a_2\}|
        + 1
        \end{align*}
    where the last term corresponds to swapping~$a_1$ and~$a_2$. Thus, we get
    \begin{align*}
    \budget_{S'}^{(v)}-\budget_S^{(v)}
    &=
    |\{b \colon b \in S_1, a_1 \succ_\vote b \succ_\vote a_2\}|
    \\
    & \qquad +   
    |\{b \colon b \in S_2, a_1 \succ_\vote b \succ_\vote a_2\}|+1
    \\
    &=\rank_\vote(a_2)-\rank_\vote(a_1).
    \end{align*}
    By symmetry, for the case $a_2 \succ_\vote a_1$ we get $\budget_{S}^{(\vote)}-\budget_{S'}^{(\vote)}=\rank_\vote(a_1)-\rank_\vote(a_2)$. Summing up over all votes, we get
    \[
    d_{S'}-d_S
    =\!  \sum_{\vote \in \voters}
    \budget_{S'}^{(v)}-\budget_S^{(v)}
    =\! \sum_{\vote \in \voters} \rank_\vote(a_2)-\rank_\vote(a_1) \! < \! 0
    \]
    where the last inequality follows from~\Cref{eq:rank_diff}. Thus, we get $\budget_{S'}<\budget_S$ which contradicts to our choice of~$(S_1,S_2)$. This proves that $S_1=\candidates_1$, i.e., our algorithm correctly finds an optimal partitioning of~$\candidates$.
\end{proof}

Next, we show that \swapNE{\Separ} can be solved in FPT time when parameterized by the number of voters. We present a simple approach that relies on our polynomial-time algorithm for \swapNE{\Strat[2]}. 

\begin{proposition}
    \label{thm:NE:swap:separation:FPT:numVoters}
    The \swapNE{\Separ} problem can be solved in $\mathcal{O}^*(2^{\numVoters})$ time.
\end{proposition}
\begin{proof}
    By the definition of the $\Separ$ domain, we know that an election $(\candidates,\voters)$ is in a $\Separ$ domain if and only if we can reverse a set~$\voters'$ of votes in~$E$ such that the obtained election is in a \Strat[2] domain. Hence, it suffices to guess this set~$\voters'$ of voters and run the polynomial-time algorithm in \Cref{thm:NE:swap:twoStrat:poly} for \swapNE{\Strat[2]}. 
    Since there are 
    $2^{\numVoters}$ possibilities to choose~$\voters'$, we obtain the claimed running time. 
\end{proof}

}

The $\FPT$ algorithm for $\Separ$ guesses which votes to reverse and
runs the algorithm for $\Strat[2]$, which itself partitions the
candidates based on their average positions (i.e., Borda scores).
Unfortunately, this approach does not
generalize and already for $\Strat[3]$ we obtain $\NP$-completeness
(but not para-$\NP$-completeness). We also show $\NP$-completeness for the $\Separ$ domain.
\begin{restatable}[\linkproof{thm:NE:swap:threeStrat:NPc}]{theorem}{thmNEswapthreeStratNPc}
\label{thm:NE:swap:threeStrat:NPc}
    \swapNE{\Strat[3]} and \swapNE{\Separ} are \NPc.
\end{restatable}
\prooftoappendixdivided{thm:NE:swap:threeStrat:NPc}{\thmNEswapthreeStratNPc*}{

We prove the two results separately. We start with the \NPcness of  \swapNE{\Strat[3]}.
\begin{theorem}
    \swapNE{\Strat[3]} is \NPc.
\end{theorem}
\begin{proof}
    We reduce from the \probName{Clique} problem. Let $(G,k)$ be our instance for \probName{Clique} with input graph $G=(U,F)$; we may assume that $k \geq 2$, $|U|>k$, and $|F| \geq \binom{k}{2}$, since otherwise the instance can be solved trivially. Let us construct an election $(\candidates,\voters)$ as follows.
    We define integers~$\alpha$ and~$\beta$ as \begin{align*}
        \alpha&=4|F|\left(|F|+|U|-k-\binom{k}{2}\right), \\
        \beta&=2\left( k(|U|-k)+\binom{k}{2}\left(|F|-\binom{k}{2}\right) \right).
    \end{align*}

    First, for each vertex in~$U$ and each edge in~$F$, we create a candidate with the same name.
    Next, we create $2\beta$ dummy candidates~$d_1,\dots,d_{2\beta}$, as well as additional sets~$B_1$ and~$B_2$ of dummies whose sizes satisfy \[\gamma=|B_1|+|U|-k=k+\binom{k}{2}+2\beta=|B_2|+|F|-\binom{k}{2}\]
    for some integer~$\gamma$. This finishes the definition of our candidate set, with $3\gamma$ candidates.

    We create two sets of  \emph{anchor voters}, $A$ and~$A'$ with $|A|=|A'|=\alpha$.
    Their preferences, provided as a list, are as follows: %
    \begin{align*}
        a &:
        B_1, \ora{U}, d_1,\dots, d_{2\beta}, \ora{F}, B_2
        \qquad \qquad \forall a \in A; \\
        a' &:
        B_1, \ola{U}, d_{2\beta},\dots, d_1, \ola{F}, B_2
        \qquad \qquad \forall a' \in A'.
    \end{align*}
    We further create \emph{incidence voters} $w_{f,u}$ and~$w'_{f,u}$ for each edge~$f \in F$ and each endpoint~$u \in f$; their preferences are as below.
    \begin{align*}
        w_{f,u} &:
        B_1, d_1,\dots, d_\beta,
        u, \ora{U \setminus \{u\}},
        \ora{F \setminus \{f\}}, f, \\
        & \qquad \qquad\qquad\qquad \qquad \qquad\qquad
        d_{\beta+1},\dots,d_{2\beta}, B_2;
        \\[4pt]
        w'_{f,u} &:
        B_1, d_1,\dots, d_\beta,
        \ola{U \setminus \{u\}}, f,
        u, \ola{F \setminus \{f\}}, \\
        & \qquad \qquad \qquad\qquad\qquad \qquad \qquad
        d_{\beta+1},\dots,d_{2\beta}, B_2.
    \end{align*}
    Hence, we have $2\alpha+4|F|$ voters in total.
    We finish the construction by setting our budget as
    \[\budget= \frac{3}{2}\alpha \cdot\beta + |F|\cdot \beta+2|F|-2\binom{k}{2}.\]

    First, assume that $G$ contains a clique of size~$k$ whose vertex and edge set are $U_K$ and~$F_K$, respectively. We are going to modify the votes in our election as follows. In each vote, we \emph{upshift} $U \setminus U_K$, meaning that we move all candidates in~$U \setminus U_K$ to the $|U|-k$ positions right after~$B_1$,
    and we
    \emph{downshift} $F \setminus F_K$, meaning that we move all candidates  in~$F \setminus F_K$ to the $|F|-\binom{k}{2}$ positions right before~$B_2$. The number of swaps necessary for this can be calculated as follows:
    \begin{itemize}
    \item Consider a pair of anchor votes $a \in A$ and~$a' \in A'$. Since the vertex candidates appear in reverse order in~$a$ and~$a'$, for each candidate~$x \in U \setminus U_K$ and each~$y \in U_K$ exactly one of these two votes ranks~$y$ above~$x$; hence, upshifting $U \setminus U_K$ costs exactly $k(|U|-k)$ swaps in the pair~$\{a,a'\}$, and thus $\alpha \cdot k(|U|-k)$ swaps in total for all anchor votes. Note that no swaps between vertex candidates and dummies or edge candidates are necessary in anchor votes.
    \item By the same arguments, downshifting~$F \setminus F_K$ in all anchor votes requires $\alpha \cdot \binom{k}{2}(|F|-\binom{k}{2})$ swaps in total.
    \item In the incidence vote pair $w_{f,u}$ and~$w'_{f,u}$ for some~$f$ and~$u$, upshifting $U \setminus U_K$ requires each candidate in~$U \setminus U_K$ to be swapped with the $\beta$ dummies $d_1,\dots,d_\beta$ in both of these votes, and---using again that the vertex candidates appear in reverse order in~$w_{f,u}$ and~$w'_{f,u}$---with each candidate of~$U_K$ exactly once, yielding $(|U|-k)(k+2\beta)$ swaps in the pair, not counting the possible swap between $f$ and~$u$ in the vote~$w'_{f,u}$.
    \item Similarly, downshifting $F \setminus F_K$
     in each pair of incidence votes requires $(|F|-\binom{k}{2})(\binom{k}{2}+2\beta)$ swaps, again not counting the possible swap between $f$ and~$u$ in the vote~$w'_{f,u}$.
    \item
    Finally, notice that we do \emph{not} need to swap~$f$ and~$u$ in~$w'_{f,u}$ if and only if $u \in U_K$  and $f \in F_K$ both hold.
    \end{itemize}
    Hence, we need $\alpha \cdot \nicefrac{\beta}{2}$ swaps in total in the anchor votes, and
    \begin{multline*}
    2|F|\cdot \left( \nicefrac{\beta}{2}+2\beta\left(|U|-k+|F|-\binom{k}{2}\right) \right)+s^\star = \\
    = |F|\cdot \beta+\alpha \cdot \beta + s^\star
    \end{multline*}
    swaps in incidence votes where $s^\star$ denotes the number of swaps between vertex and edge candidates in all incidence votes. Since $K$ is a clique, we know that $s^\star=2|F|-2\binom{k}{2}$.
    Therefore, the total number of swaps  is
    \[
     \frac{3}{2}\alpha \cdot\beta + |F|\cdot \beta+2|F|-2\binom{k}{2}=\budget.
    \]
    Observe that all of the obtained votes contain $B_1 \cup U \setminus U_K$ on the first~$\gamma$ positions,
    and contain $B_2 \cup F \setminus F_K$ on the last~$\gamma$ positions. Thus, all of these votes belong to the \Strat[3] domain determined by the tree whose first P-node has its leaves labeled with~$B_1 \cup U \setminus U_K$, whose third P-node has its leaves labeled with~$B_2 \cup F \setminus F_K$, leaving the dummies $d_1,\dots,d_{2\beta}$ and the candidates in~$U_K \cup F_K$ for the second P-node. Hence, the constructed election with budget~$\budget$ is a yes-instance of \swapNE{\Strat[3]}.

    \smallskip
    For the other direction, assume that by applying at most $\budget$ swaps, we can obtain an election~$(\candidates,\voters')$ that belongs to the \Strat[3] domain, i.e., we can partition the candidates set into sets~$\candidates_1,\candidates_2,\candidates_3$, each of size~$\gamma$, such that in each vote, the first $\gamma$ candidates are from~$\candidates_1$ and the last~$\gamma$ candidates are from~$\candidates_3$. Note first that we may assume that $B_1 \subseteq \candidates_1$ and $B_2 \subseteq \candidates_3$: indeed, each candidate of~$B_1$ is ranked above every candidate of~$\candidates \setminus B_1$ in each vote, so if some $b \in B_1$ were contained in~$\candidates_2 \cup \candidates_3$ while some candidate~$c \notin B_1$ is contained in~$\candidates_1$, then exchanging the roles of~$b$ and~$c$ would not increase the number of necessary swaps in any vote; the argument for~$B_2$ is symmetric.

    We claim that $\candidates_1 \subseteq B_1 \cup U$ and $\candidates_3 \subseteq B_2 \cup F$.
    Indeed, assuming otherwise, we either have to move some dummy~$d_i$, $i \in [2\beta]$, or some candidate in~$F$ to the first~$\gamma$ positions,
    or we have to move some dummy~$d_i$, $i \in [2\beta]$, or some candidate in~$U$ to the last~$\gamma$ positions. In either case, this requires at least $\alpha \cdot 2\beta$ swaps in the anchor votes. To see this for the first case, recall that the dummies appear in reverse order in~$A$ and in~$A'$: moving~$d_i$ to the first~$\gamma$ positions requires at least $i+k$ swaps in each vote of~$A$, and at least $2\beta+1-i+k$ swaps in each vote of~$A'$, amounting to more than $\alpha \cdot 2\beta$ swaps; the remaining cases are analogous.
    However, it is not hard to see that $\budget<2\alpha \cdot \beta$, a contradiction.
    Hence, there exists a set~$U'$ of~$k$ vertices and a set~$F'$ of $\binom{k}{2}$ edges such that
    $\candidates_1=B_1 \cup U \setminus U'$,
    $\candidates_2=U' \cup \{d_1,\dots,d_{2\beta}\} \cup F'$, and~$\candidates_3=B_2 \cup F \setminus F'$.

    We claim that $U'$ forms a clique in~$G$ with edge set~$F'$. To see this, notice that the number of swaps necessary can be computed along the same lines as showed earlier; in fact, the choice of~$U'$ and~$F'$ only matters for computing the number of swaps necessary between vertex candidates and edge candidates in incidence votes. Indeed, the number of swaps required to upshift~$U \setminus U'$ and downshift~$F \setminus F'$ in all votes is exactly $\frac{3}{2}\alpha \cdot \beta + |F| \beta + (2|F|- s')$ where $s'$ denotes the number of those pairs~$(f,u)$ for which $f \in F'$, $u \in U'$, and $u$ is an endpoint of~$f$. Since the number of swaps applied must be at most~$\budget$, we get that $s' \geq 2 \binom{k}{2}$. However, this is only possible if for each $f \in F'$, both endpoints of~$f$ are vertices in~$U'$. By $|F'|=\binom{k}{2}$ and~$|U'|=k$, this implies that $U'$ induces a clique of size~$k$ in~$G$, as required.
    This proves the correctness of our reduction.
\end{proof}

Next, we show that \swapNE{\Separ} is also intractable.
\begin{theorem}
    \label{thm:NE:swap:separation:NPh}
    The \swapNE{\Separ} problem is \NPc.
\end{theorem}
\begin{proof}
    Before we present the reduction, we introduce some auxiliary notation. A partition $(L,R)$ of~$\candidates$ is \emph{balanced} if $|L|=|R|=\nicefrac{\numCandidates}{2}$. We set $\domain_{\{L,R\}}$ to be a domain containing all votes where candidates of $L$ are above candidates of $R$ and vice versa. Clearly, $\domain_{\{L,R\}}$ belongs to the $\Separ$ family, and we say that $\domain_{\{L,R\}}$ \emph{is specified} by $(L,R)$. Additionally, we define the \emph{rank sum} for each voter $v$ and subsets of candidates $S\subseteq \candidates$ as
    \[
        \rho_v(S) = \sum_{c\in S} \rank_v(c)\,.
    \]
    
    Clearly, the rank sum of a set of $\nicefrac{\numCandidates}{2}$ candidates is at least $\underline{\rho} = 1 + \cdots + \nicefrac{\numCandidates}{2}$ and at most $\bar{\rho} = (\nicefrac{\numCandidates}{2}+1) + \cdots + \numCandidates$. Throughout, we write $K = (\nicefrac{\numCandidates}{2})^2$. Observe that $\bar{\rho} - \underline{\rho} = K$ and that $\rho_v(L) + \rho_v(R) = \underline{\rho} + \bar{\rho}$ for every voter $v\in \voters$ and every balanced partition $(L,R)$.

    Now, we show that the distance of a vote $v$ to $\domain_{\{L,R\}}$ for any balanced partition $(L,R)$ can be fully expressed using rank sum.

    \begin{claim}\label{claim:separ:rank-sum}
        Let $(L,R)$ be a balanced partition of $\candidates$, and $v\in\voters$ be a voter. Then
        \[
            \min_{\succ\in \domain_{\{L,R\}}} \dist(v,\succ) = \min\{ \bar{\rho} - \rho_v(R), \rho_v(R) - \underline{\rho} \}\,.
        \]
    \end{claim}
    \begin{claimproof}
        First, let $\succ$ be a vote such that all candidates of $L$ are ranked above all candidates of $R$. We say that a pair of candidates $\{c,d\}$ with $c\in L$ and $d\in R$ is \emph{dirty} if $d \succ_v c$ (recall that $c \succ d$ by our assumption). Then, the distance between $v$ and $\succ$ is at least
        \[
            x = |\{(c,d) \in L \times R \colon d \succ_v c\}|\,.
        \]
        Let $q_1 < \cdots < q_{\nicefrac{\numCandidates}{2}}$ be the positions of the candidates of $R$ in vote $v$. Below position $q_t$, there are $\numCandidates - q_t$ candidates, out of which $\nicefrac{\numCandidates}{2} - t$ belong to $R$. Therefore, we have
        \begin{multline*}
            x = \sum_{t=1}^{\nicefrac{\numCandidates}{2}} (\numCandidates - q_t - \frac{\numCandidates}{2} + t) = K + \underline{\rho} - \rho_v(R) = \bar{\rho} - \rho_v(R)\,,
        \end{multline*}
        where the last equality uses $\bar{\rho} - \underline{\rho} = K$. Moreover, this bound can be achieved by a vote $\succ$ which agrees on the internal ordering of $L$ and $R$ with $v$. Hence, the minimum distance to the votes ranking $L$ above $R$ is exactly $\bar{\rho} - \rho_v(R)$. Swapping the roles of $R$ and $L$, the minimum distance to the votes ranking $R$ above $L$ is $\bar{\rho} - \rho_v(L) = \rho_v(R) - \underline{\rho}$, where we used $\rho_v(L) + \rho_v(R) = \underline{\rho} + \bar{\rho}$.
    \end{claimproof}

    Using the previous claim, we can express the distance of an election $E$ to the $\Separ$ domain as
    \begin{multline}
    \label{eq:separ:dist}
        \dist(E,\Separ) = \\ = \min_{(L,R)} \sum_{v\in \voters} \min\{ \bar{\rho} - \rho_v(R), \rho_v(R) - \underline{\rho} \}\,.
    \end{multline}
    This finishes our preparation, and now we prove \NPhness{} of the problem.

    We reduce from the \probName{Minimum Bisection} problem, which is known to be \NPc. In this problem, we are given a graph $G=(W,F)$ with $|W| = 2N$ and an integer $k\in \N$, and the goal is to decide whether $W$ can be partitioned into two parts $P_1$ and $P_2$, each of size $N$, such that the number of edges with one endpoint in $P_1$ and the other in $P_2$ is at most $k$. The pair $(P_1,P_2)$ is called a \emph{bisection} of $G$. Without loss of generality, we assume $N \geq 2$ and $0 \leq k < |F|$, as otherwise the instance is trivial.

    \proofsubparagraph{Construction.} Let $\mathcal{I} = (G=(W,F),k)$ be an instance of \probName{Minimum Bisection} and $p$ be the \emph{odd} integer $p = 4N^2+1$.
    The candidate set $\candidates$ of the constructed instance $\mathcal{J}$ consists of~$W$, the \emph{vertex candidates}, and two sets of \emph{dummy candidates}~$D_1$ and $D_2$, each of size $p$. We fix an arbitrary ordering over~$\candidates$, and for any candidate set~$S \subseteq C$, we use $\ora{S}$ to denote the restriction of this ordering to~$S$, and $\ola{S}$ for its reverse.

    The set of voters consists of several different blocks of voters. First, we create $\alpha= |F|\cdot K$ copies of two \emph{anchor votes} $a$ and $a'$ with
    \[
        a\colon \ora{D_1} \succ_a \ora{W} \succ_a \ora{D_2}
    \]
    and
    \[
        a'\colon \ora{D_1} \succ_{a'} \ola{W} \succ_{a'} \ora{D_2}\,.
    \]

    Before we define additional voters, we introduce additional notation. We call position $2N+2p + 1 -j$ the \emph{mirror} of position $j$. Let $d_1\in D_1$ and $d_2\in D_2$ be a fixed pair of \emph{distinguished dummies} and, for each $i\in\{1,2\}$, let us fix a \emph{guard pair} $\delta_i = (\delta^1_i,\delta^2_i)$ such that $\{\delta^1_1,\delta^2_1\} \subseteq D_1 \setminus\{d_1\}$ and $\{\delta^1_2,\delta^2_2\} \subseteq D_2 \setminus\{d_2\}$. We call positions $p+2,\ldots,p+2N-1$ the \emph{critical positions}. For each $i\in\{1,2\}$, we define \emph{dummy template} $\pi_i$, which is a mapping between the dummies in~$(D_1 \cup D_2) \setminus \{\delta_i^1,\delta_i^2\}$ and the positions outside the critical positions and $\{1,2,2N+2p-1,2N+2p\}$. Specifically, $\pi_i$ puts $d_1$ on position $r = \nicefrac{(p+1)}{2}$ and $d_2$ on its mirror; note that $3 \leq r \leq p+1$ since $p \geq 17$. The remaining $2p-4$ dummies are distributed over the $p-2$ mirrored position pairs $\{\{q,2N+2p+1-q\} \colon q \in \{3,\ldots,p+1\}\setminus\{r\} \}$ so that each pair receives two dummies of the same set $D_1$ or $D_2$. The latter is possible: for template $\pi_i$, the numbers of dummies of $D_1$ and of $D_2$ still to be assigned are $p-3$ and $p-1$ (for $i=1$; vice versa for $i=2$), and both are even because $p$ is odd. Now, we are ready to introduce the \emph{edge votes}. In particular, for every edge $e=\{u,w\}\in F$ and every $i\in\{1,2\}$, we introduce the following votes:
    \begin{center}
        \begin{tabular}{@{\hspace{3pt}}c|c|c|c|c|c@{\hspace{3pt}}}
            \toprule
                & $1$ & $2$ & critical & $2N+2p-1$ & $2N+2p$ \\
            \midrule
            $v_e^{i,1}$ & $u$ & $w$ & $\ora{W\setminus\{u,w\}}$ & $\delta_i^1$ & $\delta_i^2$\\
            $v_e^{i,2}$ & $\delta_i^2$ & $\delta_i^1$ & $\ola{W\setminus\{u,w\}}$ & $w$ & $u$\\
            $v_e^{i,3}$ & $u$ & $w$ & $\ola{W\setminus\{u,w\}}$ & $\delta_i^1$ & $\delta_i^2$\\
            $v_e^{i,4}$ & $\delta_i^2$ & $\delta_i^1$ & $\ora{W\setminus\{u,w\}}$ & $w$ & $u$\\
            \bottomrule
        \end{tabular}
    \end{center}
    The rest of the votes are filled according to the dummy template $\pi_i$.

    The election has $\numVoters = 2\alpha + 8\cdot|F|$ voters and $|W| + 2p$ candidates. To finalize the construction, we set the budget to
    \[
        \budget = \alpha \cdot N^2 + |F|\cdot(4K-2B) - 2(2N+p)\cdot(|F|-k)\,,
    \]
    where $B = 4N+4p-4$. Next, we show that $G$ has a bisection cutting at most $k$ edges if and only if $\dist((\candidates,\voters),\Separ) \leq \budget$.

    \proofsubparagraph{Additional notation and basic facts.} Recall that we have $\numCandidates = 2N+2p$. Throughout the analysis, we additionally write $M = \numCandidates+1$, $h = \nicefrac{\numCandidates}{2} = N+p$, $Q = 2N+p$, and~$\varsigma_{\max} = 2(N-1)^2$; observe that this yields $K = h^2$ and $\budget = \alpha\cdot N^2 + |F|\cdot(4K-2B) - 2Q\cdot(|F|-k)$. Note that a candidate contributes $j$ to a rank sum when placed on position $j$, and $M-j$ when placed on its mirror. 
    Furthermore, we use the following two inequalities:
    \begin{equation}\label{eq:separ:ineq}
        \varsigma_{\max} \leq p-3
        \qquad\text{and}\qquad
        3N^2+2N+1 \leq K\,.
    \end{equation}
    The first inequality holds since $\varsigma_{\max} < 2N^2 \leq \nicefrac{p}{2} \leq p-3$, and the second one since $K = (N+p)^2 \geq p^2 \geq 16N^4$.

    We say that a balanced partition $(L,R)$ of $\candidates$ is \emph{aligned} if one of its two sides (i.e., either its top or its bottom half)  contains all of $D_2$ and no candidate of $D_1$. The two sides of an aligned partition are thus $D_1\cup W_1$ and $D_2\cup W_2$ for some partition $(W_1,W_2)$ of $W$ with $|W_1|=|W_2|=N$; that is, aligned partitions correspond exactly to bisections of~$G$.

    \proofsubparagraph{Vote pairs.} Our votes come in pairs that we analyze jointly. First, for any two votes $v,v'$ and any balanced partition $(L,R)$, we abbreviate $P = \rho_v(R)+\rho_{v'}(R)$ and $\Delta = \rho_v(R)-\rho_{v'}(R)$. Expanding the sum of the two minima from \Cref{claim:separ:rank-sum} into the minimum of four sums yields
    \begin{multline}\label{eq:separ:pair}
        \min_{\succ\in\domain_{\{L,R\}}}\dist(v,\succ)
        + \min_{\succ\in\domain_{\{L,R\}}}\dist(v',\succ)\\
        = \min\{ P-2\underline{\rho},\ 2\bar{\rho}-P,\ K-|\Delta| \}\,,
    \end{multline}
    where we used $(\rho_v(R)-\underline{\rho}) + (\bar{\rho}-\rho_{v'}(R)) = K+\Delta$ and $\bar{\rho}-\underline{\rho}=K$. Second, we call $(v,v')$ a \emph{partner pair with flip set $X$} if $v'$ arises from $v$ by moving every candidate of~$X$ from its position to the mirror position, while every other candidate keeps its position. In a partner pair, each candidate $c\in X\cap R$ contributes $\rank_v(c)$ to $\rho_v(R)$ and $M-\rank_v(c)$ to $\rho_{v'}(R)$, so
    \begin{multline}\label{eq:separ:partner}
        P = M\cdot|X\cap R| + 2\rho_v(R\setminus X)
        \qquad\text{and}\\
        \Delta = \sum_{c\in X\cap R}(2\rank_v(c)-M)\,.
    \end{multline}

    \proofsubparagraph{Contribution of the anchor votes.} Fix a balanced partition~$(L,R)$ and group the anchor votes into $\alpha$ pairs $(a,a')$; each of them is a partner pair with flip set $X=W$. By~\eqref{eq:separ:partner}, we have
    \begin{multline*}
        |\Delta| \leq \sum_{j=p+1}^{p+2N}|2j-M| = 2\cdot(1+3+\cdots+(2N-1)) = 2N^2
        \\\text{and}\qquad
        P = \sum_{c\in R} g(c)\,,
    \end{multline*}
    where we define the function~$g$ such that $g(c) = M$ if $c$ is a vertex candidate and ${g(c) = 2\rank_a(c)}$ if $c$ is a dummy. In the anchor votes, the dummies of $D_1$ have $g$-values $2,4,\ldots,2p$, all at most $M-(2N+1)$, and the dummies of $D_2$ have $g$-values $2(2N+p+1),2(2N+p+2),\ldots,2(2N+2p)$, all at least $M+(2N+1)$. Consequently, over all sets $S$ of $h$ candidates, the sum $\sum_{c\in S}g(c)$ is maximized precisely when $D_2\subseteq S$ and $S\cap D_1=\emptyset$ (the rest of $S$ consisting of vertex candidates, whose $g$-values coincide), and if $S$ misses a dummy of $D_2$ or contains a dummy of $D_1$, then $\sum_{c\in S}g(c)\leq P^\ast-(2N+1)$, where $P^\ast$ denotes the maximum value. Using $M = 2h+1$, $p = h-N$, and $2\bar{\rho} = 2\underline{\rho}+2K = h(h+1)+2h^2 = 3h^2+h$, a direct computation gives
    \begin{multline*}
        P^\ast = M\cdot N + 2\cdot\sum_{j=2N+p+1}^{2N+2p} j
        \\= (2h+1)\cdot N + (h-N)(3h+N+1)\\
        = 3h^2+h-N^2\,,
    \end{multline*}
    and therefore $2\bar{\rho}-P^\ast = N^2$ and $P^\ast-2\underline{\rho} = 2K-N^2$.

    Now, we apply~\eqref{eq:separ:pair}. Since $\rho_a(L)+\rho_a(R) = \underline{\rho}+\bar{\rho}$ holds for each anchor vote, we have $P-2\underline{\rho} = 2\bar{\rho}-\sum_{c\in L}g(c)$, so both of the first two terms in~\eqref{eq:separ:pair} are of the form $2\bar{\rho}-\sum_{c\in S}g(c)$ with $S\in\{L,R\}$. If $(L,R)$ is aligned, then one of its sides attains $P^\ast$, so the pair $(a,a')$ costs
    \[
        \min\{ 2K-N^2,\ N^2,\ K-|\Delta| \} = N^2
    \]
    \emph{exactly}---irrespective of how the vertex candidates are split---because $K-|\Delta| \geq K-2N^2 \geq N^2+2N+1$ by~\eqref{eq:separ:ineq}. 
    If $(L,R)$ is not aligned, then each of its two sides misses a dummy of $D_2$ or contains a dummy of $D_1$, so the first two terms in~\eqref{eq:separ:pair} are at least $2\bar{\rho}-P^\ast+(2N+1) = N^2+2N+1$; moreover $K-|\Delta| \geq K-2N^2 \geq N^2+2N+1$ by~\eqref{eq:separ:ineq}, as before. Hence, the pair costs at least $N^2+2N+1$.
    
    \proofsubparagraph{Contribution of the edge votes under aligned partitions.} Let $(L,R)$ be an aligned partition with $D_1\subseteq L$ and $D_2\subseteq R$ (the roles of the two sides being interchangeable), and let $W_1 = W\cap L$ and $W_2 = W\cap R$. Fix an edge $e=\{u,w\}\in F$ and a template $i\in\{1,2\}$. Both $(v_e^{i,1},v_e^{i,2})$ and $(v_e^{i,3},v_e^{i,4})$ are partner pairs with flip set $X = W\cup\{\delta_i^1,\delta_i^2\}$: positions $1$ and $2$ are the mirrors of positions $2N+2p$ and $2N+2p-1$, respectively, the critical positions are closed under mirroring, and the template dummies keep their positions.

    We first compute $P$, which is the same for both pairs. The fixed candidates of $R$, i.e., $R\setminus X$, are exactly the template dummies of $D_2$: the distinguished dummy $d_2$ sits on the mirror of position $r$ and contributes $M-r$ to ${\rho_{v_e^{i,1}}(R\setminus X)}$, and each same-set pair of $D_2$-dummies occupies two mirrored positions and contributes $M$; the number of such pairs is $\nicefrac{(p-1)}{2}$ for $i=1$ and $\nicefrac{(p-3)}{2}$ for $i=2$. The flipping candidates of $R$ number $|X\cap R| = |W_2| = N$ for $i=1$ and $|X\cap R| = N+2$ for $i=2$---in \emph{every} case, as the vertices $u$ and $w$ are already counted in $W_2$ if they lie in $R$. Hence, by~\eqref{eq:separ:partner}, in all cases
    \[
        P = M\cdot(N+p+1) - 2r = M\cdot(h+1)-(p+1)\,,
    \]
    where we used $2r = p+1$. Consequently, by $2\underline{\rho}=h(h+1)$ and $2h-p = Q$, we obtain
    \[
        P-2\underline{\rho} = (h+1)^2-(p+1) = K+Q
    \]
    and
    \[
        2\bar{\rho}-P = 2K-(K+Q) = K-Q\,.
    \]

    Next, we compute the differences $\Delta$. Let
    \[
        \varsigma = \sum_{c\in W_2\setminus\{u,w\}} (2\rank_{v_e^{i,1}}(c)-M)\,,
    \]
    so
    \[
        |\varsigma| \leq \sum_{j=p+2}^{p+2N-1} |2j-M| = 2\cdot(1+3+\cdots+(2N-3)) = \varsigma_{\max}\,,
    \]
    and the corresponding quantity for the pair $(v_e^{i,3},v_e^{i,4})$ is $-\varsigma$, since there the critical positions carry $\ola{W\setminus\{u,w\}}$. By~\eqref{eq:separ:partner}, the two pairs of template $i$ have $\Delta = c_i+\varsigma$ and $\Delta = c_i-\varsigma$, respectively, where $c_i$ collects the terms $2\rank_{v_e^{i,1}}(c)-M$ of~$u$, $w$, and the guard pair~$\delta_i$ insofar as they lie in $R$: vertex $u$ (of rank $1$) contributes $-(2N+2p-1)$ if $u\in R$,  vertex~$w$ (of rank $2$) contributes $-(2N+2p-3)$ if $w\in R$, and the guard pair $\delta_2\subseteq D_2\subseteq R$ (whose candidates rank $2N+2p-1$ and $2N+2p$, and are contained in $X$ only for $i=2$) contributes $(2N+2p-3)+(2N+2p-1) = B$. Explicitly:
    \begin{center}
        \begin{tabular}{@{\hspace{3pt}}c@{\hspace{7pt}}c@{\hspace{8pt}}c@{\hspace{7pt}}c@{\hspace{8pt}}c@{\hspace{3pt}}}
            \toprule
            & $u\in L$ & $u\in R$ & $u\in R$ & $u\in L$ \\
            & $w\in L$ & $w\in R$ & $w\in L$ & $w\in R$\\
            \midrule
            $c_1$ & $0$ & $-B$ & $-(2N+2p-1)$ & $-(2N+2p-3)$\\
            $c_2$ & $B$ & $0$ & $2N+2p-3$ & $2N+2p-1$\\
            \bottomrule
        \end{tabular}
    \end{center}
    By~\eqref{eq:separ:pair}, the two pairs of template $i$ cost $K-\max\{Q,|c_i+\varsigma|\}$ and $K-\max\{Q,|c_i-\varsigma|\}$, respectively. From $\varsigma_{\max}\leq p-3$ in~\eqref{eq:separ:ineq}, we get both \[\varsigma_{\max} < Q \quad \text{ and } \quad (2N+2p-3)-\varsigma_{\max} \geq 2N+p = Q\] (and a fortiori $(2N+2p-1)-\varsigma_{\max}\geq Q$ and $B-\varsigma_{\max}\geq Q$).

\smallskip
    If $e$ is \emph{not} cut by the bisection $(W_1,W_2)$, then $c_i=0$ for exactly one of the two templates; since $|{\pm\varsigma}|\leq \varsigma_{\max} < Q$, its two pairs cost $K-Q$ each. For the other template, we have $|c_i|=B$, and its two pairs cost $(K-(B-\varsigma)) + (K-(B+\varsigma)) = 2K-2B$ in total. Altogether, the eight votes of an uncut edge cost exactly
    \[
        (2K-2Q) + (2K-2B) = 4K-2B-2Q\,.
    \]
    
    If $e$ \emph{is} cut, then $\{|c_1|,|c_2|\} = \{2N+2p-1,\ 2N+2p-3\}$, and the four pairs cost
    \begin{multline*}
        \sum_{i\in\{1,2\}} \bigl( (K-(|c_i|-\varsigma)) + (K-(|c_i|+\varsigma)) \bigr)
        \\= 4K - 2(2N+2p-1) - 2(2N+2p-3)
        \\= 4K-2B
    \end{multline*}
    in total. Thus, a cut edge costs exactly $2Q$ more swaps than an uncut one, and for every aligned partition $(L,R)$, we obtain
    \begin{multline}\label{eq:separ:aligned-total}
        \sum_{v\in\voters} \min_{\succ\in\domain_{\{L,R\}}}\dist(v,\succ) = 
        \\
        = \alpha\cdot N^2 + |F|\cdot(4K-2B) - 2Q\cdot(|F|-t)
    \end{multline}
    where $t$ is the number of edges of $G$ cut by the bisection~$(W_1,W_2)$.

    \proofsubparagraph{$\mathcal{I}$ is a yes-instance $\Rightarrow$ $\mathcal{J}$ is a yes-instance.} Suppose $G$ admits a bisection $(P_1,P_2)$ cutting $t\leq k$ edges, and consider the aligned partition $(D_1\cup P_1,\ D_2\cup P_2)$; it is balanced because $|D_1|=|D_2|$ and $|P_1|=|P_2|$. By~\eqref{eq:separ:dist} and~\eqref{eq:separ:aligned-total}, we obtain
    \begin{multline*}
        \dist((\candidates,\voters),\Separ)
        \leq 
        \\
        \leq \alpha\cdot N^2 + |F|\cdot(4K-2B) - 2Q\cdot(|F|-t)
        \leq \budget\,.
    \end{multline*}

    \proofsubparagraph{$\mathcal{J}$ is a yes-instance $\Rightarrow$ $\mathcal{I}$ is a yes-instance.} Suppose $\dist((\candidates,\voters),\Separ)\leq\budget$, and let $(L,R)$ be a balanced partition attaining the minimum in~\eqref{eq:separ:dist}. First, assume for the sake of contradiction that $(L,R)$ is not aligned. Since every summand in~\eqref{eq:separ:dist} is non-negative and each of the $\alpha$ anchor pairs costs at least $N^2+2N+1$, we get
    \begin{align*}
        \dist((\candidates,\voters),\Separ)
        & \geq   
        \alpha\cdot(N^2+2N+1) \\
        & > \alpha\cdot N^2 + 4K\cdot|F|
        \geq \budget
    \end{align*}
    where the strict inequality holds because \[\alpha\cdot(2N+1) = |F|\cdot K\cdot(2N+1) \geq 5\cdot|F|\cdot K > 4K\cdot|F|\] by the choice of $\alpha$ and the assumption $N\geq2$, a contradiction. Therefore, $(L,R)$ is aligned, say with $D_1\subseteq L$, and $(W\cap L,\ W\cap R)$ is a bisection of $G$; let $t$ be the number of edges it cuts. By~\eqref{eq:separ:aligned-total}, we have
    \begin{multline*}
        \alpha\cdot N^2 + |F|\cdot(4K-2B) - 2Q\cdot(|F|-t)
        = 
        \\
        = \dist((\candidates,\voters),\Separ)
        \leq \budget
    \end{multline*}       
    which, by the definition of $\budget$, yields $t\leq k$. This means that  $(W\cap L, W\cap R)$ is a bisection of $G$ cutting at most $k$ edges, which completes the proof.
\end{proof}
}

Yet, for $\Strat[\nicefrac{\numCandidates}{2}]$, where
$\numCandidates$ is the number of candidates, we do obtain
para-$\NP$-completeness.
Here, as in  $\GSbal$,
there are pairs of adjacent candidates that we can freely
swap, but
we cannot make any other swaps. %

\begin{restatable}[\linkproof{thm:NE:swap:mhalfStrat:NPc}]{theorem}{thmNEswapmhalfStratNPc}
\label{thm:NE:swap:mhalfStrat:NPc}
  \swapNE{\Strat[\nicefrac{\numCandidates}{2}]} is \NPc, even for elections with
  $4$ voters. %
\end{restatable}

\prooftoappendix{thm:NE:swap:mhalfStrat:NPc}{\thmNEswapmhalfStratNPc*}{
    We reduce from the \probName{Minimum Feedback Arc Set (FAS)} problem, whose input is a directed graph~$D$ and an integer~$k \in \mathbb{N}$, and the task is to decide whether $D$ contains a set of~$k$ arcs whose deletion from~$D$ yields an acyclic directed graph. We first create the duplicate~$D'$, i.e, the disjoint union of two copies of~$D$. Next, using the method described by
    \citet{bac-bra-gei-har-kar-pet-see:j:weighted-majority-relation}, we reduce the instance $(D',2k)$ of \probName{FAS} to an equivalent instance $(D'',2k)$ in which the digraph~$D''$ is \emph{induced} by an election~$E$ with~$4$ voters over the vertices of~$D''$, meaning that the majority digraph of~$E$ is~$D''$; such an election~$E$ can be constructed in polynomial time. 
    Since $D''$ is obtained by subdividing each arc of~$D'$ exactly once, the vertices in~$D''$ can still be partitioned into two sets~$U=\{u_1,\ldots,u_\numCandidates\}$ and~$U'=\{u'_1,\ldots,u'_\numCandidates\}$ such that $D''$ is the disjoint union of the two isomorphic subgraphs of~$D''$ induced by~$U$ and~$U'$, with $u'_i$ being the copy of~$u_i$ for each $i \in [\numCandidates]$.  
    Moreover, the construction for~$E$ ensures the property that for each arc~$(v,w)$ present in~$D''$, exactly three voters prefer~$v$ to~$w$.
    We set the budget as $\budget = 2\binom{2\numCandidates}{2}-2\numCandidates-|\arcset|+4k$ where $\arcset$ is the arc set of~$D''$. 
    
    First, assume that $D''$ admits a minimal feedback arc set of size at most~$2k$. Clearly, such a solution must contain a set~$S$ of at most~$k$ arcs from~$D''[U]$ and a set~$S'$ of~$|S|$ arcs from~$D''[U']$, because these two digraphs are isomorphic---in fact, we may assume that each arc in~$S'$ is a copy of one in~$S$. Let $u_{i_1},\ldots,u_{i_\numCandidates}$ be a topological ordering of the graph obtained from~$D''[U]$ by deleting the arcs in~$S$,  
    and let $T_\numCandidates$ be a \typePQF-tree with~$\numCandidates$ leaves representing the \Strat[\numCandidates/2] domain.
    We are going to show that $E$ has distance at most~$\budget$ from $\domain(T_\numCandidates,\mappinglist{u_{i_1},u'_{i_1},\ldots,u_{i_\numCandidates},u'_{i_\numCandidates}})$.
    
    Note that for each $i \in [\numCandidates]$, candidates~$u_i$ and~$u'_i$ need not be swapped in any of the votes, because they correspond to the two children of a \typeP-node in~$T_\numCandidates$. 
    By contrast, for each arc~$f=(u_i,u_j) \in \arcset \cap (U \times U)$, we perform the following steps: 
    If $f \notin S$, then we know that $u_i$ precedes~$u_j$ in the topological ordering, so we need to swap $u_i$ and~$u_j$ in the unique vote that ranks~$u_j$ before~$u_i$.
    This amounts to $\frac{|\arcset|}{2} - |S|$ swaps.
    If $f \in S$, then by the minimality of~$S$, we know that $u_j$ must precede~$u_i$ in the topological ordering, so we need to swap these candidates in all three votes that rank~$u_i$ before~$u_j$, yielding $3\cdot |S|$ swaps.
    We need the same number of swaps for arcs in~$\arcset \cap (U' \times U')$.
    Finally, for each pair of vertices in~$D''$ that are not connected by an arc but are not formed by a vertex in~$U$ and its copy, we need to swap the corresponding candidates in exactly two votes.
    Thus, the total number of swaps used is 
    \begin{multline*}
        2 \left(\frac{|\arcset|}{2}-|S|  +3\cdot|S| \right) +
        2\left(\binom{2\numCandidates}{2} -\numCandidates -|\arcset|\right) \\
        = 2 \binom{2\numCandidates}{2}
        -2\numCandidates+4|S|-|\arcset|\leq \budget
    \end{multline*}
        \\
    due to $|S| \leq k$. 
    This proves that the constructed election~$E$ (whose majority digraph is~$D''$) has swap distance at most~$\budget$ from $\domain(T_\numCandidates,\mappinglist{u_{i_1},u'_{i_1},\dots,u_{i_\numCandidates},u'_{i_\numCandidates}})$.
    
    Assume now that $E$ has swap distance at most~$\budget$ from some election~$\widetilde{E}$ in the \Strat[\numCandidates/2] domain. Then $\widetilde{E}$ is in~$\domain(T_\numCandidates,\mappinglist{v_1,\dots,v_{2\numCandidates}})$ for some ordering $\sigma=v_1,\dots,v_{2\numCandidates}$ of the vertex set~$U \cup U'$ of~$D''$; note that for each vertex in~$U \cup U'$ we use an alias $v_j$ that encodes its position in~$\sigma$. 
    We may assume that for each $i \in [\numCandidates]$ we have $(v_{2i},v_{2i-1}) \notin \arcset$, as otherwise we can simply switch these two candidates in the ordering~$\sigma$ without changing the domain. Let $\arcset^\star$ denote the set of arcs in~$\arcset$ of the form~$(v_{2i-1},v_{2i})$ for some~$i \in [\numCandidates]$. Let $S$ contains those arcs~$(v_i,v_j)$ where $v_j$ precedes~$v_i$ in~$\sigma$ (i.e., the ``backward'' arcs for~$\sigma$). Note that $S \cap \arcset^\star=\emptyset$, and all arcs in~$\arcset \setminus S$ point from some~$v_j$ to~$v_{j'}$ with $j<j'$ and thus form an acyclic digraph. Thus, $S$ is a feedback arc set for~$D''$. We will show that $|S| \leq 2k$.  
    
    To see this, we compute the number of swaps necessary to turn~$E$ into an election in~$\domain(T_\numCandidates,\mappinglist{v_1,\dots,v_{2\numCandidates}})$ as follows.
    First, for each arc in~$S$, its endpoints must be swapped in three votes, yielding $3\cdot |S|$ swaps. For each arc in~$\arcset \setminus \arcset^\star \setminus S$, we need to swap the endpoint candidates in exactly one vote. 
    Finally, each pair of vertices in~$D''$ that are not connected by an arc and are not of the form~$\{v_{2i-1},v_{2i}\} $ for some~$i \in [\numCandidates]$
    must be swapped in exactly two votes; the number of such candidate pairs is $\binom{2m}{2}-|\arcset|-(\numCandidates-|\arcset^\star|)$, because there are $\binom{2m}{2}-|\arcset|$ non-arcs in~$D''$ and exactly $\numCandidates-|\arcset^\star|$ of them are of the form~$\{v_{2i-1},v_{2i}\}$. Thus, the number of swaps necessary is
    \begin{align*}
        3\cdot &  |S| +|\arcset \setminus \arcset^\star \setminus S| + 
        2 \left( \binom{2\numCandidates}{2}-|\arcset|-\numCandidates+|\arcset^\star| \right) \\
        &= 2\cdot |S| -|\arcset| +
        2\binom{2\numCandidates}{2}
        + |\arcset^\star|-2\numCandidates \\ 
        & = \budget -4k+2 \cdot |S|+|\arcset^\star| \leq \budget
    \end{align*}
    which implies $|S| \leq 2k - \frac{|\arcset^\star|}{2} \leq 2k$, as required. Hence, $S$ is a solution to our instance $(D'',2k)$ of \probName{FAS},
    proving the correctness of the reduction.
}

\paragraph{Antagonism.}
Each domain in the $\Antag$ family consists of two votes, each being a
reverse of the other one.  Hence, $\swapNE{\Antag}$ is very similar to
\KemenyOA and we easily inherit its hardness.

\begin{restatable}[\linkproof{thm:NE:swap:antag:NPc}]{theorem}{thmNEswapantagNPc}
\label{thm:NE:swap:antag:NPc}
    \swapNE{\Antag} is \NPc, even for elections with
  $4$ voters. %
\end{restatable}

\prooftoappendix{thm:NE:swap:antag:NPc}{\thmNEswapantagNPc*}{
    We show \NPhness by reduction from the \KemenyOA problem. Let the input instance of \KemenyOA consist of an election~$E=(C,V)$  with $C = \{c_1, c_2,\ldots, c_m\}$ and $V = (\vote_1, \vote_2, \ldots, \vote_n)$, and a budget~$\budget$. We create a new election $E'=(C', V')$ over $C' = C \cup \{c_1', c_2',\ldots,c_{2\budget+2}'\}$ by simply appending the $2\budget+2$ new candidates in a fixed order at the end of each vote,  i.e., for each $\vote_i \in V$ we create a corresponding vote $\vote_i'$ by 
    \begin{equation*}
        \vote_i' = (\vote_i, c_1', c_2',\ldots,c_{2\budget+2}' ),
    \end{equation*}
    and set $V' = (\vote_1',\ldots, \vote_n')$.

    We claim that $E'$ has distance at most~$\budget$ from some $\Antag$ domain if and only if $(E,\budget)$ is a yes-instance of  \KemenyOA.

    \proofsubparagraph{Direction ``$\Rightarrow$''.} Assume $(E,\budget)$ is a yes-instance of the \KemenyOA problem. Then there is some common ordering $\vote^* \in \mathcal{L}(C)$ which has total swap distance at most~$\budget$ from all votes in $V$. In consequence, $E'$ is swap distance at most~$\budget$ from the $\Antag$ domain that contains the vote
    \begin{equation*}
        (\vote^*,c_1', c_2',\ldots,c_{2\budget+2}' )
    \end{equation*}
    and its reverse.
    
    \proofsubparagraph{Direction ``$\Leftarrow$''.} Suppose that $E'$ has swap distance at most~$\budget$ from some election~$E''$ in the $\Antag$ domain. Since each vote of~$E''$ can be obtained from some vote in~$V'$ by at most~$\budget$ swaps, candidate~$c'_{2\budget+2}$ appears in the bottom half of each vote of~$E''$. Hence, $E''$ cannot contain both some vote and its reverse, i.e., each vote in~$E''$ is a copy of the same vote~$u$. Considering only swaps among candidates in~$C$, this means that we can transform each vote in~$V$ into~$\restr{u}{C}$ using a total of at most~$\budget$ swaps. Thus, $(E,\budget)$ is a yes-instance of \KemenyOA.

    As \KemenyOA is \NPh for 4 voters, the result for $\numVoters=4$ follows.
}

\paragraph{GS.}
Unfortunately, we were not able to establish the exact complexity of $\swapNE{\GS}$. On the one hand, we face similar difficulties as for $\GSbal$ and parameterization by the number of voters, but even more importantly, all our previous domains were defined via trees of a single ``shape,'' whereas here we have many different possibilities.

\toappendix{
\subsection{\NPcness of \swapNE{\SCros}}

    In their paper, \citet{lak-pet-elk:c:nearly-single-crossing} claim that their reduction \cite[Theorem~1]{lak-pet-elk:c:nearly-single-crossing} proving the $\NP$-hardness of computing the swap distance to a single-crossing election \emph{with a fixed ordering of voters} ``can easily be extended to the setting where voters can be permuted.'' However, they left the claim without a formal proof, and we do not find the extension entirely trivial. Therefore, we formally prove that it is indeed the case by our own proof, building on the reduction of \citet[Theorem~1]{lak-pet-elk:c:nearly-single-crossing}.

\begin{proposition}\label{thm:NE:swap:singlecrossing:NPh}
    \swapNE{\SCros} is \NPc.
\end{proposition}
\begin{proof}
    We reduce from the \NPc problem called \probName{Global Swaps} in \cite{lak-pet-elk:c:nearly-single-crossing}. The input of this problem is an election~$E=(\candidates,\voters)$ with the votes in~$\voters$ ordered as $\vote_1,\dots,\vote_\numVoters$, and an integer~$\budget$. Its task is to decide whether there exists an election $(\candidates,(\vote_1',\dots,\vote_\numCandidates'))$ with $\sum_{i \in [\numVoters]} \dist(\vote_i,\vote'_i) \leq \budget$ that is single-crossing \emph{with respect to the given ordering of its votes}. We construct an equivalent instance of \swapNE{\SCros} as follows.

    For each $i \in [\numVoters-1]$, we add a set~$A_i$ of $\budget+2$ dummy candidates, and we fix an ordering over them. Next, for each $i \in [\numVoters]$, we create a vote $\hat\vote_i$ by appending these dummies at the end of the vote~$\vote_i$ as below:
    \begin{align*}
    \hat\vote_1 &: \vote_1, \ora{A_1},\dots,\ora{A_{\numVoters-1}};     
    \\
    \hat\vote_i &: \vote_i, \ola{A_1}, \dots,\ola{A_{i-1}}, \ora{A_i},\dots,\ora{A_{\numVoters-1}} \qquad \text{ for $1 <i<\numCandidates$};  
    \\
    \hat\vote_{\numVoters} &: \vote_\numVoters, \ola{A_1}, \dots,\ola{A_{\numVoters-1}}. 
    \end{align*}
    Let $\hat{E}=(\candidates \cup A_1 \cup \dots \cup A_{\numVoters-1}),\hat\voters)$ be the obtained election where $\hat\voters=(\hat\vote_1,\dots,\hat\vote_{\numVoters})$.
    We finish the construction by setting the budget as $\budget$.

    First, if $(E,\budget)$ is a yes-instance of \probName{Global Swaps}, then it is clear that $(\hat{E},\budget)$ is a yes-instance of \swapNE{\SCros}: applying the same swaps that turn $E$ into a single-crossing election with respect to the given ordering of the votes trivially turns $\hat{E}$ as well into a single-crossing election. 

    For the other direction, assume that there is an election~$E'$ within swap distance at most~$\budget$ from~$\hat{E}$ that is single-crossing with respect to \emph{some} ordering~$\sigma$ of the votes. For each $i \in [\numVoters]$, let $\hat\vote'_i$ denote the vote in~$E'$ obtained from~$\hat\vote_i$.
    Since $\hat\vote'_i$ can be obtained from  $\hat\vote_i \in \hat\voters$ using at most~$\budget$  swaps, the first and last dummies in~$\ora{A_j}$, say $a_j^1$ and $a_j^{\budget+2}$, for each $j \in [\numVoters-1]$ must retain their ranking as determined by~$\hat\vote_i$, because swapping them in~$\hat\vote_i$ would require at least~$\budget+1$ swaps.
    This means that for each $j \in [\numVoters-1]$, either all votes in $\{\hat\vote'_1,\dots,\hat\vote'_{j-1}\}$ precede all votes in $\{\hat\vote'_j,\dots,\hat\vote'_{\numVoters}\}$ in $\sigma$ or vice versa, as otherwise $a_j^1$ and $a_j^{\budget+2}$ would cross at least twice in the ordering~$\sigma$ of the votes. Therefore, $\sigma$ must order all votes of~$E'$ either in the order $\hat\vote'_1,\dots,\hat\vote'_\numVoters$ or 
    in the order $\hat\vote'_\numVoters, \dots, \hat\vote'_1$. Since the single-crossing domain with respect to these two orderings coincide, we obtain that $E'$ is, in fact, single-crossing with respect to the ordering $\hat\vote'_1,\dots,\hat\vote'_\numVoters$ of the voters. Thus, restricting~$E'$ to~$C$ yields an election within swap distance at most~$\budget$ from~$E$ that is single-crossing with respect to the prescribed ordering of the votes, showing that $(E,\budget)$ is a yes-instance of \probName{Global Swaps}.
\end{proof}
}

\section{Algorithmic Results}
\label{sec:algorithms}

One of the most practical ways of solving $\NP$-complete problems
is by reducing them to integer linear programming (ILP).  Doing so is,
indeed, possible in our case but, somewhat surprisingly, in
\Cref{sec:experiments} we will see that our ILPs are fairly slow to
solve by modern solvers.

\begin{theorem}
  Given an election with $\numCandidates$ candidates and~$\numVoters$
  voters, there is a polynomial-time reduction of \swapNE{\mathcal{X}}, $\mathcal{X}\in\{\GScat,\GSbal\}$, to \probName{Integer Linear Programming} with $\mathcal{O}(nm^2)$ variables and $\mathcal{O}(m^3 + nm^2)$ constraints if $\mathcal{X}=\GScat$, and $\mathcal{O}(nm^2 \log m)$ variables and $\mathcal{O}((n + \log m ) \cdot m^2 \log m)$ constraints if $\mathcal{X} = \GSbal$.%
\end{theorem}
Consequently, we seek  direct, combinatorial $\FPT$ algorithms, 
parameterized either by the number of candidates or
the distance from the given domain.

\subsection{Parameter: Number of Candidates}
\label{sec:param_cands}
\appendixsection{sec:param_cands}
\toappendix{
In this section, we present all omitted proofs for our algorithmic results that consider the number of candidates as a parameter.  
}
In case of parameterization by the number $\numCandidates$ of
candidates, the baseline idea is to perform a brute-force search over
all trees and mappings. This yields $\Ohstar{2^{m\log m}}$ running
time.

\begin{restatable}[\linkproof{thm:NE:swap:any:FPT:numCandidates}]{theorem}{thmGenericCandidates}
  \label{thm:NE:swap:any:FPT:numCandidates}
  For each of our families of domains defined by $\typePQF$-trees, there is an $\FPT$ algorithm for \swapNE{\calX} when parameterized by the number of candidates.%
\end{restatable}
\prooftoappendix{thm:NE:swap:any:FPT:numCandidates}{\thmGenericCandidates*}{
  Let $\calX$ be a family of domains over a set~$\candidates$ of~$\numCandidates$ candidates that can be represented by a family $\mathcal{T}$ of \typePQF-trees.  
  The number of \typePQF trees in such a family~$\mathcal{T}$ is at most $c^\Oh{\numCandidates}$ for some
  constant $c$. Therefore, we can enumerate all trees~$T\in\mathcal{T}$ and,
  for each such tree~$T$, try all possible bijections~$\mapping$ from~$\candidates$ to the leaves of~$T$. Then, in
  $\Oh{\numCandidates!}$ time, we compute all votes in~$\calD(T,\mapping)$, and for each vote in the input election~$E$ we find the closest vote in~$\calD(T,\mapping)$ in $|\calD(T,\mapping)| \cdot \Oh{\numCandidates^2  \numVoters}$
  time. This enables us to compute $\dist(E,\calD(T,\mapping))$. There are
  $\numCandidates! = 2^\Oh{\numCandidates\cdot\log\numCandidates}$
  possible bijections~$\mapping$ for every $T$, so the overall running
  time of the algorithm is
  \[
    c^\Oh{\numCandidates} \cdot \Oh{\numCandidates!} \cdot
    \Oh{\numCandidates!} \cdot \Oh{\numCandidates^2 \numVoters} =
    2^\Oh{\numCandidates\log\numCandidates}\cdot \numVoters,
  \]
  which is clearly in \FPT. The algorithm is correct, as we
  try all maximal elections in any domain over~$C$ belonging to~$\calX$.
}

Using dynamic programming, we obtain substantially faster
algorithms. In particular, for $\GS$, the underlying idea is to build
the $\typePQF$ tree of the closest $\calD(T,\mapping)$ domain
bottom-up, at each node counting only the swaps required to ensure
that the voters rank candidates from one subtree ahead (or behind)
those from the other subtree. Consequently, for each node and each
candidate, we need to know if this candidate appears in this node and,
if so, if it belongs to the left or the right subtree. This requires 
 $\Ohstar{3^m}$ time.

\begin{restatable}[\linkproof{thm:NE:swap:GSbin:FPT:numCandidates}]{theorem}{thmGSCandidates}
  \label{thm:NE:swap:GSbin:FPT:numCandidates}
    There is an algorithm solving \swapNE{\GS} in $\Oh{3^\numCandidates\cdot \numCandidates\numVoters}$ time.
\end{restatable}
\prooftoappendix{thm:NE:swap:GSbin:FPT:numCandidates}{\thmGSCandidates*}{
  \newcommand{\DP}{\operatorname{DP}} The algorithm is based on
  dynamic programming over all subsets of candidates. Formally, we
  have a table $\DP[ X ]$, where $X\subseteq C$, which stores the
  distance to an optimal election $E'\in\GS$ for an election
  $(X,\voters)$, where $\voters$ is a list of the original votes
  restricted only to candidates in~$X$.

  Once the dynamic programming table is computed, we simply compare
  the value stored in $\DP[C]$ and compare it with the budget
  $\budget$. By definition, the stored value is the distance from the
  closest election $E'\in \GS$, so the result is correct, assuming that the table
  is filled correctly.

  Now, we formally describe how to compute each cell of the
  table. First, assume that $|X| \leq 2$. Then, there is exactly one
  possible $\typeQ$-tree $T$ and two possible bijections from $X$ to
  the leaves of~$T$, both with distance $0$ from
  $(X,\voters)$. Therefore, in this case, we set $ \DP[ X ] = 0.  $

  Now, let $|X| \geq 3$. In this case, we try all possible partitionings
  of $X$ between the left and right subtrees. Note that there are
  $2^{|X|}-2$ possibilities (putting all candidates to the left or
  right subtree is an invalid partitioning). We denote by~$L$ the
  subset of~$X$ in the left subtree and by~$R$ the subset of~$X$ for
  the right subtree, i.e., $R = X \setminus L$. Now, for every vote~$\vote\in\voters$, in~$\Oh{\numCandidates}$ time, we can compute the
  value~$\delta_{v,L,X}$: the minimum number of swaps required to have
  all candidates in~$L$ either as the top~$|L|$ candidates in~$\restr{\vote}{X}$ or the bottom~$|L|$ candidates in~$\restr{\vote}{X}$ (we
  take a smaller of these two values), and then use recursion on the subelections induced by~$L$ and~$R$. Formally, the computation is
  as follows.
  \[
    \DP[ X ] =  \min_{\emptyset \subset L \subset X}
        \left\{ \sum_{v\in\voters} \delta_{v,L,X} + \DP[L] +
      \DP[X\setminus L]\right\}.
  \]
    
  The total time required to fill the table is
  \begin{align*}
    \sum_{X \subseteq \candidates} (2^{|X|}-1)\cdot \Oh{\numCandidates\numVoters} 
    &= \Oh{\numCandidates\numVoters} \cdot \sum_{k=0}^\numCandidates \binom{\numCandidates}{k}\cdot 2^k\\
    & = \Oh{\numCandidates\numVoters} \cdot (1+2)^\numCandidates\\
    &= \Oh{\numCandidates\numVoters} \cdot 3^\numCandidates.
    \end{align*}
    After the table is filled, we can decide the instance in constant
    time.
}

\noindent
For other domains, variants of this algorithm are even faster.

\begin{restatable}[\linkproof{thm:NE:swap:GScat:FPT:numCandidates}]{theorem}{thmGSCandidatesBalCat}
    \label{thm:NE:swap:GScat:FPT:numCandidates}
    \label{thm:NE:swap:GSbal:FPT:numCandidates}
    \swapNE{\mathcal{X}} can be solved in
        $\Oh{2^\numCandidates\cdot \numCandidates\numVoters}$ time if $\mathcal{X} = \{\GScat,\Separ, \Strat[\nicefrac{\numCandidates}{2}]\}$,
        $\Oh{(2\sqrt{2})^\numCandidates\cdot\numCandidates\numVoters}$ time if $\mathcal{X} = \GSbal$, and
        $\Oh{3^\numCandidates\cdot\numCandidates\numVoters}$ time if
        $\mathcal{X} = \Strat[k]$, for every $k$.
\end{restatable}  
\prooftoappendix{thm:NE:swap:GScat:FPT:numCandidates}{\thmGSCandidatesBalCat*}{
    \newcommand{\DP}{\operatorname{DP}}
    We use the dynamic programming approach described in \Cref{thm:NE:swap:GSbin:FPT:numCandidates}, tailored to the specific domain family at hand; we re-use the notation described in the proof of \Cref{thm:NE:swap:GSbin:FPT:numCandidates}.
    
    First, let $\mathcal{X} = \GScat$. Here, the computation for any set~$X\subseteq \candidates$ of size at least $3$ reduces to the following.
    \[
        \DP[ X ] = \min_{c\in X}\left\{
            \sum_{v\in\voters} \delta_{v,L,X} + \DP[X\setminus \{c\}]\right\}.
    \]
    That is, each cell can be computed in $\Oh{\numCandidates\numVoters}$ time. Therefore, as there are $2^\numCandidates$ cells in total, we directly obtain the running time of $2^\numCandidates\cdot\Oh{\numCandidates\numVoters}$.

    If $\mathcal{X}=\GSbal$, then the base case is again the same. However, the computation for $X\subseteq \candidates$ is modified. First, we can optimize by fixing $r \in X$, which always goes to $L$. This is due to the fact that in perfect binary $\typeQ$-trees, the distinction between left and right child is not important. This saves half of the possible partitions. Formally, the computation is as follows.
    \[
        \DP[ X ] = \min_{\substack{L \subseteq X\\|L| = |X|/2\\r \in L}}\left\{ \sum_{v\in\voters} \delta_{v,L,X} + \DP[L] + \DP[X\setminus L] \right\}\,.
    \]
    Moreover, observe that we need to consider only sets $X$ such that $|X| = 2^k$ for some $k\in\N$. The DP table has entries for all $X \subseteq \candidates$ with $|X| = 2^j$ for $j \in \{0, 1, \ldots, \log_2 \numCandidates\}$. For $|X| \leq 2$, the cost is $\Oh{1}$. For $|X| = 2^j$ with $j \geq 2$, we enumerate $\binom{2^j - 1}{2^{j-1} - 1}$ partitions (fixing $r \in L$), and each partition costs $\Oh{2^j \cdot \numVoters}$ to evaluate $\sum_{v \in \voters} \delta_{v,L,X}$. The total running time is therefore
    \begin{align*}
        \sum_{j=2}^{\log_2 \numCandidates} & \binom{\numCandidates}{2^j} \cdot \binom{2^j - 1}{2^{j-1} - 1} \cdot \Oh{2^j \cdot \numVoters}\\
        &= \Oh{\numCandidates \cdot \numVoters} \cdot \sum_{j=2}^{\log_2 \numCandidates} \binom{\numCandidates}{2^j} \binom{2^j}{2^{j-1}} \\
        &= \Oh{(2\sqrt{2})^{\numCandidates} \cdot \numCandidates \cdot \numVoters}\,,
    \end{align*}
    where the last step uses that the sum $\sum_j \binom{\numCandidates}{2^j}\binom{2^j}{2^{j-1}}$ is dominated by its $j = \log_2 \numCandidates - 1$ term, which is $\Oh{2^{3\numCandidates/2}}$ by Stirling's approximation. Thus, the algorithm runs in $\Oh{(2\sqrt{2})^{\numCandidates} \cdot \numCandidates \cdot \numVoters}$ time.

    If $\mathcal{X} = \Separ$, we do the following. We guess the subset $L\subseteq \candidates$, $|L| = |C|/2$, of candidates in the left part, and for each such guess, we compute $\delta_{\vote,L}$ which is the swap distance of $\vote$ from the $\Separ$ domain corresponding to~$L$. Overall, the running time is $\Oh{2^\numCandidates \cdot \numCandidates \cdot \numVoters}$.

    Finally, let $\mathcal{X} = \Strat[k]$ for a given $k\ge 2$ with~$k \mid \numCandidates$, and let $b = \nicefrac{\numCandidates}{k}$ be the
    block size. Recall that each domain in $\Strat[k](\candidates)$ is given
    by an ordered partition of $\candidates$ into $k$ blocks of size~$b$: its
    votes rank all candidates of the first block on top, in an arbitrary
    order, then all candidates of the second block, and so on. The
    table now has a cell $\DP[X]$ for every $X\subseteq\candidates$ whose size is divisible by~$b$; it stores the distance of the election
    $(X,\restr{\voters}{X})$ to the family of domains given by ordered
    partitions of~$X$ into $\nicefrac{|X|}{b}$ blocks of size~$b$. Since
    within a single block the votes are unconstrained, we set $\DP[X]=0$
    for $|X|\le b$. For $|X|\ge 2b$, we try all possible topmost blocks,
    i.e., all sets $B\subseteq X$ with $|B|=b$. For every vote
    $\vote\in\voters$, in $\Oh{\numCandidates}$ time we can compute the value~$\delta^{\uparrow}_{\vote,B,X}$: the minimum number of swaps required to
    have all candidates in $B$ as the top $|B|$ candidates in
    $\restr{\vote}{X}$, which equals the number of pairs consisting of a
    candidate of~$B$ and a candidate of $X\setminus B$ that
    $\restr{\vote}{X}$ ranks above it. In contrast to the previous cases,
    the reversed placement is not available, as the root is of
    type~$\typeF$ and not~$\typeQ$. Formally, the computation is as
    follows.
    \[
      \DP[X] = \min_{\substack{B\subseteq X\\ |B| = b}}
      \left\{\sum_{\vote\in\voters}\delta^{\uparrow}_{\vote,B,X}
        + \DP[X\setminus B]\right\}.
    \]
    The recurrence is correct, as each pair of candidates from two different
    blocks is charged exactly once, namely in the cell in which the block of
    its higher-ranked member is chosen as~$B$, whereas pairs within a single
    block are never charged. Using
    $\binom{\numCandidates}{jb}\binom{jb}{b}
    =\binom{\numCandidates}{b}\binom{\numCandidates-b}{(j-1)b}$, the total
    time required to fill the table is
    \begin{multline*}
      \sum_{\substack{X\subseteq\candidates\\ b\,\mid\,|X|}}
      \binom{|X|}{b}\cdot\Oh{\numCandidates\numVoters}
      = \Oh{\numCandidates\numVoters}\cdot
        \binom{\numCandidates}{b}\sum_{j\ge1}\binom{\numCandidates-b}{(j-1)b}\\
        \le \Oh{\numCandidates\numVoters}\cdot
        \binom{\numCandidates}{b}\cdot 2^{\numCandidates-b}
        = \Oh{3^{\numCandidates}\cdot\numCandidates\numVoters}\,,
    \end{multline*}
    where we use
    $\binom{\numCandidates}{b}2^{\numCandidates-b}
    \le\sum_{s=0}^{\numCandidates}\binom{\numCandidates}{s}2^{\numCandidates-s}
    =3^{\numCandidates}$ in the last equality. For $\Strat[\numCandidates/2]$, that is, for~$b= 2$, we have $\binom{\numCandidates}{b}2^{\numCandidates-b}
    =\Oh{2^{\numCandidates}\numCandidates^{2}}$, and hence the running time
    becomes~$\Oh{2^{\numCandidates}\cdot\numCandidates^{3}\numVoters}$.
}

We conclude with our domains that cannot be defined using \typePQF-trees. For a single-peaked election, there is a trivial brute-force algorithm trying all possible axes $\axis$.

\begin{restatable}[\linkproof{thm:NE:swap:singlePeak:FPT:numCandidates}]{observation}{thmSPFPTCandidates}
    \label{thm:NE:swap:singlePeak:FPT:numCandidates}
    \swapNE{\SPeak} can be solved in $2^\Oh{\numCandidates \cdot \log\numCandidates} \cdot \numVoters$ time.
\end{restatable}
\prooftoappendix{thm:NE:swap:singlePeak:FPT:numCandidates}{\thmSPFPTCandidates*}{
    As our first step, we guess the axis $\axis$ of a closest single-peaked election. Next, we generate the set of all votes~$\voters'$ compatible with the guessed axis $\axis$. Then, for each vote $v\in\voters$, we find the closest vote $v'\in\voters'$ by enumeration. Finally, we compare $\sum_{v\in \voters} \dist(v,v')$ with $\budget$. If the sum is smaller than $\budget$, we return \emph{yes}. Otherwise, we continue with another axis $\axis$. If we exhausted all axes without finding a solution, we return \emph{no}.

    There are $\numCandidates!\in \numCandidates^\numCandidates \in 2^\Oh{\numCandidates\log\numCandidates}$ different axes. There are $2^{\numCandidates-1}$ votes corresponding to a fixed axis; therefore, we can compute the closest votes in time $\numVoters \cdot 2^{\numCandidates-1} \cdot \numCandidates^\Oh{1}$. Hence, the algorithm runs in $2^\Oh{\numCandidates\log\numCandidates} \cdot \numVoters$ time.
}

For the single-crossing domain, the argument is not as straightforward; however, it follows a similar strategy. %
In particular, we prove that it is enough to focus only on the %
maximal single-crossing domains~\cite{pup-sli:j:single-crossing}. %

\begin{restatable}[\linkproof{thm:NE:swap:singleCrossing:FPT:numCandidates}]{theorem}{thmSCFPTCandidates}
    \label{thm:NE:swap:singleCrossing:FPT:numCandidates}
    The \swapNE{\SCros} problem can be solved in $2^\Oh{\numCandidates^2 \cdot \log\numCandidates} \cdot \numVoters$ time.
\end{restatable}
\prooftoappendix{thm:NE:swap:singleCrossing:FPT:numCandidates}{\thmSCFPTCandidates*}{
    As our first step, we guess the domain $\domain'$ of a closest
    single-crossing election, as follows. We guess a vote
    $u_0\in\calL(\candidates)$, initialize $\domain' := \{u_0\}$ and
    $u := u_0$, and repeat the following $\binom{\numCandidates}{2}$ times:
    We guess two candidates that are adjacent in~$u$ and that $u$ orders in
    the same way as~$u_0$ does, we swap them in~$u$, and we add the
    resulting vote to~$\domain'$. Next, for each vote $v\in\voters$, we
    find, a closest vote $v'\in\domain'$ by enumeration. Finally, we
    compare $\sum_{v\in\voters}\dist(v,v')$ with~$\budget$. If the sum is
    at most~$\budget$, we return \emph{yes}. Otherwise, we continue with
    the next guess. If we exhausted all guesses without finding a solution,
    we return \emph{no}.

    To analyze the procedure, we relabel the candidates so that $u_0$
    becomes the identity permutation; then the pairs of candidates that a
    vote orders differently than~$u_0$ become its inversions, and
    containment of inversion sets defines the \emph{weak Bruhat order}
    on~$\calL(\candidates)$. This order is graded by the number of
    inversions, its maximum element is the reverse of~$u_0$, and $u'$
    \emph{covers}~$u$ exactly if $u'$ arises from~$u$ by swapping two candidates
    that are adjacent in~$u$ and ordered as in~$u_0$; we refer
    to~\citet[Chapter~3]{bjo-bre:b:coxeter-groups} for these facts. Hence,
    each step of our procedure moves from~$u$ to a cover of~$u$, and  such a
    cover exists as long as $u$ is not the maximum element, and after
    $\binom{\numCandidates}{2}$ steps the maximum is reached; so the
    procedure is well-defined, and the guessed sets~$\domain'$ are exactly
    the maximal chains of the weak Bruhat order.

    Further, a domain is single-crossing if and only if, for one of its
    votes~$t_1$, the sets of pairs on which its votes disagree with~$t_1$
    are pairwise
    nested~\citep[Proposition~1]{jae-pet-elk:c:nearly-single-crossing}, that
    is, if and only if it is a chain of the weak Bruhat order relabeled so
    that $t_1$ is the identity. In particular, every guessed~$\domain'$ is a
    single-crossing domain (these maximal chains are known in social choice
    as the maximal single-crossing
    domains~\citep{pup-sli:j:single-crossing});
    consequently, whenever we return \emph{yes}, replacing each vote~$v$
    by~$v'$ turns $E$ into a single-crossing election using at most~$\budget$
    swaps.

    Conversely, suppose that $\dist(E,\SCros)\le\budget$, witnessed by a
    domain $\domain\in\SCros(\candidates)$ with first vote~$t_1$ as above.
    Then $\domain$ is a chain of the weak Bruhat order relabeled at~$t_1$,
    and every chain of a finite partial order is contained in a maximal
    one; so some~$\domain'$ guessed for $u_0=t_1$ satisfies
    $\domain\subseteq\domain'$. Then, for every vote~$v$, we have
    $\min_{w\in\domain'}\dist(v,w)\le\min_{w\in\domain}\dist(v,w)$, so the
    sum of computed distances from~$\domain'$ is at most $\dist(E,\domain)\le\budget$,
    and we return \emph{yes}.

    There are $\numCandidates!$ possibilities for the initial vote~$u_0$
    and at most $\numCandidates-1$ possibilities for the swapped pair in
    each of the $\binom{\numCandidates}{2}$ steps; therefore, there are
    $\numCandidates!\cdot(\numCandidates-1)^{\binom{\numCandidates}{2}}\in
    2^{\Oh{\numCandidates^{2}\log\numCandidates}}$ guesses in total. There
    are $\binom{\numCandidates}{2}+1$ votes in a fixed guessed domain;
    therefore, we can compute the closest votes in
    $\numVoters\cdot\numCandidates^{\Oh{1}}$ time per guess. Hence, the
    algorithm runs in
    $2^{\Oh{\numCandidates^{2}\log\numCandidates}}\cdot\numVoters$ time.
}

\subsection{Parameter: Distance to a Domain}

\label{sec:param_distance}
\appendixsection{sec:param_distance}

\toappendix{
In this section, we present all omitted proofs for our algorithmic results that consider the distance budget~$\budget$ as the
parameter. In \Cref{thm:NE:swap:binary:FPT:budget}, we also present an alternative algorithm to solve the \swapNE{\GS} problem---running in \FPT time with~$d$ as the parameter---that proved to be efficient in practice, even though 
its worst-case running time is dominated by the one showed in~\Cref{thm:NE:swap:GS:FPT:budget} which relies on the characterization of $\GS$ domains through forbidden partial subelections. In addition, we also present a proof showing that $\Ident$ and $\Antag$ domains are hereditary (\Cref{thm:ident:antag:hereditary}). 
}

\def\forbCandidates{\candidates^\star}
\def\forbVoters{\voters^\star}
\def\F{\mathcal{F}}
\def\R{\mathcal{R}}
\def\maxProfileSize{\sigma}
\def\maxForbCand{\gamma}

For the parameterization by the distance to a domain, we obtain a
general algorithm that applies to a wide class of domains, including
many of ours,
but we need some additional
terminology and notation before presenting it.

A domain~$\domain$ is \emph{hereditary}, if for every
election~$E \in \domain$, all subelections of~$E$ belong to~$\domain$
(restricted to a given set of candidates).  It is known that every
hereditary domain can be characterized by a (possibly infinite) set~$\F$ of
forbidden subelections~\cite{lac-lac:j:hereditary-characterization}.
In fact, %
a more concise characterization is
obtained through a family of \emph{forbidden partial subelections},
each of them a pair~$(\forbCandidates,\forbVoters)$ where
$\forbCandidates$ is a set of candidates and $\forbVoters$ is a
collection of \emph{partial votes}, i.e., partial orders
over~$\forbCandidates$.  We assume that each partial
vote~$v \in \forbVoters$ is given as a set~$P_v$ of pairs $(a,b)$
for which $a \succ_v b$ and there exists no
$c \in \forbCandidates$ with $a \succ_v c \succ_v
b$. 
The \emph{size} of a partial subelection~$(\forbCandidates,\forbVoters$) is
$\sum_{v \in \forbVoters} |P_v|$.
A domain~$\domain$ is \emph{characterized by a set~$\F$ of forbidden partial subelections} if  any election~$E$ belongs to~$\domain$ exactly if there is no $(\candidates^\star,\voters^\star) \in \F$ whose linear extension is present in~$E$.

\begin{restatable}[\linkproof{thm:NE:swap:domainByForbidden:FPT:budget}]{theorem}{thmNEswapdomainByForbiddenFPTbudget}
\label{thm:NE:swap:domainByForbidden:FPT:budget}
    Let $\domain$ be a domain characterized by a  family~$\F$ of forbidden partial subelections, each with at most~$\maxForbCand$ candidates and of size at most~$\maxProfileSize$. Then \swapNE{\domain} can be solved in $\Oh{\max(2,\maxProfileSize)^\budget \cdot |\F| \cdot \numCandidates^{\maxForbCand+1} \numVoters}$ time.
\end{restatable}
\begin{proofsketch}
	The algorithm is a bounded-search-tree exploration of the search space. We start with an election $E=(\candidates,\voters)$ and we maintain for each voter $v_i\in\voters$ a set $R_i$ of \emph{committed inversions}, that is, pairs of candidates whose relative order in the sought election must be the opposite of their order in~$v_i$.
    We call an election \emph{admissible} for $\R$ if, for each $i \in [\numVoters]$, its~$i$-th vote inverts all pairs in~$R_i$ relative to~$v_i$. 
    We consider $\mathcal{R}=(R_1,\ldots,R_{\numVoters})$ the \emph{state} of the algorithm; initially, all sets $R_i$ are empty.
    We describe how to compute a value~$f(\R)$, for each state~$\mathcal{R}$, defined as the minimum swap distance between~$E$ and an election in~$\domain$ admissible for~$\mathcal{R}$ if this distance is at most $\budget$, and $\infty$ otherwise. 
    
    Since the swap distance of two votes equals the number of pairs of candidates on which they disagree, every election admissible for $\mathcal{R}$ is at swap distance at least $\sum_{i}|R_i|$ from~$E$. Thus, the algorithm sets $f(\R)=\infty$ whenever $\sum_{i}\!|R_i|>\budget$.

	Assume the algorithm is in a state $\mathcal{R}$ with $\sum_i|R_i|\leq\budget$. For each voter $v_i$, we build the tournament~$T_i$ obtained from~$v_i$ by reversing exactly the arcs corresponding to the pairs in~$R_i$. If some~$T_i$ contains a directed triangle, we set~$f(\R)=\infty$ in case all three arcs of the triangle belong to~$R_i$, and otherwise we branch by adding to~$R_i$ each of the (at most two) pairs of the triangle outside~$R_i$ and set $f(\R)=\min f(\R')$ where the minimum is taken over all  states~$\R'$ resulting from the branching. If every tournament~$T_i$ is acyclic, it is transitive and admits a unique linear extension $\hat{v}_i$. We let $\hat{E}=(\candidates,(\hat{v}_1,\ldots,\hat{v}_{\numVoters}))$. If $\hat{E}\in \domain$, we set $f(\R)=\sum_i|R_i|$. Otherwise, we find some $(\forbCandidates,\forbVoters) \in \F$ present in~$\hat{E}$;
    in particular, for every injection $\pi:\candidates^\star \to \candidates$ and each partial vote $v \in \forbVoters$ we search for a vote $\varphi(v) \in \{ \hat{v}_1, \dots,\hat{v}_{\numVoters}\}$ that extends the partial vote obtained from~$v$ by replacing each candidate~$c \in \candidates^\star$ with~$\pi(c)$. 
    Note that this can be done in $|\F| \cdot  \numCandidates^{\maxForbCand} \cdot \mathcal{O}(\numCandidates \numVoters)$ time. 
    We branch over all voters $v_i \in \{\varphi(v)\colon v \in \forbVoters\}$ and over all pairs $(a,b) \in P_{\varphi^{-1}(v_i)}$ for which $\{\pi(a),\pi(b)\}$ is not a committed inversion in~$R_i$; %
    in the branch corresponding to~$v_i$ and~$(a,b)$, we add $\{\pi(a),\pi(b)\}$ to $R_i$. We return~$f(\R)=\min f(\R')$ where the minimum is taken over all states~$\R'$ obtained in some branch. Every branch adds a new pair to some~$R_i$, so $\sum_i|R_i|$ strictly increases
    along every branch. Moreover, the recursion terminates when $\sum_i |R_i|$ exceeds~$\budget$.
    The algorithm starts in state $\R_0=(\emptyset,\ldots,\emptyset)$ and returns \emph{yes} if and only if $f(\R_0)$ is finite. 

    We prove the correctness of the algorithm in 
    \iflong \Cref{app:proof:thm:NE:swap:domainByForbidden:FPT:budget}.
    \fi %
    \toappendix{
    \subsection{Proof correctness for \Cref{thm:NE:swap:domainByForbidden:FPT:budget}}\label{app:proof:thm:NE:swap:domainByForbidden:FPT:budget}
	
    We prove the correctness of the algorithm, described in the main text, in a series of claims.
	\begin{claim}
		\label{clm:SPeak:swap:FPT:exactness}
		Let $\mathcal{R}$ be a state in which every tournament $T_i$ is acyclic. Then for every voter $v_i$, the vote $\hat{v}_i$ disagrees with~$v_i$ exactly on the pairs in~$R_i$. Consequently, $\hat{E}$ is admissible for~$\mathcal{R}$ and at swap distance exactly $\sum_i|R_i|$ from $E$, and every election admissible  for~$\R$ agrees with $\hat{E}$ on all pairs in~$R_i$, $i \in [\numVoters]$.
	\end{claim}
	\begin{claimproof}
		The vote $\hat{v}_i$ is the unique linear extension of the transitive tournament $T_i$, so it realizes every arc of $T_i$: it inverts, relative to $v_i$, all pairs in $R_i$, and it orders every pair outside $R_i$ as $v_i$ does, by construction of $T_i$. Hence, $\hat{v}_i$ disagrees with $v_i$ exactly on $R_i$, so $\hat{E}$ is admissible and $\sum_i\dist(v_i,\hat{v}_i)=\sum_i|R_i|$. Finally, any vote admissible for~$R_i$ inverts every pair in $R_i$ relative to $v_i$, and so does~$\hat{v}_i$; hence, the two agree on all pairs in $R_i$.
	\end{claimproof}

	\begin{claim}
		\label{clm:SPeak:swap:FPT:soundness}
		Whenever  $f(\mathcal{R})=\budget' \neq \infty$, we have $\budget'\leq\budget$ and there is an election in~$\domain$ admissible for $\mathcal{R}$ at swap distance exactly $\budget'$ from~$E$.
	\end{claim}
	\begin{claimproof}
		We proceed by induction over~$\budget-\sum_i |R_i|$. If the algorithm computes $f(\R)$ directly and $f(\R)=\budget'$ is finite, then $\budget'=\sum_i|R_i|\leq\budget$ and $\hat{E}$ is in~$\domain$; moreover, by \Cref{clm:SPeak:swap:FPT:exactness}, $\hat{E}$ is admissible for~$\mathcal{R}$ and at swap distance exactly $\budget'$ from $E$. Otherwise,  $f(\R)=\budget'$ is the finite value returned by some branch with state $\mathcal{R}'$, where $\mathcal{R}'$ arises from $\mathcal{R}$ by adding one pair to some~$R_i$. By the induction hypothesis, we get that $\budget'\leq\budget$ and there is an election in~$\domain$ admissible for~$\mathcal{R}'$ at swap distance exactly $\budget'$ from $E$; since $R_i\subseteq R'_i$ for every $i$, every election admissible for $\mathcal{R}'$ is admissible for $\mathcal{R}$, which completes the proof.
	\end{claimproof}

	\begin{claim}
		\label{clm:SPeak:swap:FPT:completeness}
		If there is an election $E'$ in~$\domain$ admissible for~$\mathcal{R}$ at swap distance $\budget'\leq\budget$ from $E$, then the algorithms set $f(\mathcal{R})$ to a value at most~$\budget'$.
	\end{claim}
	\begin{claimproof}
		We proceed by induction over~$d'-\sum_i |R_i|$, and denote the votes of~$E'$ by $v'_1,\ldots,v'_{\numVoters}$. Since $E'$ is admissible for~$\R$, we have $\sum_i|R_i|\leq\sum_i\dist(v_i,v'_i)=\budget'\leq\budget$, so the algorithm does not set $f(\R)$ as~$ \infty$ in the initial pruning step. Suppose first that some~$T_i$ contains a directed triangle. No linear order realizes all three arcs of a directed triangle, and $v'_i$ realizes all arcs of~$R_i$; hence $v'_i$ inverts, relative to~$v_i$, at least one pair of the triangle outside~$R_i$. In particular, not all three arcs belong to~$R_i$, and $E'$ is admissible for the state~$\mathcal{R}'$ of the corresponding branch. By the induction hypothesis, the algorithm computes a value $f(\R') \leq \budget$ in that branch, which implies $f(\R) \leq f(\R') \leq \budget$. 
        Suppose next that all tournaments are acyclic. 
        If the algorithm finds that~$\hat{E}$ belongs to~$\domain$, then it returns the value $\sum_i|R_i| \leq \budget$.
        Otherwise, assume that it finds a forbidden partial subelection $(\forbCandidates,\forbVoters) \in \F$ present in~$\hat{E}$;
        let $\pi:\candidates^\star \to \candidates$ and $\varphi:\voters^\star \to \voters$ be the injections witnessing this. Since  $E' \in \domain$, we know that $(\forbCandidates,\forbVoters)$ is not present in~$E'$. Consequently, there exists some~$v \in \forbVoters$ and a pair $(a,b) \in P_v$ for which the voter $\varphi(v)=\hat{v}_i$ in~$\hat{E}$ and the corresponding voter~$v'_i$ in~$E'$ ranks $\pi(a)$ and~$\pi(b)$ differently. 
        By \Cref{clm:SPeak:swap:FPT:exactness}, $v'_i$ agrees with $\hat{v}_i$ on all pairs in~$R_i$, so~$\{\pi(a),\pi(b)\} \notin R_i$, and $\hat{v}_i$ orders $\{\pi(a),\pi(b)\}$ as $v_i$ does. Hence, $v'_i$ inverts $\{\pi(a),\pi(b)\}$ relative to~$v_i$ and so $E'$ is admissible for the state of the branch that adds~$\{\pi(a),\pi(b)\}$ to~$R_i$. By the induction hypothesis, this branch returns a value at most~$\budget'$, which ensures that the algorithm sets $f(\R)$ to a value at most~$\budget'$. 
	\end{claimproof}

	Combining \Cref{clm:SPeak:swap:FPT:soundness,clm:SPeak:swap:FPT:completeness}, the algorithm sets $f(\mathcal{R})$ as the minimum swap distance between~$E$ and an election in~$\domain$ admissible for $\mathcal{R}$ if this distance is at most $\budget$, and $\infty$ otherwise. In particular, $f(\R_0)$ is a finite value if and only if some election in~$\domain$ lies within swap distance~$\budget$ of~$E$, and in that case it returns the minimum such distance.
    }%
    There are at most two branches in each triangle-resolution branching and at most $\sum_{v\in \forbVoters} |P_v|$ branches in each branching that results from finding a forbidden partial subelection in~$\hat{E}$; therefore, the number of branches in such a case is at most the size of this partial subelection and, hence, 
    at most~$\maxProfileSize$. Each branch adds a new committed inversion, and states with $\sum_i|R_i|>\budget$ are pruned, so the depth of the search tree is at most $\budget$. Hence, starting from $\R_0=(\emptyset,\ldots,\emptyset)$, we obtain a search tree of size $\Oh{\max(2,\maxProfileSize)^{\budget}}$, and in each state we spend $|\F| \cdot  \numCandidates^{\maxForbCand} \cdot \mathcal{O}(\numCandidates  \numVoters)$ time for building the tournaments, detecting directed triangles, and finding a forbidden partial subelection (or detecting that $\hat{E}$ belongs to~$\domain$). This proves the theorem.%
\end{proofsketch}

It is known that both the $\SPeak$ and the $\GS$ domains can be characterized by a family of forbidden partial subelections, each of size~$6$~\cite{bal-har:j:characterization-single-peaked}. 
Moreover, the $\GScat$ domain is characterized by a family of forbidden partial subelections, one of size~$8$ and the rest of size~$6$~\cite{kra-elk:c:explaining-prefs}.  
The $\SCros$ domain is characterized by two forbidden partial subelections, one of size~$8$ and the other of size~$9$~\cite{bre-che-woe:j:single-crossing}. Hence,
we have the following corollary.

\begin{corollary}
\label{cor:forbidden-stuff}
\label{thm:NE:swap:singlpeaked:FPT:budget}
\label{thm:NE:swap:singlcrossing:FPT:budget}
\label{thm:NE:swap:GS:FPT:budget}
\label{thm:NE:swap:GScat:FPT:budget}
     \swapNE{\mathcal{X}} can be solved in 
     \begin{itemize}
         \item $6^\budget  \cdot (\numCandidates \numVoters)^\Oh{1}$ time if $\mathcal{X} \in \{\SPeak,\GS\}$;
         \item $8^\budget \cdot (\numCandidates \numVoters)^\Oh{1}$ time if $\mathcal{X}=\GScat$;
         \item $9^\budget \cdot (\numCandidates \numVoters)^\Oh{1}$ time if $\mathcal{X}=\SCros$.
     \end{itemize}
\end{corollary}
We remark that we also developed an algorithm tailored directly to the $\GS$ domain; see 
\iflong
\Cref{app:GSbin}.
\fi

\toappendix{

\subsection{An alternative algorithm for \probName{Nearly GS Election}}
\label{app:GSbin}
\begin{theorem}
    \label{thm:NE:swap:binary:FPT:budget}
    \swapNE{\GS} can be solved in $\mathcal{O}^*( 2^{3\budget^2+\mathcal{O}(\budget \log \budget) })$ time.
\end{theorem}
\begin{proof}
Let $E=(\candidates,\voters)$ be our input election, and $\budget$ our budget.
We say that $(C_1,C_2)$ is a \emph{clean split} for~$E$ if in each vote~$\vote \in \voters$, either all candidates in~$C_1$ precede all those in~$C_2$, or vice versa, all candidates in~$C_2$ precede those in~$C_1$.
We start with a claim that allows us to divide an instance into two smaller instances whenever the election admits a clean split.
\begin{claim}
\label{clm:clean-split}
    Suppose that $(C_1,C_2)$ is a partition of~$\candidates$ into two non-empty sets that is a clean split for~$E$. Let $\budget_0=\dist(E,\GS)$, $\budget_1=\dist(\restr{E}{C_1},\GS)$, and $\budget_2=\dist(\restr{E}{C_2},\GS)$ be the swap distances of~$E$, of~$\restr{E}{C_1}$, and of~$\restr{E}{C_2}$, respectively, from the family of $\GS$ domains. Then $\budget_0=\budget_1+\budget_2$.
\end{claim}
\begin{claimproof}
    For two votes~$u$ and~$w$ over the same candidate set, let $\Inversions(u,w)$ denote the set of candidate pairs that are ordered differently in~$u$ and in~$w$; recall that $\dist(u,w)=|\Inversions(u,w)|$.
    We use the following two observations.
    First, for each subset~$D$ of candidates,
    \begin{equation}
    \label{eq:clean-split:restriction}
        \Inversions(\restr{u}{D},\restr{w}{D})=\Inversions(u,w) \cap \binom{D}{2},
    \end{equation}
    since restricting both votes to~$D$ retains exactly those pairs having both of their candidates in~$D$.
    Second, given a binary \typeQ-tree~$T$ with a mapping~$\mapping$ and a subset~$D$ of candidates, let $\restr{T}{D}$ denote the binary \typeQ-tree obtained from~$T$ by deleting all leaves in~$\mapping(\candidates \setminus D)$ and then contracting every edge that connects a node to a unique child. Then
    \begin{equation}
    \label{eq:clean-split:heredity}
        \restr{u}{D} \in \domain(\restr{T}{D},\restr{\mapping}{D}) \quad \textrm{ for each vote } u \in \domain(T,\mapping),
    \end{equation}
    because the permutations of the children of the internal nodes of~$T$ that witness $u \in \domain(T,\mapping)$ induce such permutations in~$\restr{T}{D}$ as well.

    We first prove $\budget_1+\budget_2 \leq \budget_0$; note that this direction holds for an arbitrary partition~$(C_1,C_2)$ of~$\candidates$.
    Let $T$ be a binary \typeQ-tree with a mapping~$\mapping$ such that $\dist(E,\domain(T,\mapping))=\budget_0$, and for each vote~$\vote$ of~$E$ let $u_\vote$ be a vote in~$\domain(T,\mapping)$ closest to~$\vote$.
    Since the pairs within~$C_1$ and the pairs within~$C_2$ form two disjoint subsets of~$\Inversions(\vote,u_\vote)$, we obtain from~\eqref{eq:clean-split:restriction} that
    \[
        \dist(\vote,u_\vote) \geq \dist(\restr{\vote}{C_1},\restr{{u_\vote}}{C_1})+\dist(\restr{\vote}{C_2},\restr{{u_\vote}}{C_2}).
    \]
    By~\eqref{eq:clean-split:heredity} we know $\restr{{u_\vote}}{C_i} \in \domain(\restr{T}{C_i},\restr{\mapping}{C_i})$ for both $i \in [2]$, so summing the above inequality over all votes of~$E$ yields
    \begin{multline*}
        \budget_0 = \dist(E,\domain(T,\mapping)) \\\geq
        \sum_{i \in [2]} \dist(\restr{E}{C_i},\domain(\restr{T}{C_i},\restr{\mapping}{C_i}))
        \geq \budget_1+\budget_2,
    \end{multline*}
    where the last inequality holds because $\budget_i$ is the smallest distance of~$\restr{E}{C_i}$ from a $\GS$ domain.

    We now prove $\budget_0 \leq \budget_1+\budget_2$; this is the direction where we use that $(C_1,C_2)$ is a clean split.
    For each $i \in [2]$, let $T_i$ be a binary \typeQ-tree with a mapping~$\mapping_i$ from~$C_i$ to the leaves of~$T_i$ such that $\dist(\restr{E}{C_i},\domain(T_i,\mapping_i))=\budget_i$.
    Let $T$ be the binary \typeQ-tree whose root has $T_1$ and~$T_2$ as the subtrees rooted at its two children, and let $\mapping=\mapping_1 \cup \mapping_2$.
    Observe that a vote~$u$ belongs to~$\domain(T,\mapping)$ whenever the candidates of~$C_1$ and those of~$C_2$ form two contiguous blocks in~$u$ and $\restr{u}{C_i} \in \domain(T_i,\mapping_i)$ for both $i \in [2]$, since the permutation of the children of the root of~$T$ allows for both orders of these two blocks.

    Consider now a vote~$\vote$ of~$E$, and for both $i \in [2]$ let $u_i$ be a vote in~$\domain(T_i,\mapping_i)$ closest to~$\restr{\vote}{C_i}$.
    As $(C_1,C_2)$ is a clean split, $\vote$ is the concatenation of the blocks $\restr{\vote}{C_1}$ and~$\restr{\vote}{C_2}$ in one of the two possible orders; let $u$ be the concatenation of $u_1$ and~$u_2$ in the same order.
    By the observation above, $u \in \domain(T,\mapping)$. Moreover, $\vote$ and~$u$ order each pair of candidates coming from different sets~$C_1$ and~$C_2$ in the same way, so \eqref{eq:clean-split:restriction} yields
    \[
        \dist(\vote,u)=\dist(\restr{\vote}{C_1},u_1)+\dist(\restr{\vote}{C_2},u_2).
    \]
    Summing over all votes of~$E$, we get \[\budget_0 \leq \dist(E,\domain(T,\mapping)) \leq \budget_1+\budget_2,\] 
    which finishes the proof.
\end{claimproof}

Our algorithm starts by trying to identify a clean split for~$E$. To this end, we fix a vote~$\vote \in \voters$, and we check for each ${k \in [\numCandidates-1]}$ whether the candidates ranked within the top $k$ positions in~$\vote$ form a clean split (with the set of the remaining candidates) for~$E$; this can be checked in time $\Oh{\numCandidates \cdot \numVoters}$. If so, we divide our instance along this split into two sub-instances using \Cref{clm:clean-split}, and apply recursion. Henceforth, we assume that $E$ admits no clean split.

Consider now an election~$E'$ in some $\GS$ domain with minimum swap distance from~$E$, and let $T$ be a binary \typeQ-tree with a mapping~$\mapping$ from $\candidates$ to the leaves of~$T$ such that $E' \in \domain(T,\mapping)$. Let us fix a vote~$\vote \in \voters$ and a \emph{default ordering} of the leaves of~$T$ (complying with the structure of~$T$)  which results in a vote~$\vote'$ in~$E'$ closest to~$\vote$ among all votes in~$\domain(T,\mapping)$.
Let $S^\star$ denote the set of candidates mapped to leaves in the  subtree of~$T$ rooted at the left child of the root~$r$.
We provide an algorithm that determines a family~$\mathcal{S}$ of candidate sets in FPT time with parameter~$\budget$ such that $S^\star \in \mathcal{S}$ and  $|\mathcal{S}|$ is bounded by a function of~$\budget$.

First note that for a set~$S \subseteq \candidates$,
we can compute in $\Oh{\numCandidates \cdot \numVoters}$ time the number of swaps necessary to ensure that $(S,\candidates \setminus S)$ becomes a clean split in~$E$ (i.e., so that in all votes of~$E$, candidates in~$S$ are placed on the first or the last~$|S|$ positions); we call this the \emph{cost} of the split induced by~$S$  and denote it as $\cost(S)$.
Given a position~$k \in [\numCandidates-1]$, we call a candidate set~$S$ \emph{valid for~$k$}
if $|S|=k$, $\cost(S) \leq \budget$, and at most~$\budget$ swaps suffice to move~$S$ to the first positions in~$\vote$.
Clearly, $S^\star$ is valid for~$|S^\star|$.

\begin{claim}
\label{clm:number-of-valid-sets}
For a fixed $k \in [\numCandidates-1]$, there are at most~$2^{2\budget}$ candidate sets valid for~$k$. All such sets and their costs can be computed in $\mathcal{O}^*(2^{2\budget})$ time.
\end{claim}
\begin{claimproof}
Let $C^\vote_j=\{c \colon c \in \candidates, \rank_\vote(c) \leq j\}$ for each $j \in [\numCandidates]$. %
    If $S$ is valid for~$k$, then we can move all candidates in~$S$ with at most~$\budget$ swaps to the first $k$ positions in~$\vote$.
This implies $C_{k-\budget}^\vote \subseteq S$ and $S \cap (\candidates \setminus C^\vote_{k+\budget})=\emptyset$, yielding at most~$2^{2\budget}$ possibilities for~$S$.

    To compute~$\cost(S)$, we simply need to calculate the number of swaps in each vote that suffices to move $S$ to the first or the last~$k$ positions; this can be done in $\Oh{\numCandidates \cdot \numVoters}$ time.
\end{claimproof}

We say that a position in $[\numCandidates-1]$ is \emph{valid} if there exists a candidate set that is valid for it.

\begin{claim}
\label{clm:number-of-valid-positions}
    If $E$ admits no clean split, then there are at most~$8\budget \cdot 2^{\budget}$ valid positions.
\end{claim}
\begin{claimproof}
Let $\budget'=2\budget$.
Fix some valid position~$k \in [\numCandidates-1]$, and consider the path~$P_k$ in~$T$ leading from the root to the $k$-th leaf~$\kappa$ (counted from the left).
We say that a node~$x$ on this path is \emph{heavy}, if the subtree rooted at~$x$ has more than $2\budget'$ leaves to the left of~$\kappa$ and more than $2\budget'$ leaves to the right of~$\kappa$.
By definition, if a node of~$P_k$ is heavy, then so is its parent.
Let $P'_k$ denote the subpath of~$P_k$ formed by heavy nodes.

We say that a node~$x$ and its child~$y$, both on~$P'_k$, form a \emph{rigid pair} if for each vote~$\hat\vote$ of~$E'$, the children of~$x$ have to be reversed  (compared to their ordering in~$T$) if and only if the children of~$y$ have to be reversed when constructing the vote~$\hat\vote \in \domain(T,\mapping)$.
We claim that all edges on~$P'_k$ form a rigid pair.

Assume for a contradiction that $(x,y)$ is not rigid for some edge~$\{x,y\}$ on~$P'_k$; we choose $x$ to be as close to the root as possible.
Let the subpath of~$P'_k$ between $x$ and the root~$r$ be denoted by~$P_{rx}$.
Let $T_1,T_2,\dots,T_\ell$ denote the trees obtained from~$T$ by deleting the vertices of~$P_{rx}$, where $\ell-1$ is the number of vertices on~$P_{rx}$, and let $T_i$ be the unique tree among these that contains~$y$. For each $j \in [\ell]$, let $A_j$ denote the candidates mapped by~$\mapping$ to the leaves of~$T_j$,
and we assume that the default ordering of~$T$ yields a left-to-right order of the candidates that is a refinement of
$A_1,A_2,\dots,A_\ell$; we call this series the partitioning of~$\candidates$ \emph{induced by the path~$P_{rx}$}.
By our choice of~$x$, all edges on~$P_{rx}$ are rigid. Therefore, in all votes~$\hat\vote$ in~$E'$, the ordering of the candidates is a refinement of either~$A_1,A_2,\dots,A_\ell$ or of~$A_\ell,A_{\ell-1},\dots,A_1$.

Let $S$ be a candidate set valid for~$k$. By definition, it is possible to move $S$ to the first~$k$ positions in~$\vote$ by at most~$\budget$ swaps. Moreover, $\vote'$ can also be obtained from~$\vote$ by at most~$\budget$ swaps. Therefore, it is possible to move~$S$ to the first~$k$ positions in~$\vote'$ using at most~$\budget'=2\budget$ swaps; we refer to this fact as ($\blacklozenge$).
Since $y$ is heavy, this implies that $A_1 \cup \cdots \cup A_{i-1} \subseteq S$ and
$(A_{i+1} \cup \cdots \cup A_{\ell}) \cap S = \emptyset$; we refer by ($\spadesuit$) to this property.

Let $T_i^1$ and~$T_i^2$ denote the subtrees rooted at the left and right child of~$y$ (according to the default ordering of~$T$), with candidates sets~$A_i^1$ and~$A_i^2$ mapped to their leaves, respectively.
Since $(x,y)$ is not rigid, there exists a vote $\hat\vote$ in~$E'$ where the ordering of the candidates is a refinement of $A_1,\dots,A_{i-1},A^2_i,A^1_i,A_{i+1},\dots,A_\ell$ or its reverse; we assume the former, as the latter case is analogous.
Due to ($\spadesuit$) and $|A_i|> 4\budget'$, it is not possible to move the vertices in~$S$ to the last~$k$ positions in~$\vote'$ with at most~$\budget'$ swaps; thus, we need to move~$S$ to the first~$k$ positions in~$\vote'$.
Since $S$ is valid for~$k$ and $y$ is heavy, we have that
\begin{itemize}
    \item  the set~$B_1 \subseteq A_i^1$ of candidates with rank at most~$k-\budget'$ in~$\vote'$ must be in~$S$, and
    \item the set~$B_2 \subseteq A_i^2$ of candidates with rank at least $k+\budget'+1$ in~$\vote'$ is disjoint from~$S$.
\end{itemize}

First, if $\kappa$ is contained in $T_i^2$, then by $|A_i^1 \setminus B_1| \leq \budget'-1$ we know that
at least one candidate from~$B_1$, say~$b_1$, must have rank at least~$k+\budget'+1$ in~$\hat\vote$. This means that $b_1$ cannot be moved to the first~$k$ positions within~$\hat\vote$ using at most~$\budget'$ swaps. However, by $b_1 \in S$, this contradicts~($\blacklozenge$). %
Second, if~$\kappa$ is contained in~$T_i^1$, then
by $|A_i^2 \setminus B_2| \leq \budget'$,
at least one candidate from~$B_2$, say~$b_2$, must have rank at most~$k-\budget'$ in~$\hat\vote$; however, this means that $b_2$ cannot be moved out from the first~$k$ positions within~$\hat\vote$ with at most~$\budget'$ swaps. By $b_2 \notin S$, this again contradicts~($\blacklozenge$).

This proves that each edge on~$P'_k$ forms a rigid pair. Let $A_1,\dots,A_{\ell'}$ denote the partitioning of~$\candidates$ induced by~$P'_k$. Since all edges on~$P'_k$ are rigid, in all votes~$\hat\vote$ in~$E'$, the ordering of the candidates is a refinement of either~$A_1,\dots,A_{\ell'}$ or of~$A_{\ell'},\dots,A_1$. Thus, $(A_1 \cup \cdots \cup A_i,A_{i+1} \cup \cdots \cup A_{\ell'})$ forms a clean split for each $i \in [\ell'-1]$ for~$E'$. However, since $E$ has swap distance at most~$\budget$ from~$E'$, and by applying $\budget$ swaps one can increase the cost of at most~$\budget$ of these splits, we obtain that $\ell'-1 \leq \budget$, as otherwise $E$ would admit at least one clean split, a contradiction.
This implies that the node on~$P'_k$ farthest from~$r$ is at most $\budget$ edges away from~$r$ within~$T$; we call this node the \emph{guard} for~$k$, and denote it as~$g(k)$. Let $z$ be the non-heavy child of~$g(k)$ that contains~$\kappa$. Since $z$ is not heavy, we know that $\kappa$ is among the first or last $2\budget'=4\budget$ leaves within the subtree of~$T$ rooted~$g(k)$.

Therefore, for each node~$t$ in~$T$, there can be at most~$8\budget$ valid positions whose guard is~$t$. Moreover, $t$ can only be a guard if it is at most~$\budget$ edges away from the root, implying that there are most~$2^\budget$ guards. Thus, there are at most~$8\budget \cdot 2^\budget$ valid positions in total.
\end{claimproof}

\Cref{clm:number-of-valid-sets,clm:number-of-valid-positions} offer a straightforward approach to construct the desired  family~$\mathcal{S}$ of candidate sets: for each position~$k$, we compute the set of all valid candidate sets. By \Cref{clm:number-of-valid-sets}, this can be done in~$\mathcal{O}^*(2^{2\budget})$ time %
and by \Cref{clm:number-of-valid-sets}, the size of the resulting family $\mathcal{S}$ is  in~$\mathcal{O}^*(8\budget \cdot 2^{3\budget})$. %
Therefore, we can use a classic bounded-search tree approach: if $E$ admits no clean split, we branch by guessing the set~$S^\star$ from among all sets in~$\mathcal{S}$, and in each branch separate our instance to two sub-instances restricted to~$S^\star$ and to~$\candidates \setminus S^\star$. We solve the two obtained sub-instances recursively; note that the parameter decreases by $\cost(S^\star)>0$ in both sub-instances. The recursion stops trivially on instances with only one candidate or with a negative budget. Since the height of the resulting search tree is at most~$d$, the total running time is $\mathcal{O}^*\left((8\budget \cdot 2^{3\budget})^\budget\right)=\mathcal{O}^*( 2^{3\budget^2+\mathcal{O}(\budget \log \budget) })$.
\end{proof}
}

We also provide a specialized algorithm for $\GScat$, which relies on the
observation that in a $\GScat$ election there is a candidate that is always
ranked first or last. The algorithm guesses this candidate, moves him or her
to the closer end of each vote, removes it from consideration and repeats the
process. A similar in spirit approach also works for $\GSbal$, sidestepping
the fact that it is not hereditary.

\begin{restatable}[\linkproof{thm:NE:swap:GSbal:GScat:FPT:budget}]{theorem}{thmNEswapGSbalGScatFPTbudget}
\label{thm:NE:swap:GSbal:GScat:FPT:budget}
    \swapNE{{\GScat}} can be solved in $2^{\Oh{\budget\cdot\log\budget}}\cdot(\numCandidates \numVoters)^{\Oh{1}}$ time, and 
    \swapNE{\GSbal} can be solved in $2^{\Oh{\budget^2}}\cdot(\numCandidates\numVoters)^{\Oh{1}}$ time. 
\end{restatable}

\prooftoappendixdivided{thm:NE:swap:GSbal:GScat:FPT:budget}{\thmNEswapGSbalGScatFPTbudget*}{
We prove the two results separately, in \Cref{thm:NE:swap:GScat:FPT:budget:only,thm:NE:swap:GSbal:FPT:budget:practical}.

We start with our algorithm for \swapNE{\GScat}, which uses a simpler branching strategy than our algorithm for \swapNE{\GScat}, thanks to the following fact: if~$x$ is an inner node of a caterpillar tree and $C_x$ denotes the set  of candidates mapped to the leaves of the subtree rooted at~$x$, then the corresponding partitioning of~$C_x$ (formed by the candidate sets mapped to the leaves of the subtree rooted at the left and the right child of~$x$) always consists of a singleton~$\{c\}$ and the remainder $C_x \setminus \{c\}$ for some candidate~$c \in C_x$.
\begin{theorem}   
\label{thm:NE:swap:GScat:FPT:budget:only}
    \swapNE{\GScat} can be solved in $2^{\Oh{\budget\cdot\log\budget}}\cdot(\numCandidates \numVoters)^{\Oh{1}}$ time. 
\end{theorem}
\begin{proof}
	The algorithm is a bounded-search-tree exploration of the search space. We start with an election $E=(\candidates,\voters)$ and the caterpillar \typeQ-tree $T$ with $\numCandidates$ leaves $\ell_1,\ldots,\ell_\numCandidates$, with the mapping of the candidate set to the leaves not yet fixed. In every step, we fix the candidate mapped to at least one leaf; every branching step additionally decreases the remaining budget by at least one. The algorithm relies on the two following observations:
	\begin{enumerate}[left=-10pt,label=\texttt{(\roman*)}]
		\item If a candidate $c$ is ranked first or last by every voter, then some optimal solution places $c$ on the left-most leaf, and we may fix it there and remove it from the instance.\label{prop:GScat:swap:FPT:optimalStep}
		\item Given an election $(\candidates',\voters)$, we can check in polynomial time whether it is caterpillar-consistent.\label{prop:GScat:swap:FPT:check}
	\end{enumerate}
	The property in~\ref{prop:GScat:swap:FPT:optimalStep} can also be verified in polynomial time; we call a candidate $c$ satisfying it an \emph{ideal candidate}.

	To prove~\ref{prop:GScat:swap:FPT:optimalStep}, let $c$ be an ideal candidate and let $E''$ be a caterpillar-consistent election that can be obtained from~$E$ using $d$ swaps. We construct a caterpillar-consistent election~$E'''$ that can be obtained from~$E$ using at most~$d$ swaps and whose mapping places~$c$ on the left-most leaf. In each vote of~$E'''$, we place~$c$ at the extreme position (first or last) that it occupies in the corresponding vote of~$E$, and we order the remaining candidates as in the corresponding vote of~$E''$. Since the swap distance of two votes is the number of candidate pairs they order differently, and since each vote of~$E'''$ agrees with the corresponding vote of~$E$ on every pair containing~$c$ and with the corresponding vote of~$E''$ on every other pair, the election $E'''$ can indeed be obtained from~$E$ using at most~$d$ swaps. Moreover, $E'''$ is caterpillar-consistent: as the $\GScat$ domain is hereditary~\cite{kra-elk:c:explaining-prefs}, the election $\restr{E''}{\candidates\setminus\{c\}}=\restr{E'''}{\candidates\setminus\{c\}}$ is caterpillar-consistent with respect to some mapping~$\mapping'$, and since every vote of~$E'''$ ranks~$c$ first or last, by the recursive characterization $E'''$ is caterpillar-consistent with respect to the mapping that places~$c$ on the left-most leaf, followed by~$\mapping'$. Hence some optimal solution places~$c$ on the left-most leaf and never swaps~$c$. Applying this construction to an optimal~$E''$ further shows that the optimal number of swaps for~$E$ equals that for~$\restr{E}{\candidates\setminus\{c\}}$: the solution~$E'''$ restricted to $\candidates\setminus\{c\}$ costs the same number of swaps, and conversely, re-inserting~$c$ at its original extreme positions into any caterpillar-consistent election obtained from~$\restr{E}{\candidates\setminus\{c\}}$ costs no additional swaps and, by the argument above, preserves caterpillar-consistency. We may therefore fix~$c$ on the left-most leaf and remove it from the instance.

	A \emph{state} of the algorithm is a pair $(\candidates',\budget')$, where ${\candidates'\subseteq\candidates}$ and $\budget'\leq\budget$ is the remaining number of swaps. Each state returns the minimum number of swaps necessary to make the election $\restr{E}{\candidates'}$ caterpillar-consistent if this number is at most~$\budget'$, and $\infty$ otherwise. The whole instance is a yes-instance if and only if the initial state $(\candidates,\budget)$ returns a finite value.

	Assume the algorithm is in a state $(\candidates',\budget')$. If $\restr{E}{\candidates'}$ is caterpillar-consistent, we return $0$, as no swap is necessary. Otherwise, as long as $\candidates'$ contains an ideal candidate, we fix it on the left-most free leaf and remove it from $\candidates'$, which by~\ref{prop:GScat:swap:FPT:optimalStep} preserves optimality and costs no swaps. After all ideal candidates have been removed, every remaining candidate requires at least one swap to become ideal.

	We now create the branching states. We fix a single voter~$v_1$ and let $P\subseteq\candidates'$ be the set of candidates that $v_1$ ranks among its top $\budget'+1$ or bottom $\budget'+1$ positions; then $|P|\leq 2\budget'+2$. No other candidate can occupy the left-most leaf of a caterpillar tree representing an election that has swap distance at most~$\budget'$ from~$\restr{E}{\candidates'}$: such a candidate would have to become first or last in~$v_1$, which alone would already require more than $\budget'$ swaps, as more than~$\budget'$ candidates separate it from either end of~$v_1$. For each candidate $c\in P$ we compute the number $\sigma_c$ of swaps required to make $c$ ideal, i.e., $\sigma_c=\sum_{\vote\in\voters}\min\{\rank_{\restr{\vote}{\candidates'}}(c)-1,\ |\candidates'|-\rank_{\restr{\vote}{\candidates'}}(c)\}$, moving~$c$ to its cheaper extreme in each vote independently; we discard $c$ if $\sigma_c>\budget'$, and let $P'=\{c\in P:\sigma_c\leq\budget'\}$. If $P'=\emptyset$, we return $\infty$: in any caterpillar-consistent election within swap distance~$\budget'$ of~$\restr{E}{\candidates'}$, the candidate~$c$ on the left-most leaf is first or last in every vote, so already the swaps of pairs containing~$c$ certify $\sigma_c\leq\budget'$; hence no caterpillar election lies within swap distance $\budget'$ of this sub-election. Otherwise we branch over the candidates $c\in P'$: for each, we make $c$ ideal at cost $\sigma_c$, fix it on the left-most leaf, and recurse on $(\candidates'\setminus\{c\},\budget'-\sigma_c)$; note that the swaps making~$c$ ideal only affect pairs containing~$c$, so the remaining sub-election is again the restriction $\restr{E}{\candidates'\setminus\{c\}}$ of the original election. We return the minimum over $c\in P'$ of $\sigma_c$ plus the value returned from the corresponding subproblem. This is correct since the swap distance decomposes over candidate pairs: any solution placing~$c$ on the left-most leaf spends at least~$\sigma_c$ swaps on pairs containing~$c$ and at least the optimum of $\restr{E}{\candidates'\setminus\{c\}}$ on the remaining pairs, and conversely, combining the $\sigma_c$ swaps with an optimal solution for $\restr{E}{\candidates'\setminus\{c\}}$ yields a caterpillar-consistent election, as in~\ref{prop:GScat:swap:FPT:optimalStep}.

	There are at most $2\budget'+2$ branches in each state, and each recursive call decreases $\budget'$ by at least one, as every candidate in~$P'$ requires at least one swap to become ideal. Hence, starting from $(\candidates,\budget)$, we obtain a search tree of size $(2\budget+2)^{\budget}=\budget^{\Oh{\budget}}$, and in each state we spend $(\numCandidates \numVoters )^{\Oh{1}}$ time for the checks of~\ref{prop:GScat:swap:FPT:optimalStep} and~\ref{prop:GScat:swap:FPT:check} and for computing the values $\sigma_c$. Therefore the algorithm runs in $2^{\Oh{\budget\cdot\log\budget}}\cdot(\numCandidates\numVoters)^{\Oh{1}}$ time.%
\end{proof}

Next, we apply the same technique to solve \swapNE{\GSbal}; the branching step here will be more involved.  

\begin{theorem}   
\label{thm:NE:swap:GSbal:FPT:budget:practical}
    \swapNE{\GSbal} can be solved in $2^{\Oh{\budget^2}}\cdot(\numCandidates\numVoters)^{\Oh{1}}$ time. 
\end{theorem}
\begin{proof}
    The algorithm uses a bounded-search-tree strategy.
    Let $E=(\candidates,\voters)$ and $\budget$ be our input, and recall that for an election to belong to $\GSbal$ we need $\numCandidates=|\candidates|=2^{\treeheight}$ for some integer~$\treeheight$; otherwise we reject immediately.
    Throughout, $C'$ denotes a subset of~$\candidates$ whose size is a power of~$2$, and we write $\opt(C')=\dist(\restr{E}{C'},\GSbal)$.
    We call a partitioning~$(L,R)$ of~$C'$ with $|L|=|R|=\nicefrac{|C'|}{2}$ a \emph{split of~$C'$}, and we let $\splitcost(L,R)$ denote the number of swaps needed to ensure that each voter ranks either all of~$L$ above all of~$R$, or vice versa; that is,
    \[
        \splitcost(L,R)=\sum_{\vote \in \voters} \min\{x_\vote,\ |L| \cdot |R|-x_\vote\},
    \]
    where $x_\vote=|\{(l,r) \in L \times R \colon r \succ_\vote l\}|$.
    A split with $\splitcost(L,R)=0$ is a \emph{clean split} for~$\restr{E}{C'}$. Note that $\splitcost(L,R)$ is computable in $\Oh{\numCandidates \numVoters}$ time.

    \proofsubparagraph{Two observations.}
    We first record the recursion underlying the algorithm.

    \begin{claim}
    \label{clm:gsbal:decomposition}
        For each~$C'$ with $|C'| \geq 2$ we have $\opt(C')=\min_{(L,R)} \bigl( \splitcost(L,R)+\opt(L)+\opt(R) \bigr)$, where the minimum is taken over all splits~$(L,R)$ of~$C'$.
    \end{claim}
    \begin{claimproof}
        Let $T$ be a balanced \typeQ-tree with a mapping~$\mapping$ from~$C'$ to its leaves, and let $(L,R)$ be the sets of candidates mapped to the leaves of the two subtrees rooted at the children of the root of~$T$. Every vote of~$\domain(T,\mapping)$ ranks the candidates of~$L$ and those of~$R$ in two contiguous blocks, in one of the two possible orders, and its restrictions to~$L$ and to~$R$ range independently over the domains defined by the two subtrees. Since the swap distance of two votes is the number of candidate pairs they order differently, and since each candidate pair either belongs to $L \times R$ or lies within~$L$ or within~$R$, we get $\dist(\restr{E}{C'},\domain(T,\mapping))=\splitcost(L,R)+\dist(\restr{E}{L},\domain(T^L,\restr{\mapping}{L}))+\dist(\restr{E}{R},\domain(T^R,\restr{\mapping}{R}))$ for the two subtrees $T^L$ and~$T^R$. Taking the minimum over all balanced \typeQ-trees and mappings, and noting that each split~$(L,R)$ of~$C'$ arises from some such tree, proves the claim.
    \end{claimproof}

    \begin{claim}
    \label{clm:gsbal:unique-clean-split}
        The election $\restr{E}{C'}$ has at most one clean split, and it can be identified in $\Oh{\numCandidates\numVoters}$ time.
    \end{claim}
    \begin{claimproof}
        Let $(L,R)$ be a clean split for~$\restr{E}{C'}$ and let $\vote$ be any vote of~$\voters$. Then $L$ consists of the candidates that $\vote$ ranks within the top $\nicefrac{|C'|}{2}$ positions of~$C'$, or of those it ranks within the bottom $\nicefrac{|C'|}{2}$ positions; thus $\{L,R\}$ is determined by~$\vote$. To find it, we compute this partitioning for one fixed vote and check in $\Oh{\numCandidates\numVoters}$ time whether it is clean.
    \end{claimproof}

    \proofsubparagraph{The skeleton.}
    By \Cref{clm:gsbal:unique-clean-split} we may define, for each~$C'$, the following rooted tree, the \emph{skeleton} of~$\restr{E}{C'}$: its root is~$C'$; a node~$D$ with $|D| \geq 2$ is a leaf and is called \emph{blocked} if $\restr{E}{D}$ has no clean split, and otherwise its two children are the two sides of the unique clean split of~$\restr{E}{D}$; nodes of size~$1$ are leaves as well. Observe that the blocked nodes are pairwise disjoint, and that the skeleton has depth at most $\log \numCandidates$.
    We call a node of the skeleton \emph{active} if the subtree of the skeleton rooted at it contains a blocked node; the active nodes are thus exactly the ancestors of blocked nodes, including the blocked nodes themselves. Since each blocked node has at most $\log \numCandidates+1$ ancestors, including itself, a skeleton with $b$ blocked nodes has at most $b\,(\log \numCandidates+1)$ active nodes.
    Note also that if the skeleton of~$\restr{E}{D}$ contains no blocked node, then it is a perfect binary tree whose leaves are the singletons of~$D$; regarding it as a balanced \typeQ-tree~$T$ and letting $\mapping$ map each candidate to its own leaf, every vote of~$\restr{E}{D}$ lies in~$\domain(T,\mapping)$ by the definition of a clean split, and hence $\opt(D)=0$.

    \proofsubparagraph{Few blocked nodes.}
    The following lemma is the heart of the running-time analysis: it states that a bounded budget permits only a bounded number of blocked nodes, no matter how large~$\numCandidates$ is.

    \begin{claim}
    \label{clm:gsbal:few-blocked}
        The skeleton of~$\restr{E}{C'}$ contains at most $2\,\opt(C')$ blocked nodes.
    \end{claim}
    \begin{claimproof}
        It suffices to prove the statement for $C'=\candidates$: the general case follows by applying it to the election~$\restr{E}{C'}$ in place of~$E$. We may further assume $\voters \neq \emptyset$, as otherwise every split is clean and no node is blocked. Let $F$ be an election in a $\GSbal$ domain $\domain(T,\mapping)$ with $\dist(E,F)=\opt(\candidates)$, and for each vote~$\vote$ of~$E$ let $\vote^F$ be the corresponding vote of~$F$. We say that a candidate pair \emph{touches} a set~$D$ if at least one of its two candidates lies in~$D$, and we call it \emph{broken} if some vote~$\vote$ of~$E$ orders it differently than~$\vote^F$ does. The number of broken pairs is at most $\dist(E,F)=\opt(\candidates)$.
        We show that every blocked node~$D$ is touched by a broken pair; since the blocked nodes are pairwise disjoint and each pair touches at most two of them, this proves the claim.

        We first observe that in each vote~$\vote$ of~$E$, the candidates of a skeleton node~$D$ with $|D|=2^{j}$ occupy the positions $a2^{j}+1,a2^{j}+2,\dots,(a+1)2^{j}$ for some integer~$a \geq 0$; we call such a set of positions a \emph{block}. This holds for the root $\candidates$, and if it holds for a node, then its two children each occupy the upper or the lower half of its block, which is again a block.

        Now suppose for contradiction that no broken pair touches the blocked node~$D$, that is, $\vote$ and~$\vote^F$ order every pair touching~$D$ in the same way, for every vote~$\vote$. Then each candidate of~$D$ is preceded by the same candidates in~$\vote$ as in~$\vote^F$, so the candidates of~$D$ occupy the same positions in~$\vote$ and in~$\vote^F$, and $\restr{E}{D}$ and $\restr{F}{D}$ are the same election. By the previous paragraph, the candidates of~$D$ therefore occupy a block in each vote of~$F$.
        Next, we claim that in any vote $u \in \domain(T,\mapping)$, the candidates occupying a block of size~$2^{j}$ are exactly those mapped by~$\mapping$ to the leaves of the subtree rooted at some node of~$T$ at distance $\treeheight-j$ from its root. Indeed, a vote of~$\domain(T,\mapping)$ arises by permuting the children of the nodes of~$T$ and reading the leaves from left to right; as $T$ is a perfect binary tree, the leaves of the subtrees rooted at the nodes at distance $\treeheight-j$ from the root occupy consecutive blocks of size~$2^{j}$ in this reading, in some order.
        Since the candidate sets belonging to the nodes of~$T$ at a fixed distance from the root are pairwise disjoint, and $D$ coincides with one of them, as witnessed by any single vote of~$F$, the set~$D$ is the candidate set of a single node~$t^\star$ of~$T$. Consequently $\restr{F}{D}$ lies in the domain defined by the subtree of~$T$ rooted at~$t^\star$, which is a perfect binary \typeQ-tree; in particular, the two candidate sets belonging to the children of~$t^\star$ form a clean split for $\restr{F}{D}=\restr{E}{D}$, contradicting the assumption that $D$ is blocked.
    \end{claimproof}

    \proofsubparagraph{The algorithm.}
    A \emph{state} is a pair $(C',\budget')$ with $\budget' \leq \budget$; it \emph{returns} $\min\{\opt(C'),\budget'+1\}$, so that the input is a yes-instance if and only if the initial state $(\candidates,\budget)$ returns a value of at most~$\budget$. We proceed as follows on a state $(C',\budget')$.
    \begin{enumerate}
        \item If $|C'|=1$, we return~$0$.
        \item We build the skeleton of $\restr{E}{C'}$; by \Cref{clm:gsbal:unique-clean-split} this takes polynomial time. If it contains no blocked node, then $\opt(C')=0$ as observed above, and we return~$0$. If it contains more than $2\budget'$ blocked nodes, then $\opt(C')>\budget'$ by~\Cref{clm:gsbal:few-blocked} and we return $\budget'+1$.
        \item Otherwise, we fix a vote~$\vote \in \voters$ and let $P$ contain those candidates of~$C'$ that $\restr{\vote}{C'}$ ranks within $\budget'$ positions of the boundary between the top and the bottom half of~$C'$; thus $|P| \leq 2\budget'$. Any candidate of $C' \setminus P$ lies more than~$\budget'$ positions away from this boundary, so moving it to the other half would already require more than~$\budget'$ swaps within~$\vote$; hence every split~$(L,R)$ with $\splitcost(L,R) \leq \budget'$ agrees with the split of~$\restr{\vote}{C'}$ into its top and bottom half outside of~$P$. We branch over all at most $\binom{2\budget'}{\budget'}<4^{\budget'}$ splits of this form. For each of them we compute $\splitcost(L,R)$, discard it if $\splitcost(L,R)>\budget'$, and otherwise recurse on the states $(L,\budget'-\splitcost(L,R))$ and $(R,\budget'-\splitcost(L,R))$. We return the minimum of $\splitcost(L,R)+\alpha_L+\alpha_R$ over the surviving splits (i.e., over those splits in which both calls return a value at most $\budget'-\splitcost(L,R)$), where $\alpha_L$ and~$\alpha_R$ are the values returned by the two recursive calls, capped at $\budget'+1$.
    \end{enumerate}
    Correctness is immediate from~\Cref{clm:gsbal:decomposition} together with the observation in step~3 that all splits of cost at most~$\budget'$ are among those we enumerate: the optimum of the recursion in \Cref{clm:gsbal:decomposition} is attained either by such a split, or $\opt(C')>\budget'$ and the returned value exceeds~$\budget'$ as well.

    \proofsubparagraph{Running time.}
    \def\searchtree{\mathcal{T}}
    \def\skeleton{\mathcal{S}}
    Consider the search tree~$\searchtree$ of the algorithm in which every node~$x$ is a call of the algorithm, initiated with some state of the algorithm. 
    We call an edge of $\searchtree$ \emph{free} if it corresponds to a split of cost~$0$, and \emph{paying} otherwise. By \Cref{clm:gsbal:unique-clean-split}, a state has at most one clean split available, so the free edges leaving a node~$x$ of~$\searchtree$ lead to the two sides of the unique clean split of the state corresponding to~$x$. Consequently, the descendants of a node~$x$ with state~$(C',\budget')$ that are reachable from~$x$ using free edges only form a subgraph of the skeleton  of~$\restr{E}{C'}$; let $\skeleton_x$ denote this subgraph of~$\searchtree$. Observe that each node of $\skeleton_x$ that is not a leaf in~$\searchtree$ must be an active node in~$\skeleton_x$, due to step~2. %
    Again by step~2, the skeleton of a state $(\candidates',\budget')$ corresponding to some non-leaf node of~$\searchtree$  contains at most $2\budget'$ blocked nodes and, hence, at most~$2\budget'(\log \numCandidates+1)$ active nodes; therefore, the number of nodes in~$\skeleton_x$ is at most $2\budget'(\log \numCandidates+1)$. %
    For each node~$y$ in~$\skeleton_x$, the children of~$y$ outside~$\skeleton_x$ arise as a consequence of at most $4^{\budget'}$ (non-clean) splits: each split yields two recursive calls, with the budget decreased by at least~$1$ in both. This means that each node~$y$ in~$\skeleton_x$ has at most~$2 \cdot  4^{\budget'}$ children in~$\searchtree$ but outside~$\skeleton_x$. Moreover, $y$ is connected each these children nodes via a paying edge, each of them associated with a state having decreased parameter. Writing $N(\budget')$ for the maximum number of non-leaf nodes in a search tree whose root has initial budget~$\budget'$, we thus get
    \[
        N(\budget') \leq 2\budget'(\log \numCandidates+1) \cdot \bigl(1+2 \cdot 4^{\budget'} N(\budget'-1)\bigr).
    \]
    Using induction on~$\budget$, one can prove that this implies \[N(\budget) \leq \bigl(4\budget(\log \numCandidates+1)4^{\budget}\bigr)^{\budget} = 2^{\Oh{\budget^2}} \cdot (\log \numCandidates)^{\Oh{\budget}}.\]
    
    Finally, we show $(\log \numCandidates)^{\budget} \leq 2^{\budget^2}+\numCandidates^{o(1)}$: if $\log \numCandidates \leq 2^{\budget}$, then $(\log \numCandidates)^\budget \leq 2^{\budget^2}$, and otherwise $\budget<\log\log \numCandidates$ and thus $(\log \numCandidates)^{\budget}<2^{(\log\log \numCandidates)^2}=\numCandidates^{o(1)}$.
    Since we spend polynomial time in each state, the total running time of the algorithm is $2^{\Oh{\budget^2}} \cdot (\numCandidates \numVoters)^{\Oh{1}}$.
\end{proof}
}

For domains with a simpler structure, a more direct approach---using
branching but no recursive decomposition---offers FPT algorithms with a better
dependence on the parameter~$\budget$.
For $\Antag$ (and $\Ident$), which  are hereditary 
\iflong
(see \Cref{thm:ident:antag:hereditary} in
 \Cref{app:sec:param_distance})
\fi
we could have applied
\Cref{thm:NE:swap:domainByForbidden:FPT:budget}, but
the algorithm would be slower.

\begin{restatable}[\linkproof{thm:NE:swap:FPT:budget:collapsed}]{theorem}{thmNEswapFPTbudgetcollapsed}
\label{thm:NE:swap:FPT:budget:collapsed}
     \swapNE{\mathcal{X}} can be solved in
     \begin{itemize}
         \item $3^{\budget} \cdot \Oh{\numCandidates \numVoters}$ time if $\mathcal{X}=\Strat$, for every~$k$;
         \item $2.832^{\budget} \cdot \Oh{\numCandidates \numVoters}$ time if $\calX=\Strat[\numCandidates/2]$;
         \item $5^{\budget} \cdot \Oh{\numCandidates \numVoters}$ time if $\calX=\Separ$;
         \item $3.06^{\budget} \cdot \Oh{\numCandidates \numVoters}$ time if $\calX = \Antag$.
     \end{itemize}
\end{restatable}

\prooftoappendixdivided{thm:NE:swap:FPT:budget:collapsed}{\thmNEswapFPTbudgetcollapsed*}{
We prove the results one by one in the following four theorems.
\begin{theorem}    
\label{thm:NE:swap:Strat:FPT:budget}
     \swapNE{\Strat} can be solved in $3^{\budget} \cdot \Oh{\numCandidates  \numVoters}$ time for any integer~$k$.
\end{theorem}
\begin{proof}
    Let $E=(\candidates,\voters)$ and~$\budget$ be our input instance with $|\candidates|=\numCandidates$. 
    Let $T$ be the tree representing the \Strat[k] domain for~$\numCandidates$ candidates, and let $\mapping$ be a bijection from~$\candidates$ to the leaves of~$T$ such that $E$ has minimum swap distance from~$\domain(T,\mapping)$.

    We start by computing the set~$\candidates_1$ of candidates that are mapped by~$\mapping$ to the children of the first \typeP-node in~$T$. For a vote~${v \in \voters}$, let $F_\vote$ denote the set of candidates appearing on the first $\numCandidates/k$ positions in~$\vote$.
    We will call a candidate~$c$ \emph{dirty} if $c$ is in~$F_\vote$ for some but not all votes~$\vote$. 
    As long as there exists some dirty candidate~$c$, we proceed as follows. First, we guess whether~$c \in C_1$ or not, and find a vote~$\vote$ where $c$ appears in the ``wrong half'' of~$\vote$.
    That is, if $c \in C_1$, we find a vote~$\vote$ where  $c \notin F_\vote$, and we guess the 
    candidate~$b$ closest to~$c$ in~$\vote$ that precedes~$c$ but is not in~$C_1$; such a candidate must exist by $|C_1|=|F_\vote|$. In a correct guess, $\rank_\vote(c)-\rank_\vote(b) \leq \budget$, because $b$ and~$c$ need to be swapped, which requires $\rank_\vote(c)-\rank_\vote(b)$ swaps (as~$b$ needs to be swapped with~$c$ and all candidates between~$b$ and~$c$ in~$\vote$, all belonging to~$C_1$).
    Similarly, if $c \notin C_1$, then  we find a vote~$\vote$ where $c \in F_\vote$, and we guess the candidate~$b$ closest to~$c$ that follows~$c$ and is in~$C_1$; again, swapping~$b$ and~$c$ requires $\rank_\vote(b)-\rank_\vote(c)$ swaps.
    
    In both cases, we perform the necessary $\budget_c=|\rank_\vote(b)-\rank_\vote(c)|$ swaps, decrease the distance budget by~$\budget_c$, and proceed with the next dirty candidate. If there are no dirty candidates left, it must be the case that $C_1=F_{\vote}$ for all votes~$\vote \in \voters$. Hence, we remove $C_1$ from the set of candidates, and proceed with the remaining instance of \swapNE{\Strat[(k-1)]}.

    To compute the running time of this algorithm, let $t^{\numCandidates,\numVoters}(\budget)$ denote the running time on an instance with~$\numCandidates$ candidates, $\numVoters$ voters, and parameter~$\budget$.
    Recall that we first make a binary guess, then branch into at most $\budget$ possibilities for candidate~$b$, decreasing the parameter by~$i$ in the $i$-th branch, for $i \in [\budget]$. Therefore, we have
    \[t^{\numCandidates,\numVoters}(\budget) \leq 
    2\left(
    t^{\numCandidates,\numVoters}(\budget-1)+t^{\numCandidates,\numVoters}(\budget-2)+\cdots+t^{\numCandidates,\numVoters}(0)\right)\]
    where we also know $t^{\numCandidates,\numVoters}(0)=\Oh{\numCandidates\cdot \numVoters}$.
    Using induction on~$\budget$, we can show that 
    $t^{\numCandidates,\numVoters}(\budget) = 3^{\budget} \cdot \Oh{\numCandidates \cdot \numVoters}$, because
    \begin{align*}
    t^{\numCandidates,\numVoters}(\budget) &\leq 
    2(
    3^{\budget-1}+3^{\budget-2}+\cdots+1) \cdot \Oh{\numCandidates \cdot \numVoters}  \\
    &=(3^{\budget}-1) \cdot \Oh{\numCandidates \cdot \numVoters} \leq 3^\budget \cdot \Oh{\numCandidates \cdot \numVoters}. \qedhere
    \end{align*}
\end{proof}

\begin{corollary}    
\label{thm:NE:swap:StratHalf:FPT:budget}
    \swapNE{\Strat[\numCandidates/2]} can be solved in $2.832^{\budget} \cdot \Oh{\numCandidates  \numVoters}$ time.    
\end{corollary}

\begin{proof}
    The same algorithm can be used as in~\Cref{thm:NE:swap:Strat:FPT:budget} for $k=\numCandidates/2$; however, observe that guessing the candidate~$b$ is a binary guess if $c \in C_1$, and a ternary guess if $c \notin C_1$, using that $|C_1|=2$. Thus, we obtain the recurrence  
    \[t^{\numCandidates,\numVoters}(\budget) \leq 
    2
    t^{\numCandidates,\numVoters}(\budget-1)
    +2t^{\numCandidates,\numVoters}(\budget-2)+
    t^{\numCandidates,\numVoters}(\budget-3)
    \]
    which yields $t^{\numCandidates,\numVoters}(\budget) \leq 2.832^\budget \cdot \Oh{\numCandidates \cdot \numVoters}$.
\end{proof}

\begin{theorem}
    \label{thm:NE:swap:Separ:FPT:budget}
    \swapNE{\Separ} can be solved in $5^{\budget} \cdot \Oh{\numCandidates  \numVoters}$ time. 
\end{theorem}

\begin{proof}
    We apply the same ideas as in \Cref{thm:NE:swap:Strat:FPT:budget}, and we re-use the notation therein for~$k=2$, in particular, we let $F_\vote$ denote the candidates in the first $\numCandidates/2$ positions of~$\vote$. 
    Let us say that a candidate pair~$(c,c')$ is \emph{dirty} if there exist votes~$\vote$ and~$\vote'$ such that (i) $c$ and~$c'$ are \emph{in the same half} of~$\vote$, meaning that either~$\{c,c'\} \subseteq F_\vote$ or $\{c,c'\} \cap F_\vote = \emptyset$, but (ii) $c$ and~$c'$ are not in the same half of~$\vote'$.
    Then some candidate in~$\{c,c'\}$ needs to be moved to the ``other half'' either in~$\vote$ or in~$\vote'$; we first guess this vote~$\hat{\vote} \in \{\vote,\vote'\}$. Next, we guess which candidate, $c$ or~$c'$, is the one to be moved in~$\hat{\vote}$; let $\hat{c}$ be this candidate. From this point on, we proceed as in \Cref{thm:NE:swap:Strat:FPT:budget}, guessing the candidate~$b$ closest to~$\hat{c}$ in the appropriate direction from~$\hat{c}$ in~$\hat{\vote}$ that needs to be swapped with~$\hat{c}$. After performing these swaps and reducing the parameter with $|\rank_{\hat{\vote}}(\hat{c})-\rank_{\hat{\vote}}(b)|$, we proceed with the next dirty pair. Once there are no dirty pairs left, we know that the election belongs to the $\Separ$ domain.

    The running time of this algorithm can be computed with the same techniques as in \Cref{thm:NE:swap:Strat:FPT:budget}, using that the running time function satisfies
    \[t^{\numCandidates,\numVoters}(\budget) \leq 
    4\left(
    t^{\numCandidates,\numVoters}(\budget-1)+t^{\numCandidates,\numVoters}(\budget-2)+\cdots+t^{\numCandidates,\numVoters}(0)\right)
    \]
    due to the two binary guesses we start with in each iteration; this leads to $t^{\numCandidates,\numVoters}(\budget) =5^\budget \cdot \Oh{\numCandidates \cdot \numVoters}$.
\end{proof}

\begin{theorem}
    \label{thm:NE:swap:Antag:FPT:budget}
    \swapNE{\Antag} can be solved in $3.06^\budget\cdot\Oh{\numCandidates \numVoters}$ time. 
\end{theorem}

\begin{proof}
    Consider an input election~$E=(\candidates,\voters)$ and distance~$\budget$ with $\numVoters=|\voters|$.
    Let $T$ be the tree with $|\candidates|$ leaves representing the $\Antag$ domain, and $\mapping$ a  bijection from~$\candidates$ to the leaves of~$T$ such that $E$ has minimum swap distance from~$\domain(T,\mapping)$.

    First, if $\numVoters > \budget$, then there must exist a vote~$\vote$ in which no swaps are needed, that is, for which $\vote \in \domain(T,\mapping)$. We guess such a vote~$\vote$, which immediately determines~$\domain(T,\mapping)$, allowing us to  compute the swap distance of~$E$ from~$\domain(T,\mapping)$ in polynomial time.
    
    If $\numVoters \leq \budget$, then we proceed as follows. For each vote~$\vote$, let~$\vote'$ denote the vote in~$\domain(T,\mapping)$ with minimal swap distance to~$\vote$. We guess for each vote~$\vote$ whether $\vote'$ is obtained by reversing the children of the root of~$T$ (a type \typeQ\ node), and if so, we reverse the vote~$\vote$ in~$E$. Assuming correct guesses, the obtained election~$\hat{E}$ has distance at most~$\budget$ from the $\Ident$  domain if and only if $E$ has distance at most~$\budget$ from the $\Antag$ domain. Thus, we can apply the algorithm by \citet{bet-fel-guo-nie-ros:j:fpt-kemeny-aaim} for \KemenyOA which runs in $\mathcal{O}^*(1.53^\budget)$ time. The total running time is therefore $2^{\numVoters} \cdot \mathcal{O}^*(1.53^\budget) \leq 2^\budget \cdot \mathcal{O}^*(1.53^\budget) = \mathcal{O}^*(3.06^\budget)$.
\end{proof}

} %

\toappendix{
\subsection{Identity and Antagonism are hereditary}

\begin{proposition}\label{thm:ident:antag:hereditary}
    Both $\Ident$ and $\Antag$ are hereditary and admit a finite set of forbidden subelections.
\end{proposition}
\begin{proof}
    For $\Ident$, the forbidden subelection $E_{\textup{id}}$ consists of two voters $v$ and $v'$ and two candidates $c$ and $d$ with $v\colon c \succ d$ and $v'\colon d \succ c$. If an election~$E$ contains two distinct votes, then these disagree on some pair of candidates, which means that $E_{\textup{id}}$ is present in~$E$; conversely, an election whose votes all coincide clearly avoids $E_{\textup{id}}$ (meaning that $E_{\textup{id}}$ is not present in such an election).

    For $\Antag$, we forbid three subelections $A_1$, $A_2$, $A_3$, each consisting of two voters $v$ and $v'$ and three candidates $c$, $d$, and $e$, where in all three $v\colon c \succ d \succ e$, and
    \begin{align*}
        A_1&\colon\ v'\colon c \succ e \succ d \\ 
        A_2&\colon\ v'\colon d \succ c \succ e \\
        A_3&\colon\ v'\colon d \succ e \succ c
    \end{align*}
    Two votes that are equal or reversals of each other restrict, on every set of three candidates, to two orders that are again equal or reversals of each other; in each~$A_i$ the two orders are neither. Hence no profile in~$\Antag$ contains any of~$A_1,A_2,A_3$.
    
    Conversely, assume that none of $A_1,A_2,$ or~$A_3$ is present in some election $E$; we show that $E$ belongs to an $\Antag$ domain.
    We  use induction over the size~$\numCandidates$ of the candidate set~$C$; note that the cases $\numCandidates  \leq 2$ are trivial.
    Let $m \geq 3$, and  assume that our claim holds for all elections over candidate sets of size~$m-1$.
    Let $u=(c_1,\dots,c_m)$ and~$u'$ be two votes in~$E$. Since the claim holds for $\restr{E}{C'}$ for $C'=C \setminus \{c_m\}$, we know that $\restr{u'}{C'} \in \{(c_1,\dots,c_{m-1}), (c_{m-1},\dots,c_{1})\}$. We distinguish between two cases: 
    \begin{description}
        \item[Case A:] If $\restr{u'}{C'}=(c_1,\dots,c_{m-1})$.  
        In this case, $u'$ cannot rank $c_m$ before $c_1$, since then $A_3$ would be isomorphic to the restriction of~$E$ to~$\{c_m,c_1,c_2\}$.
        Moreover, $u'$ cannot rank $c_m$ in between two candidates $c_i$ and $c_{i+1}$ for some $i \in [m-2]$, because then $A_1$ would be isomorphic to the restriction of~$E$ to~$\{c_i,c_m,c_{i+1}\}$. Thus, $u'$ must place~$c_m$ after $c_{m-1}$, and thus $u'$ coincides with~$u$.
        \item[Case B:] If $\restr{u'}{C'}=(c_{m-1},\dots,c_{1})$.  
        In this case, $u'$ cannot rank $c_m$ after $c_1$, since then $A_2$ would be isomorphic to the restriction of~$E$ to~$\{c_1,c_2,c_m\}$.
        Moreover, $u'$ cannot rank $c_m$ in between two candidates $c_i$ and $c_{i+1}$ for some $i \in [m-2]$, because then $A_3$ would be isomorphic to the restriction of~$E$ to~$\{c_i,c_{i+1},c_m\}$. Thus, $u'$ must place~$c_m$ before $c_{m-1}$, and thus $u'$ is the reverse of~$u$.
    \end{description}
    This finishes the proof of claim, implying that $E$ indeed belongs to the $\Antag$ domain.
    We remark that the given characterization is minimal in the sense that not forbidding some $A_i$, $i \in [3]$, to be present in an election~$E$ allows~$E$ to be a domain that is not an $\Antag$ domain, as witnessed by the election~$A_i$ itself.
\end{proof}
}

\section{Experiments}
\label{sec:experiments}
\appendixsection{sec:experiments}

We consider three main questions: (a)~How fast are our algorithms in
practice? (b)~How close are various elections to being
(caterpillar/balanced) group-separable? (c)~How likely are closest
$\GS$ elections to retain the original winners?
The latter two questions evaluate similarity of the source and derived
elections, either on the level of individual swaps or on the level of
features relevant for voting rules.

Following the map framework of
\citet{szu-boe-bre-fal-nie-sko-sli-tal:j:map}, we focus on synthetic
data, specifically using the setting of
\citet{fal-kac-sor-szu-was:c:div-agr-pol-map}; see these works, and
the survey of \citet{boe-fal-jan-kac-lis-pie-rey-sto-szu-was:c:guide},
for details on generating elections. Here we provide high-level
intuitions regarding our models:
\begin{itemize}
\item In the \emph{normalized Mallows model} (NM-$\phi$) each vote is
  derived from a fixed central one, by introducing random swaps of
  candidates, driven by the value of the $\phi$
  parameter~\citep{mal:j:mallows,boe-kra-fal:c:normalized-mallows}; $\phi=1$ means
  that, in effect, the resulting vote is chosen uniformly at random
  (this is also known as \emph{Impartial Culture}; IC) and smaller
  values of $\phi$ mean proportionally fewer swaps (for $\phi = 0$ no
  swaps are made and all votes are equal to the central one).
\item In the \emph{P\'olya--Eggenberger urn model} (URN-$\alpha$) we
  have several different-sized clusters of votes. All the votes in a
  single cluster are identical, chosen uniformly at random. The larger
  the $\alpha$ parameter, the fewer clusters there
  are~\citep{berg:j:urn-model,mcc-sli:j:similarity-rules}
\item In the \emph{$t$D-Euclidean model} (EUC-$t$D) candidates and
  voters are points drawn uniformly at random from $[0,1]^t$; the
  voters rank the candidates from the closest to the farthest to them (in terms
  of the Euclidean distance).
\end{itemize}
We also consider elections generated by drawing votes uniformly at
random from the $\GSbin$, $\GScat$, and $\GSbal$ domains, from the
single-peaked~\citep{wal:t:generate-sp,con:j:eliciting-singlepeaked} and single-crossing~\citep{szu-boe-bre-fal-nie-sko-sli-tal:j:map} distributions.

Following \citet{fal-kac-sor-szu-was:c:div-agr-pol-map}, we focus on
the setting with $8$ candidates and $96$ voters. All experiments were
performed on a machine equipped with two $64$-core AMD EPYC 7742 CPUs and $512$~GB of RAM. All implementations are in Python and use the Prefsampling library~\cite{boe-fal-jan-kac-lis-pie-rey-sto-szu-was:c:guide}. We used the Gurobi Optimizer, version 13.0.1.

\begin{figure*}[t]
    \centering

    \begin{subfigure}[t]{0.32\textwidth}
        \centering
        \includegraphics[width=\linewidth]{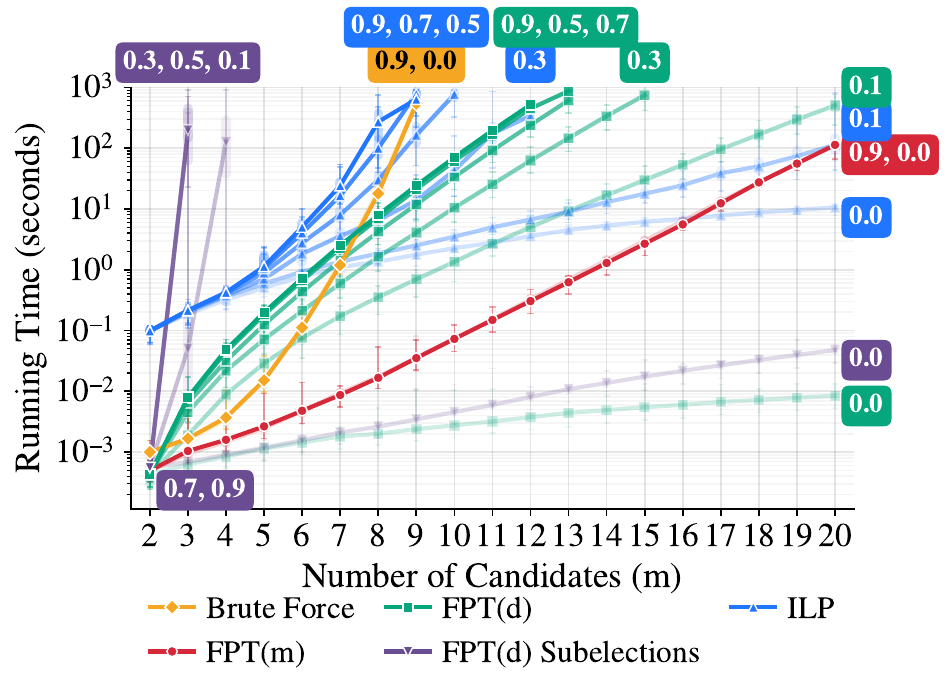}
        \caption{Running times of our algorithms for the
          \swapNE{\GScat} problem on NM-$\phi$ elections with $96$ voters,
          for different numbers of candidates and different values of~$\phi$.}
        \label{fig:results:runningTime:mallows}
    \end{subfigure}
    \hfill
    \begin{subfigure}[t]{0.32\textwidth}
        \centering
        \includegraphics[width=\linewidth]{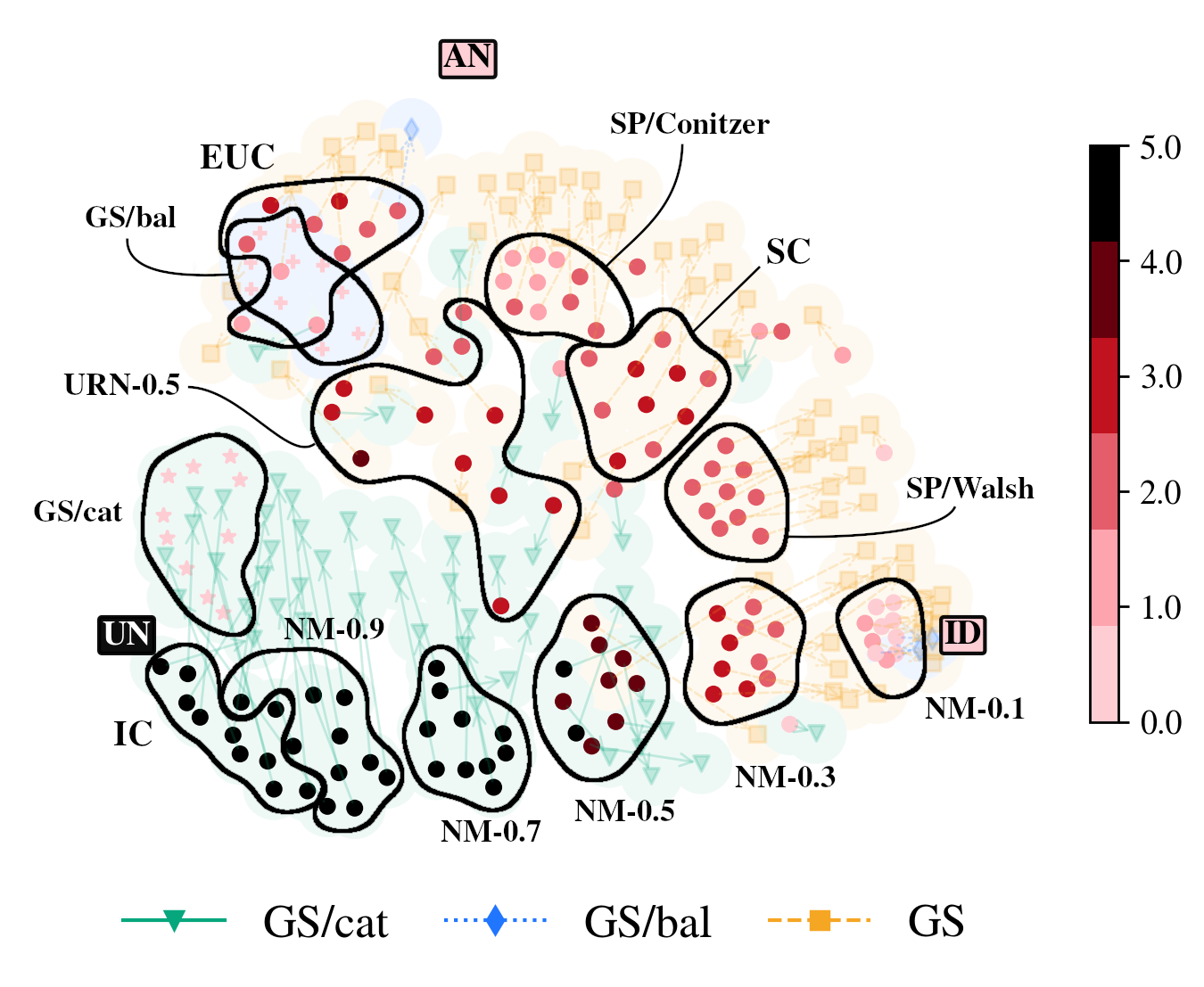}
        \caption{Map of synthetic elections with $8$ candidates and
          $96$ voters, and their $\GS$ projections (pointed to by the
          arrows). Colors represent the per-voter distances to the
          projections.}
        \label{fig:results:map}
    \end{subfigure}
    \hfill
    \begin{subfigure}[t]{0.32\textwidth}
        \centering
        \includegraphics[width=0.95\linewidth]{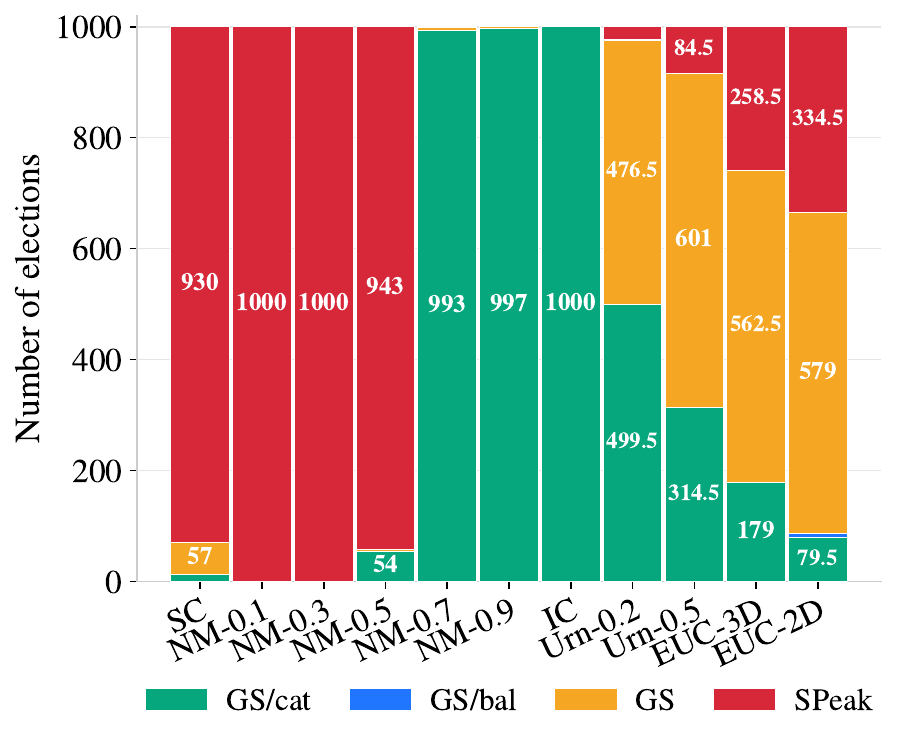}
        \caption{Fractions of elections ($|\candidates| = 8$, $|V| = 96$, sampled from given models) for which $\GScat$,
          $\GSbal$, $\GSbin$, and $\SPeak$ projections are closest
          (among these four).}
        \label{fig:results:distances}
    \end{subfigure}

    \caption{Results of our numerical experiments.}
    \label{fig:results:combined}
\end{figure*}

\subsection{Algorithms Running Time}
\label{sec:experiments:running}
\appendixsubsection{sec:experiments:running}
\toappendix{\FloatBarrier}

Due to limited space, for running-time analysis we focus on $\GScat$;
our results are most varied and most interesting in this case.  We
compare the brute-force algorithm (trying all possible mappings
$\mapping$ between $\candidates$ and $\operatorname{leaves}(T)$), the
\FPT algorithm parameterized by the number of candidates
$\numCandidates$ (\Cref{thm:NE:swap:GScat:FPT:numCandidates}), two
\FPT algorithms parameterized by the distance~$\budget$---the
algorithm tailored to the GS-caterpillar domain
(\Cref{thm:NE:swap:GSbal:GScat:FPT:budget}), denoted $\FPT(\budget)$,
and the general algorithm for domains characterized by forbidden
partial subelections (\Cref{thm:NE:swap:GScat:FPT:budget}), denoted
$\FPT(\budget)$-subelections---and the 
\iflong 
ILP formulation of the problem
(\Cref{sec:ILP:caterpillar}). 
For an analogous analysis regarding
other domains, see \Cref{app:sec:experiments:running}.
\fi

For each number of candidates $m$ in $\{2, \ldots, 20\}$ and each
$\phi \in \{0.0, 0.1, 0.3, 0.5, 0.7, 0.9\}$, we attempted to run each
of our algorithms on $1000$ elections sampled from $\phi$-NM, with 96
voters. If an algorithm did not complete within $15$ minutes, we
terminated it and did not run it for larger numbers of candidates. We
plot the averages of the obtained running times in
\Cref{fig:results:runningTime:mallows} (log scale), with the range
between the first and the third quartile highlighted by a shaded bar,
and the minimum and maximum values marked by ticks.

Barring small values of $\phi$, the $\FPT(\numCandidates)$ algorithm
performed best, whereas the $\FPT(\budget)$-subelections performed
worst, losing even to the trivial brute-force one. The $\FPT(\budget)$
algorithm performed far better and was the closest to challenging the
$\FPT(\numCandidates)$ one. Intuitively, this difference between the
theoretical and experimental results stems from the fact that the
$\FPT(\budget)$-subelections algorithm decreases the value of the
parameter by $1$ in each recursive call, whereas the $\FPT(\budget)$
may decrease it even by numbers proportional to the number of voters
(indeed, one might suspect that it is even $\FPT$ with respect to the
average number of swaps per voter, but we have not been able to prove
this).  Finally, we note that
our ILP-based algorithms perform poorly, which is surprising given
how effective this approach usually is.

\toappendix{
In \Cref{fig:runningTime:SPconitzer,fig:runningTime:SPwalsh,fig:runningTime:Scross,fig:runningTime:GScat,fig:runningTime:GSbal,fig:runningTime:IC,fig:runningTime:URN2,fig:runningTime:URN5,fig:runningTime:EUC}, we also provide the running time analysis for the remaining domains used later for the map of elections experiment. Each election consisted of 96 voters and $m$ candidates (x-axis). The running time in seconds (y-axis) is on a log scale. Each solver had 15 minutes to compute the distance and was not run for higher values of $m$ if it failed to finish in time for some $m' < m$. The results are comparable to the normalized Mallows culture. The only culture where the $\FPT(m)$ algorithm  is slower than all other algorithms except for the naïve brute-force is the group-separable caterpillar elections.

\begin{figure}[h!]
    \centering
    \includegraphics[width=0.9\linewidth]{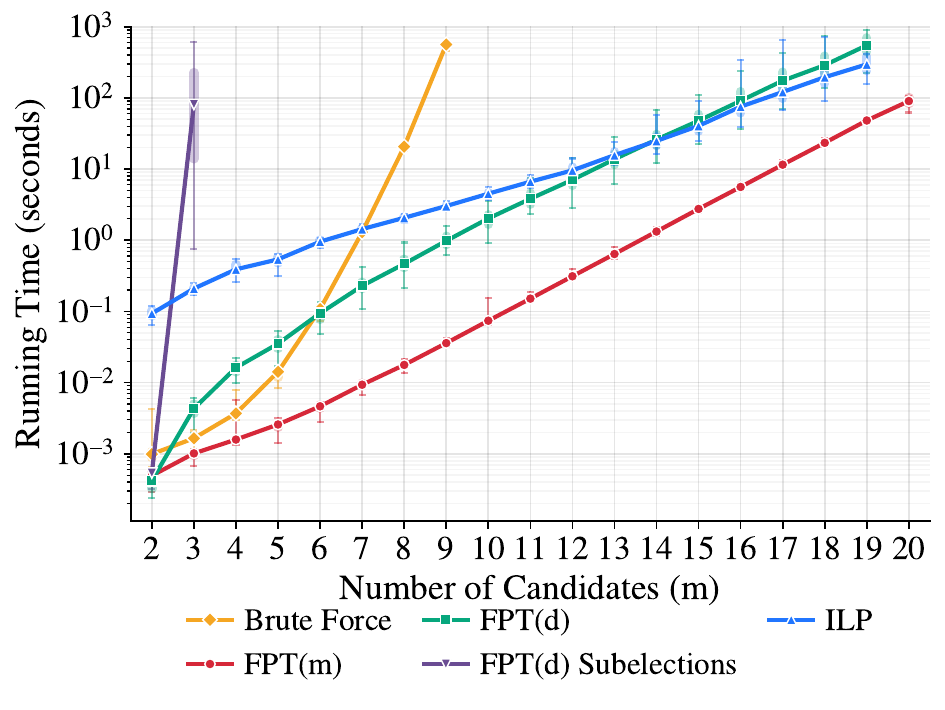}
    \caption{Running time of our algorithms on 1000 elections sampled from the SP/Conitzer culture.}
    \label{fig:runningTime:SPconitzer}
\end{figure}

\begin{figure}[h!]
    \centering
    \includegraphics[width=0.9\linewidth]{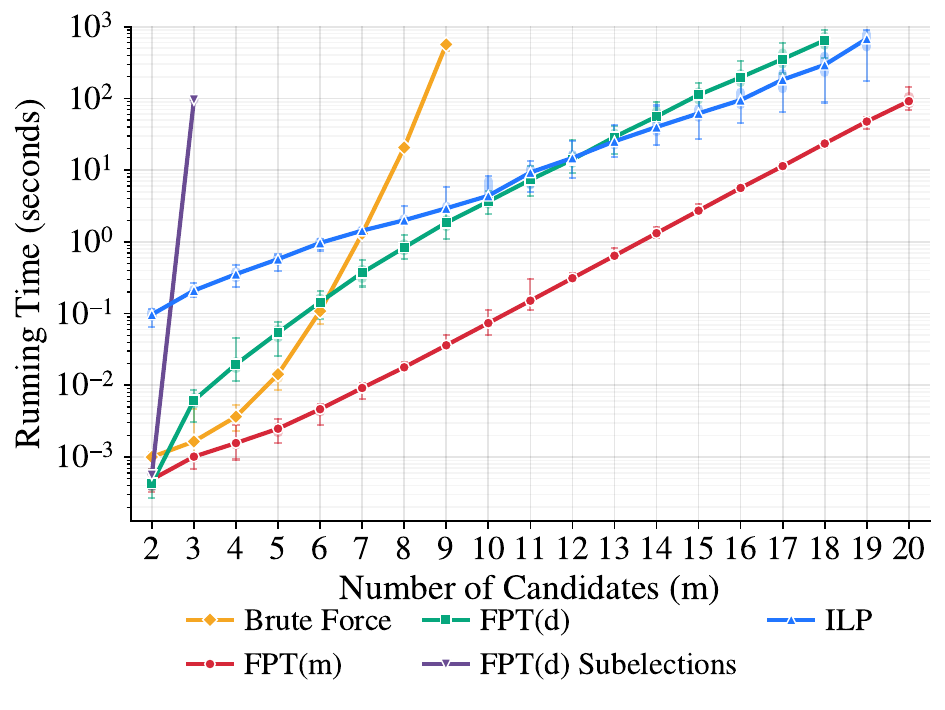}
    \caption{Running time of our algorithms on 1000 elections sampled from the SP/Walsh culture.}
    \label{fig:runningTime:SPwalsh}
\end{figure}

\begin{figure}[h!]
    \centering
    \includegraphics[width=0.9\linewidth]{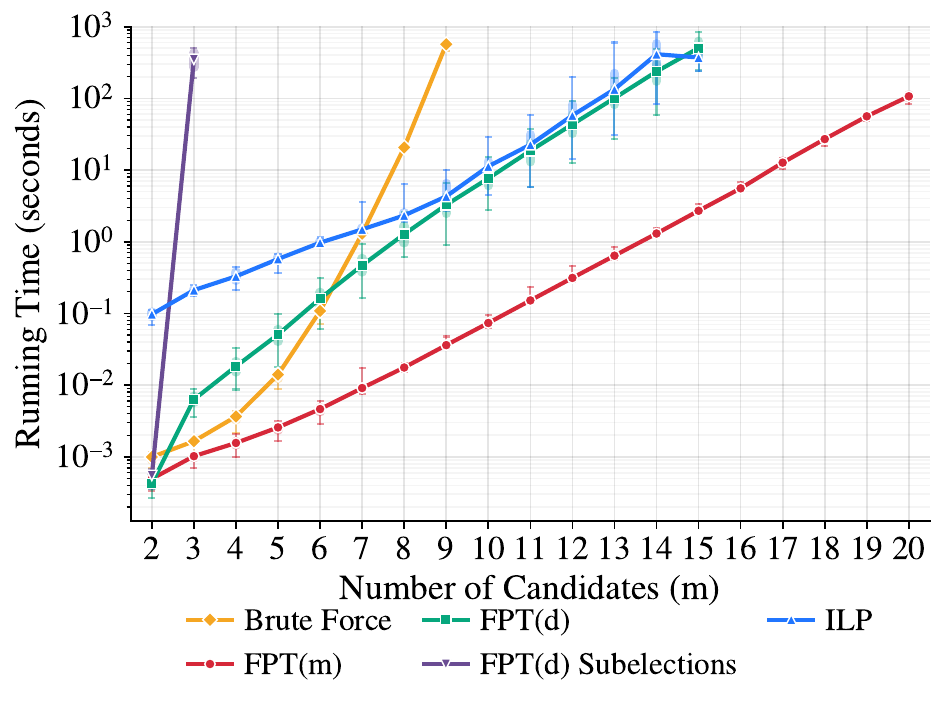}
    \caption{Running time of our algorithms on 1000 elections sampled from the Single-Crossing culture.}
    \label{fig:runningTime:Scross}
\end{figure}

\begin{figure}[h!]
    \centering
    \includegraphics[width=0.9\linewidth]{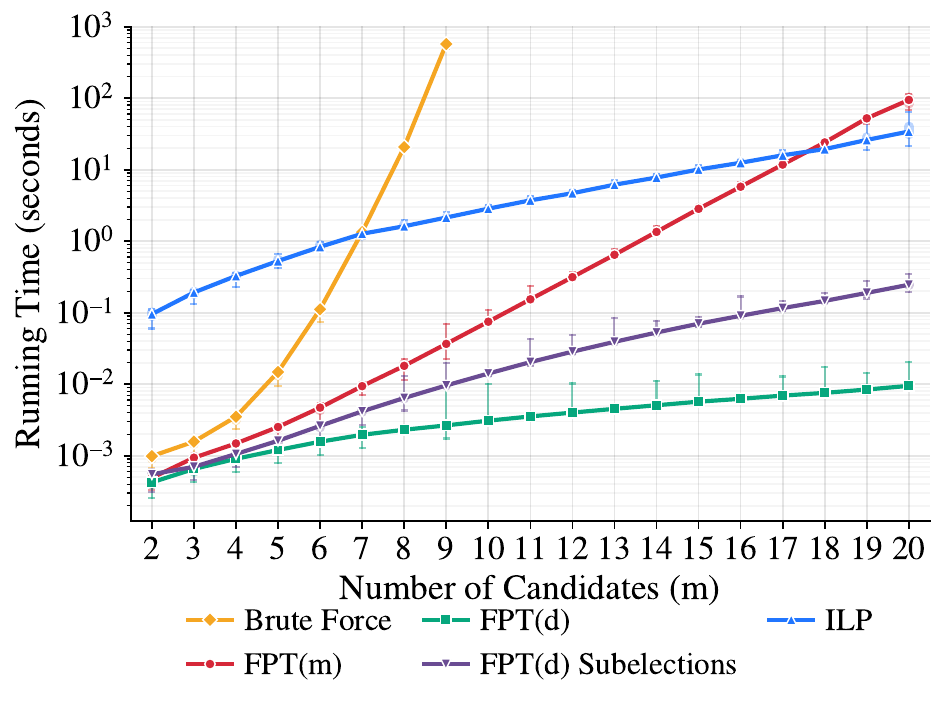}
    \caption{Running time of our algorithms on 1000 elections sampled from the GS/cat culture.}
    \label{fig:runningTime:GScat}
\end{figure}

\begin{figure}[h!]
    \centering
    \includegraphics[width=0.9\linewidth]{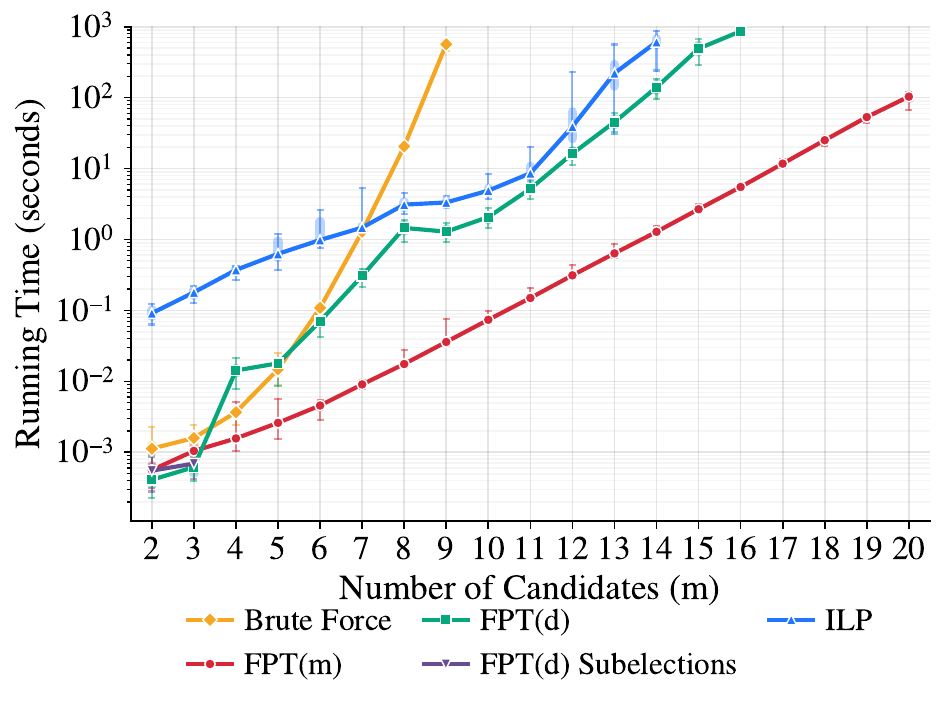}
    \caption{Running time of our algorithms on 1000 elections sampled from the GS/bal culture.}
    \label{fig:runningTime:GSbal}
\end{figure}

\begin{figure}[h!]
    \centering
    \includegraphics[width=0.9\linewidth]{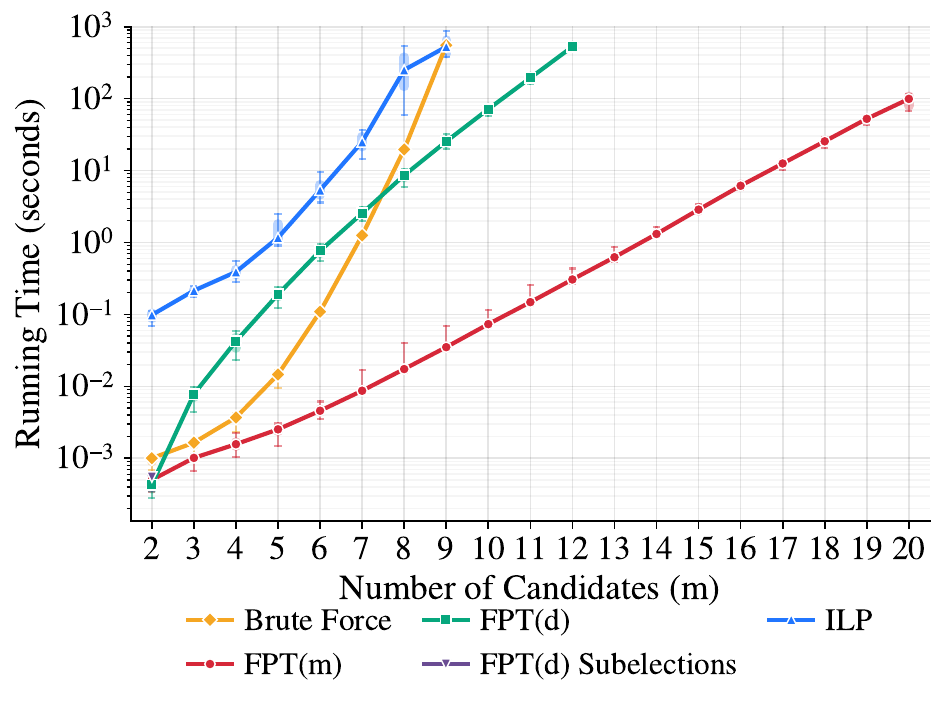}
    \caption{Running time of our algorithms on 1000 elections sampled from the Impartial culture.}
    \label{fig:runningTime:IC}
\end{figure}

\begin{figure}[h!]
    \centering
    \includegraphics[width=0.9\linewidth]{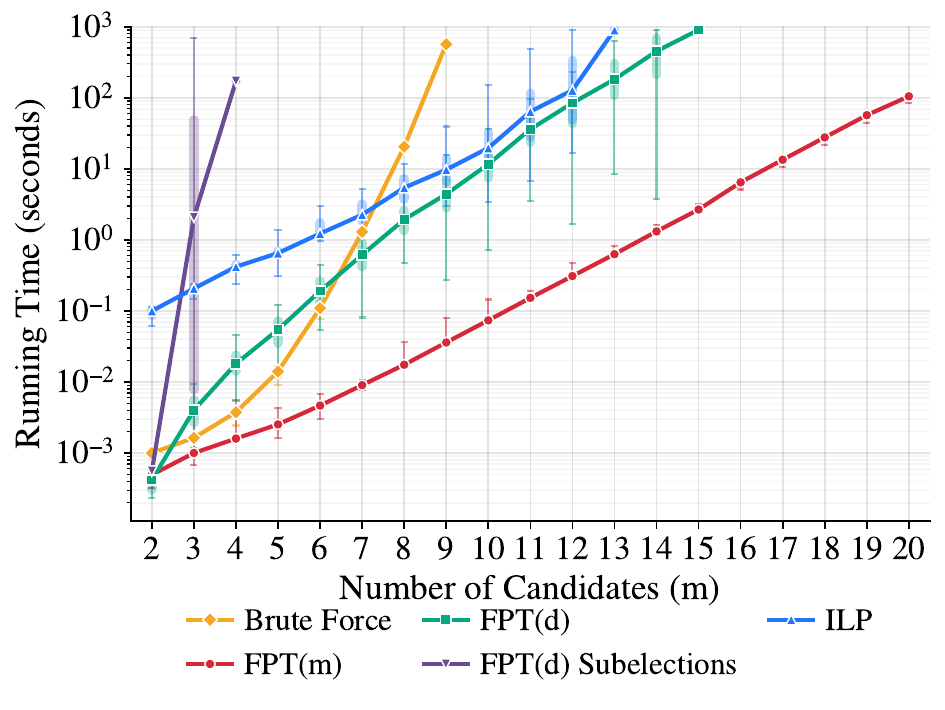}
    \caption{Running time of our algorithms on 1000 elections sampled from the URN-0.2 culture.}
    \label{fig:runningTime:URN2}
\end{figure}

\begin{figure}[h!]
    \centering
    \includegraphics[width=0.9\linewidth]{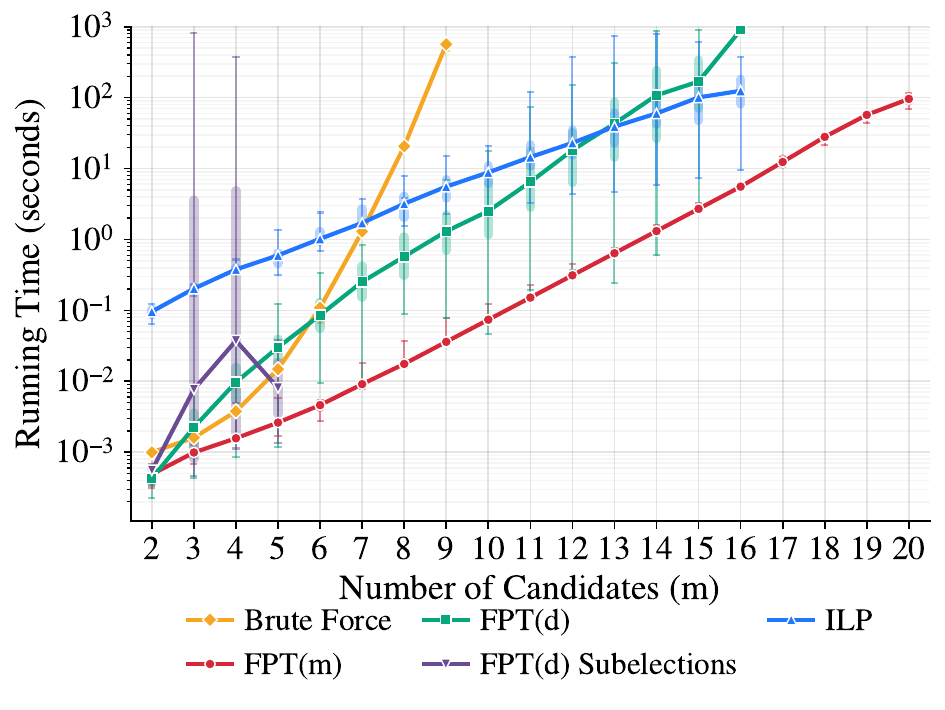}
    \caption{Running time of our algorithms on 1000 elections sampled from the URN-0.5 culture.}
    \label{fig:runningTime:URN5}
\end{figure}

\begin{figure}[h!]
    \centering
    \includegraphics[width=0.9\linewidth]{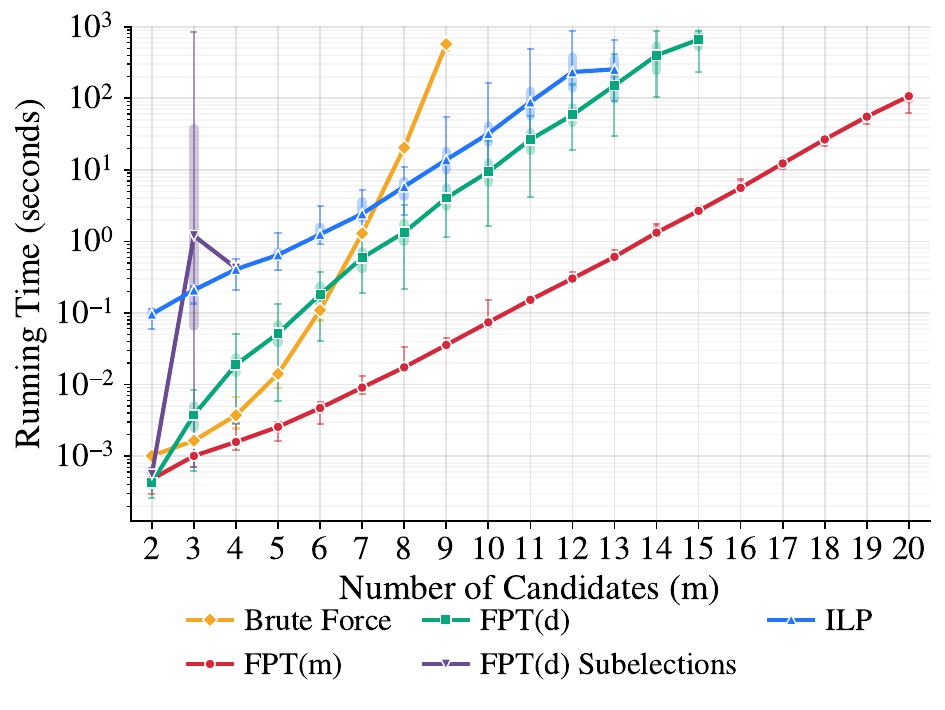}
    \caption{Running time of our algorithms on 1000 elections sampled from the EUC/uni culture.}
    \label{fig:runningTime:EUC}
\end{figure}

}

\subsection{Distance from a Structured Domain}
\label{sec:experiments:distance}
\appendixsubsection{sec:experiments:distance}
\toappendix{
\Cref{tab:map} shows the swap distance to the closest election of each group-separable domain for the elections of the map experiment.
\Cref{tab:closest-domain-c8-v96} then shows similar numerical results for our statistically more significant study with $1000$ samples per data point.
}
Briefly put, a map of elections is a visualization of an election
dataset so that each election is a dot on a 2D plane and the closer
two dots are, the fewer swaps of adjacent candidates are needed to
make the respective elections isomorphic; see, e.g., the works of
\citet{szu-boe-bre-fal-nie-sko-sli-tal:j:map} and
\citet{fal-kac-sor-szu-was:c:div-agr-pol-map}.  The map also includes
three special points: identity (ID), representing elections were all
votes are equal, antagonism (AN), representing elections where voters
are split into two halves, where each half has identical votes that
are reverses of those in the other half, and uniformity (UN),
representing elections with all possible votes.

\toappendix{
    \providecommand{\shadecell}[2]{\tikz[baseline=(s.base)]{\node (s) [transform shape, rounded corners=1pt, fill=blue!#1, inner sep=0, minimum width=15pt, minimum height=7pt] {#2};}}
    \begin{table*}[t]
      \centering
      \caption{Swap distance to the closest election of each group-separable domain for the elections of the map experiment ($8$ candidates, $96$ voters, $10$ elections per culture). The top panel reports absolute distances, the bottom panel the same statistics normalized per voter (i.e., divided by $n=96$). The share column gives the percentage of elections for which the domain is the closest of the three (ties split evenly; the general group-separable domain is only credited when strictly closer than both special cases). In every row the domain with the largest share is set in bold; cells are shaded proportionally to their value relative to the maximum of their statistic over all cells (shares relative to $100\%$).}
      \label{tab:map}
      \setlength{\tabcolsep}{2.5pt}
      \footnotesize
    \scalebox{0.8}{
      \begin{tabular}{l rrrrr rrrrr rrrrr}
        \toprule
        & \multicolumn{5}{c}{GS/cat} & \multicolumn{5}{c}{GS/bal} & \multicolumn{5}{c}{GS} \\
        \cmidrule(lr){2-6}\cmidrule(lr){7-11}\cmidrule(lr){12-16}
        Culture & avg$\pm$std & med & min & max & share & avg$\pm$std & med & min & max & share & avg$\pm$std & med & min & max & share \\
        \midrule
        \multicolumn{16}{c}{\emph{Absolute distance}} \\
        \midrule
        NM-0.1        & \shadecell{1.0}{108.6$\pm$4.7} & \shadecell{1.0}{109} & \shadecell{0.8}{99} & \shadecell{1.0}{115} & \shadecell{0.0}{0} & \shadecell{0.5}{80.3$\pm$6.5} & \shadecell{0.5}{79.5} & \shadecell{0.4}{69} & \shadecell{0.6}{91} & \shadecell{1.3}{20} & \textbf{\shadecell{0.5}{76.6$\pm$4.9}} & \textbf{\shadecell{0.5}{75.5}} & \textbf{\shadecell{0.4}{69}} & \textbf{\shadecell{0.5}{85}} & \textbf{\shadecell{21.1}{80}} \\
        NM-0.3        & \shadecell{5.9}{270.1$\pm$14.0} & \shadecell{6.1}{274.5} & \shadecell{5.1}{245} & \shadecell{6.1}{284} & \shadecell{0.0}{0} & \shadecell{5.9}{270.2$\pm$20.3} & \shadecell{6.1}{273} & \shadecell{5.2}{247} & \shadecell{7.4}{312} & \shadecell{0.0}{0} & \textbf{\shadecell{5.0}{247.3$\pm$18.2}} & \textbf{\shadecell{5.0}{248}} & \textbf{\shadecell{4.1}{218}} & \textbf{\shadecell{5.7}{275}} & \textbf{\shadecell{33.0}{100}} \\
        NM-0.5        & \textbf{\shadecell{11.9}{383.5$\pm$21.4}} & \textbf{\shadecell{12.0}{384}} & \textbf{\shadecell{10.9}{358}} & \textbf{\shadecell{13.1}{415}} & \textbf{\shadecell{21.1}{80}} & \shadecell{16.3}{448.5$\pm$31.7} & \shadecell{16.3}{448} & \shadecell{13.7}{401} & \shadecell{18.4}{492} & \shadecell{0.0}{0} & \shadecell{11.9}{383.1$\pm$21.4} & \shadecell{11.9}{382.5} & \shadecell{10.9}{358} & \shadecell{13.1}{415} & \shadecell{1.3}{20} \\
        NM-0.7        & \textbf{\shadecell{16.4}{448.9$\pm$16.6}} & \textbf{\shadecell{16.3}{447.5}} & \textbf{\shadecell{15.1}{421}} & \textbf{\shadecell{17.1}{475}} & \textbf{\shadecell{33.0}{100}} & \shadecell{27.2}{578.3$\pm$14.6} & \shadecell{27.3}{580} & \shadecell{26.7}{560} & \shadecell{27.7}{604} & \shadecell{0.0}{0} & \shadecell{16.4}{448.9$\pm$16.6} & \shadecell{16.3}{447.5} & \shadecell{15.1}{421} & \shadecell{17.1}{475} & \shadecell{0.0}{0} \\
        NM-0.9        & \textbf{\shadecell{17.6}{465.2$\pm$8.3}} & \textbf{\shadecell{17.5}{464.5}} & \textbf{\shadecell{17.6}{454}} & \textbf{\shadecell{17.3}{477}} & \textbf{\shadecell{33.0}{100}} & \shadecell{32.7}{634.6$\pm$11.3} & \shadecell{32.5}{632.5} & \shadecell{32.8}{620} & \shadecell{33.0}{659} & \shadecell{0.0}{0} & \shadecell{17.6}{465.2$\pm$8.3} & \shadecell{17.5}{464.5} & \shadecell{17.6}{454} & \shadecell{17.3}{477} & \shadecell{0.0}{0} \\
        IC            & \textbf{\shadecell{17.5}{464.9$\pm$14.8}} & \textbf{\shadecell{17.8}{468.5}} & \textbf{\shadecell{16.4}{438}} & \textbf{\shadecell{17.6}{481}} & \textbf{\shadecell{33.0}{100}} & \shadecell{33.0}{637.5$\pm$9.0} & \shadecell{33.0}{637.5} & \shadecell{33.0}{622} & \shadecell{32.4}{653} & \shadecell{0.0}{0} & \shadecell{17.5}{464.9$\pm$14.8} & \shadecell{17.8}{468.5} & \shadecell{16.4}{438} & \shadecell{17.6}{481} & \shadecell{0.0}{0} \\
        URN-0.2       & \shadecell{6.2}{277.1$\pm$27.5} & \shadecell{5.9}{268.5} & \shadecell{5.4}{251} & \shadecell{8.2}{329} & \shadecell{5.3}{40} & \shadecell{14.1}{416.3$\pm$51.5} & \shadecell{13.5}{408} & \shadecell{10.8}{356} & \shadecell{19.0}{500} & \shadecell{0.0}{0} & \textbf{\shadecell{5.9}{269$\pm$30.7}} & \textbf{\shadecell{5.8}{268}} & \textbf{\shadecell{4.1}{220}} & \textbf{\shadecell{7.8}{321}} & \textbf{\shadecell{11.9}{60}} \\
        URN-0.5       & \textbf{\shadecell{2.3}{167.4$\pm$64.6}} & \textbf{\shadecell{2.3}{169.5}} & \textbf{\shadecell{0.5}{79}} & \textbf{\shadecell{5.7}{273}} & \textbf{\shadecell{8.2}{50}} & \shadecell{8.6}{325.1$\pm$82.5} & \shadecell{10.4}{357.5} & \shadecell{3.5}{204} & \shadecell{14.1}{430} & \shadecell{0.0}{0} & \textbf{\shadecell{1.9}{153.5$\pm$59.5}} & \textbf{\shadecell{2.0}{156}} & \textbf{\shadecell{0.4}{66}} & \textbf{\shadecell{3.8}{223}} & \textbf{\shadecell{8.2}{50}} \\
        SP/Walsh      & \shadecell{4.4}{232.9$\pm$11.7} & \shadecell{4.4}{234} & \shadecell{3.9}{214} & \shadecell{4.7}{250} & \shadecell{0.0}{0} & \shadecell{5.7}{265.3$\pm$18.7} & \shadecell{5.8}{268} & \shadecell{4.1}{220} & \shadecell{6.2}{286} & \shadecell{0.0}{0} & \textbf{\shadecell{3.7}{212.1$\pm$13.8}} & \textbf{\shadecell{3.7}{214}} & \textbf{\shadecell{2.9}{183}} & \textbf{\shadecell{4.3}{237}} & \textbf{\shadecell{33.0}{100}} \\
        SP/Conitzer   & \shadecell{2.8}{186.2$\pm$20.6} & \shadecell{2.6}{179.5} & \shadecell{2.3}{165} & \shadecell{4.2}{235} & \shadecell{0.0}{0} & \shadecell{2.6}{177.9$\pm$19.0} & \shadecell{2.5}{174} & \shadecell{2.1}{157} & \shadecell{3.4}{211} & \shadecell{0.0}{0} & \textbf{\shadecell{2.1}{162.5$\pm$18.3}} & \textbf{\shadecell{2.0}{155.5}} & \textbf{\shadecell{1.8}{147}} & \textbf{\shadecell{3.3}{209}} & \textbf{\shadecell{33.0}{100}} \\
        SC            & \shadecell{6.1}{275.2$\pm$14.3} & \shadecell{6.2}{276.5} & \shadecell{5.1}{244} & \shadecell{6.6}{295} & \shadecell{0.0}{0} & \shadecell{8.1}{316.1$\pm$27.5} & \shadecell{8.0}{314.5} & \shadecell{5.8}{260} & \shadecell{9.7}{357} & \shadecell{0.0}{0} & \textbf{\shadecell{4.8}{242.4$\pm$19.4}} & \textbf{\shadecell{4.8}{242}} & \textbf{\shadecell{3.9}{214}} & \textbf{\shadecell{5.7}{274}} & \textbf{\shadecell{33.0}{100}} \\
        EUC-sphere    & \shadecell{4.8}{243.6$\pm$70.0} & \shadecell{5.3}{255} & \shadecell{1.5}{134} & \shadecell{7.5}{315} & \shadecell{0.3}{10} & \shadecell{5.5}{259.3$\pm$48.6} & \shadecell{6.0}{271.5} & \shadecell{2.7}{179} & \shadecell{8.0}{324} & \shadecell{0.3}{10} & \textbf{\shadecell{2.7}{181.4$\pm$38.4}} & \textbf{\shadecell{2.5}{174.5}} & \textbf{\shadecell{1.3}{124}} & \textbf{\shadecell{4.5}{243}} & \textbf{\shadecell{21.1}{80}} \\
        GS/cat        & \textbf{\shadecell{0.0}{0$\pm$0}} & \textbf{\shadecell{0.0}{0}} & \textbf{\shadecell{0.0}{0}} & \textbf{\shadecell{0.0}{0}} & \textbf{\shadecell{33.0}{100}} & \shadecell{30.3}{611.1$\pm$24.3} & \shadecell{30.8}{615.5} & \shadecell{25.7}{549} & \shadecell{30.8}{637} & \shadecell{0.0}{0} & \shadecell{0.0}{0$\pm$0} & \shadecell{0.0}{0} & \shadecell{0.0}{0} & \shadecell{0.0}{0} & \shadecell{0.0}{0} \\
        GS/bal        & \shadecell{7.3}{300.7$\pm$13.7} & \shadecell{7.5}{303.5} & \shadecell{6.5}{277} & \shadecell{7.8}{321} & \shadecell{0.0}{0} & \textbf{\shadecell{0.0}{0$\pm$0}} & \textbf{\shadecell{0.0}{0}} & \textbf{\shadecell{0.0}{0}} & \textbf{\shadecell{0.0}{0}} & \textbf{\shadecell{33.0}{100}} & \shadecell{0.0}{0$\pm$0} & \shadecell{0.0}{0} & \shadecell{0.0}{0} & \shadecell{0.0}{0} & \shadecell{0.0}{0} \\
        \midrule
        \multicolumn{16}{c}{\emph{Distance per voter}} \\
        \midrule
        NM-0.1        & \shadecell{1.0}{1.13$\pm$0.05} & \shadecell{1.0}{1.14} & \shadecell{0.8}{1.03} & \shadecell{1.0}{1.20} &  & \shadecell{0.5}{0.84$\pm$0.07} & \shadecell{0.5}{0.83} & \shadecell{0.4}{0.72} & \shadecell{0.6}{0.95} &  & \textbf{\shadecell{0.5}{0.80$\pm$0.05}} & \textbf{\shadecell{0.5}{0.79}} & \textbf{\shadecell{0.4}{0.72}} & \textbf{\shadecell{0.5}{0.89}} &  \\
        NM-0.3        & \shadecell{5.9}{2.81$\pm$0.15} & \shadecell{6.1}{2.86} & \shadecell{5.1}{2.55} & \shadecell{6.1}{2.96} &  & \shadecell{5.9}{2.81$\pm$0.21} & \shadecell{6.1}{2.84} & \shadecell{5.2}{2.57} & \shadecell{7.4}{3.25} &  & \textbf{\shadecell{5.0}{2.58$\pm$0.19}} & \textbf{\shadecell{5.0}{2.58}} & \textbf{\shadecell{4.1}{2.27}} & \textbf{\shadecell{5.7}{2.86}} &  \\
        NM-0.5        & \textbf{\shadecell{11.9}{3.99$\pm$0.22}} & \textbf{\shadecell{12.0}{4.00}} & \textbf{\shadecell{10.9}{3.73}} & \textbf{\shadecell{13.1}{4.32}} &  & \shadecell{16.3}{4.67$\pm$0.33} & \shadecell{16.3}{4.67} & \shadecell{13.7}{4.18} & \shadecell{18.4}{5.12} &  & \shadecell{11.9}{3.99$\pm$0.22} & \shadecell{11.9}{3.98} & \shadecell{10.9}{3.73} & \shadecell{13.1}{4.32} &  \\
        NM-0.7        & \textbf{\shadecell{16.4}{4.68$\pm$0.17}} & \textbf{\shadecell{16.3}{4.66}} & \textbf{\shadecell{15.1}{4.39}} & \textbf{\shadecell{17.1}{4.95}} &  & \shadecell{27.2}{6.02$\pm$0.15} & \shadecell{27.3}{6.04} & \shadecell{26.7}{5.83} & \shadecell{27.7}{6.29} &  & \shadecell{16.4}{4.68$\pm$0.17} & \shadecell{16.3}{4.66} & \shadecell{15.1}{4.39} & \shadecell{17.1}{4.95} &  \\
        NM-0.9        & \textbf{\shadecell{17.6}{4.85$\pm$0.09}} & \textbf{\shadecell{17.5}{4.84}} & \textbf{\shadecell{17.6}{4.73}} & \textbf{\shadecell{17.3}{4.97}} &  & \shadecell{32.7}{6.61$\pm$0.12} & \shadecell{32.5}{6.59} & \shadecell{32.8}{6.46} & \shadecell{33.0}{6.86} &  & \shadecell{17.6}{4.85$\pm$0.09} & \shadecell{17.5}{4.84} & \shadecell{17.6}{4.73} & \shadecell{17.3}{4.97} &  \\
        IC            & \textbf{\shadecell{17.5}{4.84$\pm$0.15}} & \textbf{\shadecell{17.8}{4.88}} & \textbf{\shadecell{16.4}{4.56}} & \textbf{\shadecell{17.6}{5.01}} &  & \shadecell{33.0}{6.64$\pm$0.09} & \shadecell{33.0}{6.64} & \shadecell{33.0}{6.48} & \shadecell{32.4}{6.80} &  & \shadecell{17.5}{4.84$\pm$0.15} & \shadecell{17.8}{4.88} & \shadecell{16.4}{4.56} & \shadecell{17.6}{5.01} &  \\
        URN-0.2       & \shadecell{6.2}{2.89$\pm$0.29} & \shadecell{5.9}{2.80} & \shadecell{5.4}{2.61} & \shadecell{8.2}{3.43} &  & \shadecell{14.1}{4.34$\pm$0.54} & \shadecell{13.5}{4.25} & \shadecell{10.8}{3.71} & \shadecell{19.0}{5.21} &  & \textbf{\shadecell{5.9}{2.80$\pm$0.32}} & \textbf{\shadecell{5.8}{2.79}} & \textbf{\shadecell{4.1}{2.29}} & \textbf{\shadecell{7.8}{3.34}} &  \\
        URN-0.5       & \textbf{\shadecell{2.3}{1.74$\pm$0.67}} & \textbf{\shadecell{2.3}{1.77}} & \textbf{\shadecell{0.5}{0.82}} & \textbf{\shadecell{5.7}{2.84}} &  & \shadecell{8.6}{3.39$\pm$0.86} & \shadecell{10.4}{3.72} & \shadecell{3.5}{2.12} & \shadecell{14.1}{4.48} &  & \textbf{\shadecell{1.9}{1.60$\pm$0.62}} & \textbf{\shadecell{2.0}{1.62}} & \textbf{\shadecell{0.4}{0.69}} & \textbf{\shadecell{3.8}{2.32}} &  \\
        SP/Walsh      & \shadecell{4.4}{2.43$\pm$0.12} & \shadecell{4.4}{2.44} & \shadecell{3.9}{2.23} & \shadecell{4.7}{2.60} &  & \shadecell{5.7}{2.76$\pm$0.19} & \shadecell{5.8}{2.79} & \shadecell{4.1}{2.29} & \shadecell{6.2}{2.98} &  & \textbf{\shadecell{3.7}{2.21$\pm$0.14}} & \textbf{\shadecell{3.7}{2.23}} & \textbf{\shadecell{2.9}{1.91}} & \textbf{\shadecell{4.3}{2.47}} &  \\
        SP/Conitzer   & \shadecell{2.8}{1.94$\pm$0.21} & \shadecell{2.6}{1.87} & \shadecell{2.3}{1.72} & \shadecell{4.2}{2.45} &  & \shadecell{2.6}{1.85$\pm$0.20} & \shadecell{2.5}{1.81} & \shadecell{2.1}{1.64} & \shadecell{3.4}{2.20} &  & \textbf{\shadecell{2.1}{1.69$\pm$0.19}} & \textbf{\shadecell{2.0}{1.62}} & \textbf{\shadecell{1.8}{1.53}} & \textbf{\shadecell{3.3}{2.18}} &  \\
        SC            & \shadecell{6.1}{2.87$\pm$0.15} & \shadecell{6.2}{2.88} & \shadecell{5.1}{2.54} & \shadecell{6.6}{3.07} &  & \shadecell{8.1}{3.29$\pm$0.29} & \shadecell{8.0}{3.28} & \shadecell{5.8}{2.71} & \shadecell{9.7}{3.72} &  & \textbf{\shadecell{4.8}{2.52$\pm$0.20}} & \textbf{\shadecell{4.8}{2.52}} & \textbf{\shadecell{3.9}{2.23}} & \textbf{\shadecell{5.7}{2.85}} &  \\
        EUC-sphere    & \shadecell{4.8}{2.54$\pm$0.73} & \shadecell{5.3}{2.66} & \shadecell{1.5}{1.40} & \shadecell{7.5}{3.28} &  & \shadecell{5.5}{2.70$\pm$0.51} & \shadecell{6.0}{2.83} & \shadecell{2.7}{1.86} & \shadecell{8.0}{3.38} &  & \textbf{\shadecell{2.7}{1.89$\pm$0.40}} & \textbf{\shadecell{2.5}{1.82}} & \textbf{\shadecell{1.3}{1.29}} & \textbf{\shadecell{4.5}{2.53}} &  \\
        GS/cat        & \textbf{\shadecell{0.0}{0.00$\pm$0.00}} & \textbf{\shadecell{0.0}{0.00}} & \textbf{\shadecell{0.0}{0.00}} & \textbf{\shadecell{0.0}{0.00}} &  & \shadecell{30.3}{6.37$\pm$0.25} & \shadecell{30.8}{6.41} & \shadecell{25.7}{5.72} & \shadecell{30.8}{6.64} &  & \shadecell{0.0}{0.00$\pm$0.00} & \shadecell{0.0}{0.00} & \shadecell{0.0}{0.00} & \shadecell{0.0}{0.00} &  \\
        GS/bal        & \shadecell{7.3}{3.13$\pm$0.14} & \shadecell{7.5}{3.16} & \shadecell{6.5}{2.89} & \shadecell{7.8}{3.34} &  & \textbf{\shadecell{0.0}{0.00$\pm$0.00}} & \textbf{\shadecell{0.0}{0.00}} & \textbf{\shadecell{0.0}{0.00}} & \textbf{\shadecell{0.0}{0.00}} &  & \shadecell{0.0}{0.00$\pm$0.00} & \shadecell{0.0}{0.00} & \shadecell{0.0}{0.00} & \shadecell{0.0}{0.00} &  \\
        \bottomrule
      \end{tabular}
      }
    \end{table*}
}

We show our map in \Cref{fig:results:map}, containing 10 elections for
each of the synthetic models we considered (with $8$ candidates and
$96$ voters); for each of these \emph{original} elections we computed
the closest $\GS$ election, called its \emph{projection}, and also
included it on the map, with an arrow pointing from the original one
to the projection. Each original election is colored by the number of
swaps (per voter) needed to obtain its projection. Each projection
is colored according to its type.

We note that it suffices to make fewer than $5$ swaps of adjacent
candidates (per voter) to transform any of our elections to its
projection. The largest possible swap distance between two votes with
$8$ candidates is $28$, so all our elections are at a modest distance
from having $\GS$ structure. Second, there is a clear separation
between elections projected to $\GScat$ ones and the remaining
ones. Specifically, the former have less internal structure (and,
hence, are closer to IC) and the latter are structured in one way or
another; e.g., are close to ID, are single-peaked (SP),
single-crossing (SC), or are Euclidean (EUC). Third, the projections
are systematically located in a particular neighborhood of their
originals (often in the direction of AN, suggesting increased
polarization).

While the map gives an intuitive overview, we also performed a more
statistically significant study. For several models, we have sampled
$1000$ elections and computed their closest $\GScat$, $\GSbal$, and
$\GSbin$ elections, as well as the closest single-peaked ones, as an
external reference point.  In \Cref{fig:results:distances}, we show
which fractions of elections from a given culture are closest to each
of these domains ($\GS$ gets ``a point'' only if neither a $\GScat$
nor a $\GSbal$ election was closest to the considered one; any
remaining ties split the point equally among the tied domains). 
\iflong 
A more
detailed overview with particular numerical values is available in \Cref{tab:closest-domain-c8-v96}.
\fi 
The general conclusions are the
same as in the case of the map, with the additional observation that
single-peaked elections are closest to NM-$\phi$ ones with $\phi \leq 0.5$.

\toappendix{
\providecommand{\shadecell}[2]{\tikz[baseline=(s.base)]{\node (s) [transform shape, rounded corners=1pt, fill=blue!#1, inner sep=0, minimum width=15pt, minimum height=7pt] {#2};}}
\begin{table*}[t]
  \centering
  \caption{Swap distance to the closest election of each restricted domain ($8$ candidates, $96$ voters, $1000$ elections per culture). The top panel reports absolute distances, the bottom panel the same statistics normalized per voter (i.e., divided by $n=96$). The share column gives the percentage of elections for which the domain is the closest of the four (ties split evenly; the general group-separable domain is only credited when strictly closer than both special cases). In every row the domain with the largest share is set in bold; cells are shaded proportionally to their value relative to the maximum of their statistic over all cells (shares relative to $100\%$).}
  \label{tab:closest-domain-c8-v96}
  \setlength{\tabcolsep}{2.5pt}
  \footnotesize
  \scalebox{0.8}{
  \begin{tabular}{l rrrrr rrrrr rrrrr rrrrr}
    \toprule
    & \multicolumn{5}{c}{GS/cat} & \multicolumn{5}{c}{GS/bal} & \multicolumn{5}{c}{GS} & \multicolumn{5}{c}{SPeak} \\
    \cmidrule(lr){2-6}\cmidrule(lr){7-11}\cmidrule(lr){12-16}\cmidrule(lr){17-21}
    Culture & avg$\pm$std & med & min & max & share & avg$\pm$std & med & min & max & share & avg$\pm$std & med & min & max & share & avg$\pm$std & med & min & max & share \\
    \midrule
    \multicolumn{21}{c}{\emph{Absolute distance}} \\
    \midrule
    NM-0.1        & \shadecell{0.7}{101.2$\pm$9.4} & \shadecell{0.7}{101} & \shadecell{0.4}{66} & \shadecell{1.1}{133} & \shadecell{0.0}{0} & \shadecell{0.4}{75.0$\pm$9.9} & \shadecell{0.4}{75} & \shadecell{0.2}{44} & \shadecell{0.8}{112} & \shadecell{0.0}{0} & \shadecell{0.4}{70.9$\pm$8.4} & \shadecell{0.4}{71} & \shadecell{0.2}{44} & \shadecell{0.6}{97} & \shadecell{0.0}{0} & \textbf{\shadecell{0.1}{31.5$\pm$6.7}} & \textbf{\shadecell{0.1}{31}} & \textbf{\shadecell{0.0}{13}} & \textbf{\shadecell{0.2}{55}} & \textbf{\shadecell{33.0}{100}} \\
    NM-0.3        & \shadecell{5.3}{268.6$\pm$15.1} & \shadecell{5.2}{268} & \shadecell{4.4}{222} & \shadecell{6.3}{313} & \shadecell{0.0}{0} & \shadecell{4.9}{260.2$\pm$20.2} & \shadecell{4.9}{260} & \shadecell{3.4}{195} & \shadecell{6.8}{324} & \shadecell{0.0}{0} & \shadecell{4.2}{239.3$\pm$17.1} & \shadecell{4.2}{239} & \shadecell{3.2}{190} & \shadecell{5.6}{293} & \shadecell{0.0}{0} & \textbf{\shadecell{2.4}{179.6$\pm$17.0}} & \textbf{\shadecell{2.4}{180}} & \textbf{\shadecell{1.5}{129}} & \textbf{\shadecell{3.5}{233}} & \textbf{\shadecell{33.0}{100}} \\
    NM-0.5        & \shadecell{10.7}{383.2$\pm$15.7} & \shadecell{10.7}{384} & \shadecell{10.0}{334} & \shadecell{11.5}{422} & \shadecell{0.1}{5.4} & \shadecell{14.4}{443.7$\pm$24.7} & \shadecell{14.3}{443} & \shadecell{11.2}{354} & \shadecell{17.4}{519} & \shadecell{0.0}{0} & \shadecell{10.6}{381.7$\pm$16.2} & \shadecell{10.6}{382} & \shadecell{10.0}{334} & \shadecell{11.5}{422} & \shadecell{0.0}{0.3} & \textbf{\shadecell{9.1}{352.5$\pm$22.9}} & \textbf{\shadecell{9.0}{352}} & \textbf{\shadecell{6.6}{272}} & \textbf{\shadecell{12.3}{436}} & \textbf{\shadecell{29.3}{94.3}} \\
    NM-0.7        & \textbf{\shadecell{14.6}{447.0$\pm$12.3}} & \textbf{\shadecell{14.6}{448}} & \textbf{\shadecell{14.4}{401}} & \textbf{\shadecell{15.0}{481}} & \textbf{\shadecell{32.5}{99.3}} & \shadecell{24.8}{582.3$\pm$20.1} & \shadecell{24.8}{583} & \shadecell{22.6}{502} & \shadecell{26.6}{641} & \shadecell{0.0}{0} & \shadecell{14.6}{447.0$\pm$12.3} & \shadecell{14.6}{448} & \shadecell{14.2}{398} & \shadecell{15.0}{481} & \shadecell{0.0}{0.6} & \shadecell{19.3}{513.5$\pm$23.4} & \shadecell{19.2}{514} & \shadecell{17.3}{439} & \shadecell{22.6}{591} & \shadecell{0.0}{0.1} \\
    NM-0.9        & \textbf{\shadecell{15.6}{462.2$\pm$10.8}} & \textbf{\shadecell{15.6}{463}} & \textbf{\shadecell{16.0}{423}} & \textbf{\shadecell{15.4}{488}} & \textbf{\shadecell{32.8}{99.7}} & \shadecell{29.0}{629.6$\pm$11.5} & \shadecell{29.0}{631} & \shadecell{28.6}{565} & \shadecell{27.9}{656} & \shadecell{0.0}{0} & \shadecell{15.6}{462.2$\pm$10.8} & \shadecell{15.6}{463} & \shadecell{16.0}{423} & \shadecell{15.4}{488} & \shadecell{0.0}{0.3} & \shadecell{30.4}{644.7$\pm$19.0} & \shadecell{30.4}{646} & \shadecell{30.1}{580} & \shadecell{32.0}{703} & \shadecell{0.0}{0} \\
    IC            & \textbf{\shadecell{15.7}{462.9$\pm$11.0}} & \textbf{\shadecell{15.7}{464}} & \textbf{\shadecell{15.7}{419}} & \textbf{\shadecell{15.6}{491}} & \textbf{\shadecell{33.0}{100}} & \shadecell{29.0}{630.3$\pm$11.2} & \shadecell{29.1}{632} & \shadecell{30.9}{587} & \shadecell{27.9}{657} & \shadecell{0.0}{0} & \shadecell{15.7}{462.9$\pm$11.0} & \shadecell{15.7}{464} & \shadecell{15.7}{419} & \shadecell{15.6}{491} & \shadecell{0.0}{0} & \shadecell{33.0}{671.9$\pm$15.3} & \shadecell{33.0}{673} & \shadecell{33.0}{607} & \shadecell{33.0}{714} & \shadecell{0.0}{0} \\
    URN-0.2       & \textbf{\shadecell{5.8}{281.3$\pm$56.8}} & \textbf{\shadecell{6.1}{289}} & \textbf{\shadecell{0.3}{59}} & \textbf{\shadecell{11.3}{418}} & \textbf{\shadecell{8.2}{49.95}} & \shadecell{13.5}{429.8$\pm$71.5} & \shadecell{14.2}{441} & \shadecell{1.1}{111} & \shadecell{20.8}{567} & \shadecell{0.0}{0} & \shadecell{5.5}{273.8$\pm$57.7} & \shadecell{5.8}{281} & \shadecell{0.3}{59} & \shadecell{10.9}{411} & \shadecell{7.5}{47.65} & \shadecell{11.1}{389.8$\pm$82.4} & \shadecell{11.6}{399} & \shadecell{0.4}{67} & \shadecell{23.1}{597} & \shadecell{0.0}{2.4} \\
    URN-0.5       & \shadecell{2.6}{187.7$\pm$75.6} & \shadecell{2.7}{192} & \shadecell{0.0}{12} & \shadecell{8.4}{360} & \shadecell{3.3}{31.45} & \shadecell{7.5}{320.9$\pm$102.4} & \shadecell{8.2}{335.5} & \shadecell{0.0}{20} & \shadecell{17.9}{526} & \shadecell{0.0}{0} & \textbf{\shadecell{2.2}{173.1$\pm$74.0}} & \textbf{\shadecell{2.2}{173}} & \textbf{\shadecell{0.0}{8}} & \textbf{\shadecell{7.7}{346}} & \textbf{\shadecell{11.9}{60.1}} & \shadecell{4.6}{251.3$\pm$102.5} & \shadecell{4.6}{251} & \shadecell{0.0}{4} & \shadecell{18.6}{536} & \shadecell{0.2}{8.45} \\
    SP/Walsh      & \shadecell{4.1}{235.8$\pm$12.9} & \shadecell{4.1}{236} & \shadecell{3.0}{182} & \shadecell{4.9}{276} & \shadecell{0.0}{0} & \shadecell{5.3}{268.3$\pm$20.5} & \shadecell{5.2}{267} & \shadecell{3.7}{202} & \shadecell{7.6}{343} & \shadecell{0.0}{0} & \shadecell{3.5}{217.9$\pm$13.5} & \shadecell{3.5}{218} & \shadecell{2.8}{178} & \shadecell{4.4}{262} & \shadecell{0.0}{0} & \textbf{\shadecell{0.0}{0$\pm$0}} & \textbf{\shadecell{0.0}{0}} & \textbf{\shadecell{0.0}{0}} & \textbf{\shadecell{0.0}{0}} & \textbf{\shadecell{33.0}{100}} \\
    SP/Conitzer   & \shadecell{2.5}{184.6$\pm$16.1} & \shadecell{2.5}{185} & \shadecell{1.7}{136} & \shadecell{3.5}{234} & \shadecell{0.0}{0} & \shadecell{2.3}{178.3$\pm$22.2} & \shadecell{2.3}{178} & \shadecell{1.3}{119} & \shadecell{4.5}{265} & \shadecell{0.0}{0} & \shadecell{1.9}{160.1$\pm$17.1} & \shadecell{1.9}{160} & \shadecell{1.0}{106} & \shadecell{2.7}{206} & \shadecell{0.0}{0} & \textbf{\shadecell{0.0}{0$\pm$0}} & \textbf{\shadecell{0.0}{0}} & \textbf{\shadecell{0.0}{0}} & \textbf{\shadecell{0.0}{0}} & \textbf{\shadecell{33.0}{100}} \\
    SC            & \shadecell{5.1}{265.3$\pm$34.6} & \shadecell{5.2}{267} & \shadecell{2.0}{148} & \shadecell{7.8}{348} & \shadecell{0.0}{1.3} & \shadecell{7.8}{325.7$\pm$33.4} & \shadecell{7.8}{327.5} & \shadecell{4.1}{215} & \shadecell{11.3}{417} & \shadecell{0.0}{0} & \shadecell{4.1}{237.8$\pm$27.0} & \shadecell{4.2}{240} & \shadecell{1.8}{142} & \shadecell{5.7}{298} & \shadecell{0.1}{5.7} & \textbf{\shadecell{1.8}{156.2$\pm$38.7}} & \textbf{\shadecell{1.8}{157}} & \textbf{\shadecell{0.1}{29}} & \textbf{\shadecell{4.5}{265}} & \textbf{\shadecell{28.5}{93}} \\
    1D EUC-cube   & \shadecell{1.3}{131.0$\pm$41.7} & \shadecell{1.3}{132} & \shadecell{0.1}{26} & \shadecell{4.2}{256} & \shadecell{0.0}{0} & \shadecell{1.9}{162.4$\pm$43.7} & \shadecell{2.0}{164} & \shadecell{0.1}{32} & \shadecell{6.0}{305} & \shadecell{0.0}{0} & \shadecell{0.6}{93.4$\pm$29.1} & \shadecell{0.6}{94} & \shadecell{0.0}{5} & \shadecell{1.9}{172} & \shadecell{0.0}{0} & \textbf{\shadecell{0.0}{0$\pm$0}} & \textbf{\shadecell{0.0}{0}} & \textbf{\shadecell{0.0}{0}} & \textbf{\shadecell{0.0}{0}} & \textbf{\shadecell{33.0}{100}} \\
    2D EUC-cube   & \shadecell{5.6}{276.5$\pm$39.3} & \shadecell{5.6}{278} & \shadecell{2.1}{154} & \shadecell{9.2}{376} & \shadecell{0.2}{7.95} & \shadecell{9.2}{354.1$\pm$52.9} & \shadecell{9.3}{358} & \shadecell{2.4}{162} & \shadecell{15.2}{485} & \shadecell{0.0}{0.7} & \textbf{\shadecell{4.5}{248.6$\pm$38.2}} & \textbf{\shadecell{4.6}{250}} & \textbf{\shadecell{1.4}{125}} & \textbf{\shadecell{7.6}{343}} & \textbf{\shadecell{11.1}{57.9}} & \shadecell{5.3}{269.5$\pm$59.7} & \shadecell{5.5}{274} & \shadecell{0.5}{78} & \shadecell{10.0}{393} & \shadecell{3.7}{33.45} \\
    3D EUC-cube   & \shadecell{7.9}{328.5$\pm$35.3} & \shadecell{8.0}{331} & \shadecell{3.5}{197} & \shadecell{11.7}{425} & \shadecell{1.1}{17.9} & \shadecell{13.2}{424.4$\pm$45.4} & \shadecell{13.3}{427.5} & \shadecell{6.7}{273} & \shadecell{19.2}{545} & \shadecell{0.0}{0} & \textbf{\shadecell{7.1}{312.1$\pm$35.3}} & \textbf{\shadecell{7.1}{313}} & \textbf{\shadecell{3.3}{191}} & \textbf{\shadecell{11.4}{419}} & \textbf{\shadecell{10.4}{56.25}} & \shadecell{8.6}{342.4$\pm$55.9} & \shadecell{8.7}{346} & \shadecell{2.1}{153} & \shadecell{16.3}{502} & \shadecell{2.2}{25.85} \\
    4D EUC-cube   & \shadecell{9.2}{354.8$\pm$31.9} & \shadecell{9.3}{357} & \shadecell{5.6}{251} & \shadecell{12.0}{430} & \shadecell{3.6}{33.25} & \shadecell{15.8}{465.2$\pm$41.2} & \shadecell{15.9}{467.5} & \shadecell{9.9}{333} & \shadecell{21.1}{571} & \shadecell{0.0}{0} & \textbf{\shadecell{8.7}{345.9$\pm$32.1}} & \textbf{\shadecell{8.8}{347}} & \textbf{\shadecell{5.6}{249}} & \textbf{\shadecell{11.9}{428}} & \textbf{\shadecell{7.4}{47.2}} & \shadecell{11.0}{387.6$\pm$56.0} & \shadecell{11.0}{389} & \shadecell{4.0}{211} & \shadecell{18.9}{540} & \shadecell{1.3}{19.55} \\
    5D EUC-cube   & \textbf{\shadecell{10.1}{372.4$\pm$30.7}} & \textbf{\shadecell{10.2}{374}} & \textbf{\shadecell{5.7}{252}} & \textbf{\shadecell{13.4}{455}} & \textbf{\shadecell{6.2}{43.4}} & \shadecell{17.5}{488.7$\pm$36.7} & \shadecell{17.6}{492} & \shadecell{11.2}{354} & \shadecell{21.6}{578} & \shadecell{0.0}{0} & \shadecell{9.8}{366.5$\pm$31.4} & \shadecell{9.9}{368} & \shadecell{5.7}{252} & \shadecell{13.3}{454} & \shadecell{5.1}{39.45} & \shadecell{12.4}{411.9$\pm$56.1} & \shadecell{12.5}{414} & \shadecell{4.2}{217} & \shadecell{20.1}{557} & \shadecell{1.0}{17.15} \\
    GS/cat        & \textbf{\shadecell{0.0}{0$\pm$0}} & \textbf{\shadecell{0.0}{0}} & \textbf{\shadecell{0.0}{0}} & \textbf{\shadecell{0.0}{0}} & \textbf{\shadecell{33.0}{100}} & \shadecell{28.3}{622.1$\pm$16.6} & \shadecell{28.5}{625} & \shadecell{28.4}{563} & \shadecell{28.3}{661} & \shadecell{0.0}{0} & \shadecell{0.0}{0$\pm$0} & \shadecell{0.0}{0} & \shadecell{0.0}{0} & \shadecell{0.0}{0} & \shadecell{0.0}{0} & \shadecell{24.4}{577.4$\pm$25.6} & \shadecell{24.3}{578} & \shadecell{22.3}{499} & \shadecell{27.3}{649} & \shadecell{0.0}{0} \\
    GS/bal        & \shadecell{6.6}{299.5$\pm$12.2} & \shadecell{6.6}{301} & \shadecell{5.6}{251} & \shadecell{6.8}{325} & \shadecell{0.0}{0} & \textbf{\shadecell{0.0}{0$\pm$0}} & \textbf{\shadecell{0.0}{0}} & \textbf{\shadecell{0.0}{0}} & \textbf{\shadecell{0.0}{0}} & \textbf{\shadecell{33.0}{100}} & \shadecell{0.0}{0$\pm$0} & \shadecell{0.0}{0} & \shadecell{0.0}{0} & \shadecell{0.0}{0} & \shadecell{0.0}{0} & \shadecell{6.6}{300.1$\pm$12.4} & \shadecell{6.6}{301} & \shadecell{5.7}{252} & \shadecell{6.9}{326} & \shadecell{0.0}{0} \\
    \midrule
    \multicolumn{21}{c}{\emph{Distance per voter}} \\
    \midrule
    NM-0.1        & \shadecell{0.7}{1.05$\pm$0.10} & \shadecell{0.7}{1.05} & \shadecell{0.4}{0.69} & \shadecell{1.1}{1.39} &  & \shadecell{0.4}{0.78$\pm$0.10} & \shadecell{0.4}{0.78} & \shadecell{0.2}{0.46} & \shadecell{0.8}{1.17} &  & \shadecell{0.4}{0.74$\pm$0.09} & \shadecell{0.4}{0.74} & \shadecell{0.2}{0.46} & \shadecell{0.6}{1.01} &  & \textbf{\shadecell{0.1}{0.33$\pm$0.07}} & \textbf{\shadecell{0.1}{0.32}} & \textbf{\shadecell{0.0}{0.14}} & \textbf{\shadecell{0.2}{0.57}} &  \\
    NM-0.3        & \shadecell{5.3}{2.80$\pm$0.16} & \shadecell{5.2}{2.79} & \shadecell{4.4}{2.31} & \shadecell{6.3}{3.26} &  & \shadecell{4.9}{2.71$\pm$0.21} & \shadecell{4.9}{2.71} & \shadecell{3.4}{2.03} & \shadecell{6.8}{3.38} &  & \shadecell{4.2}{2.49$\pm$0.18} & \shadecell{4.2}{2.49} & \shadecell{3.2}{1.98} & \shadecell{5.6}{3.05} &  & \textbf{\shadecell{2.4}{1.87$\pm$0.18}} & \textbf{\shadecell{2.4}{1.88}} & \textbf{\shadecell{1.5}{1.34}} & \textbf{\shadecell{3.5}{2.43}} &  \\
    NM-0.5        & \shadecell{10.7}{3.99$\pm$0.16} & \shadecell{10.7}{4.00} & \shadecell{10.0}{3.48} & \shadecell{11.5}{4.40} &  & \shadecell{14.4}{4.62$\pm$0.26} & \shadecell{14.3}{4.61} & \shadecell{11.2}{3.69} & \shadecell{17.4}{5.41} &  & \shadecell{10.6}{3.98$\pm$0.17} & \shadecell{10.6}{3.98} & \shadecell{10.0}{3.48} & \shadecell{11.5}{4.40} &  & \textbf{\shadecell{9.1}{3.67$\pm$0.24}} & \textbf{\shadecell{9.0}{3.67}} & \textbf{\shadecell{6.6}{2.83}} & \textbf{\shadecell{12.3}{4.54}} &  \\
    NM-0.7        & \textbf{\shadecell{14.6}{4.66$\pm$0.13}} & \textbf{\shadecell{14.6}{4.67}} & \textbf{\shadecell{14.4}{4.18}} & \textbf{\shadecell{15.0}{5.01}} &  & \shadecell{24.8}{6.07$\pm$0.21} & \shadecell{24.8}{6.07} & \shadecell{22.6}{5.23} & \shadecell{26.6}{6.68} &  & \shadecell{14.6}{4.66$\pm$0.13} & \shadecell{14.6}{4.67} & \shadecell{14.2}{4.15} & \shadecell{15.0}{5.01} &  & \shadecell{19.3}{5.35$\pm$0.24} & \shadecell{19.2}{5.35} & \shadecell{17.3}{4.57} & \shadecell{22.6}{6.16} &  \\
    NM-0.9        & \textbf{\shadecell{15.6}{4.81$\pm$0.11}} & \textbf{\shadecell{15.6}{4.82}} & \textbf{\shadecell{16.0}{4.41}} & \textbf{\shadecell{15.4}{5.08}} &  & \shadecell{29.0}{6.56$\pm$0.12} & \shadecell{29.0}{6.57} & \shadecell{28.6}{5.89} & \shadecell{27.9}{6.83} &  & \shadecell{15.6}{4.81$\pm$0.11} & \shadecell{15.6}{4.82} & \shadecell{16.0}{4.41} & \shadecell{15.4}{5.08} &  & \shadecell{30.4}{6.72$\pm$0.20} & \shadecell{30.4}{6.73} & \shadecell{30.1}{6.04} & \shadecell{32.0}{7.32} &  \\
    IC            & \textbf{\shadecell{15.7}{4.82$\pm$0.11}} & \textbf{\shadecell{15.7}{4.83}} & \textbf{\shadecell{15.7}{4.36}} & \textbf{\shadecell{15.6}{5.11}} &  & \shadecell{29.0}{6.57$\pm$0.12} & \shadecell{29.1}{6.58} & \shadecell{30.9}{6.11} & \shadecell{27.9}{6.84} &  & \shadecell{15.7}{4.82$\pm$0.11} & \shadecell{15.7}{4.83} & \shadecell{15.7}{4.36} & \shadecell{15.6}{5.11} &  & \shadecell{33.0}{7.00$\pm$0.16} & \shadecell{33.0}{7.01} & \shadecell{33.0}{6.32} & \shadecell{33.0}{7.44} &  \\
    URN-0.2       & \textbf{\shadecell{5.8}{2.93$\pm$0.59}} & \textbf{\shadecell{6.1}{3.01}} & \textbf{\shadecell{0.3}{0.61}} & \textbf{\shadecell{11.3}{4.35}} &  & \shadecell{13.5}{4.48$\pm$0.74} & \shadecell{14.2}{4.59} & \shadecell{1.1}{1.16} & \shadecell{20.8}{5.91} &  & \shadecell{5.5}{2.85$\pm$0.60} & \shadecell{5.8}{2.93} & \shadecell{0.3}{0.61} & \shadecell{10.9}{4.28} &  & \shadecell{11.1}{4.06$\pm$0.86} & \shadecell{11.6}{4.16} & \shadecell{0.4}{0.70} & \shadecell{23.1}{6.22} &  \\
    URN-0.5       & \shadecell{2.6}{1.96$\pm$0.79} & \shadecell{2.7}{2.00} & \shadecell{0.0}{0.12} & \shadecell{8.4}{3.75} &  & \shadecell{7.5}{3.34$\pm$1.07} & \shadecell{8.2}{3.49} & \shadecell{0.0}{0.21} & \shadecell{17.9}{5.48} &  & \textbf{\shadecell{2.2}{1.80$\pm$0.77}} & \textbf{\shadecell{2.2}{1.80}} & \textbf{\shadecell{0.0}{0.08}} & \textbf{\shadecell{7.7}{3.60}} &  & \shadecell{4.6}{2.62$\pm$1.07} & \shadecell{4.6}{2.61} & \shadecell{0.0}{0.04} & \shadecell{18.6}{5.58} &  \\
    SP/Walsh      & \shadecell{4.1}{2.46$\pm$0.13} & \shadecell{4.1}{2.46} & \shadecell{3.0}{1.90} & \shadecell{4.9}{2.88} &  & \shadecell{5.3}{2.79$\pm$0.21} & \shadecell{5.2}{2.78} & \shadecell{3.7}{2.10} & \shadecell{7.6}{3.57} &  & \shadecell{3.5}{2.27$\pm$0.14} & \shadecell{3.5}{2.27} & \shadecell{2.8}{1.85} & \shadecell{4.4}{2.73} &  & \textbf{\shadecell{0.0}{0.00$\pm$0.00}} & \textbf{\shadecell{0.0}{0.00}} & \textbf{\shadecell{0.0}{0.00}} & \textbf{\shadecell{0.0}{0.00}} &  \\
    SP/Conitzer   & \shadecell{2.5}{1.92$\pm$0.17} & \shadecell{2.5}{1.93} & \shadecell{1.7}{1.42} & \shadecell{3.5}{2.44} &  & \shadecell{2.3}{1.86$\pm$0.23} & \shadecell{2.3}{1.85} & \shadecell{1.3}{1.24} & \shadecell{4.5}{2.76} &  & \shadecell{1.9}{1.67$\pm$0.18} & \shadecell{1.9}{1.67} & \shadecell{1.0}{1.10} & \shadecell{2.7}{2.15} &  & \textbf{\shadecell{0.0}{0.00$\pm$0.00}} & \textbf{\shadecell{0.0}{0.00}} & \textbf{\shadecell{0.0}{0.00}} & \textbf{\shadecell{0.0}{0.00}} &  \\
    SC            & \shadecell{5.1}{2.76$\pm$0.36} & \shadecell{5.2}{2.78} & \shadecell{2.0}{1.54} & \shadecell{7.8}{3.62} &  & \shadecell{7.8}{3.39$\pm$0.35} & \shadecell{7.8}{3.41} & \shadecell{4.1}{2.24} & \shadecell{11.3}{4.34} &  & \shadecell{4.1}{2.48$\pm$0.28} & \shadecell{4.2}{2.50} & \shadecell{1.8}{1.48} & \shadecell{5.7}{3.10} &  & \textbf{\shadecell{1.8}{1.63$\pm$0.40}} & \textbf{\shadecell{1.8}{1.64}} & \textbf{\shadecell{0.1}{0.30}} & \textbf{\shadecell{4.5}{2.76}} &  \\
    1D EUC-cube   & \shadecell{1.3}{1.36$\pm$0.43} & \shadecell{1.3}{1.38} & \shadecell{0.1}{0.27} & \shadecell{4.2}{2.67} &  & \shadecell{1.9}{1.69$\pm$0.46} & \shadecell{2.0}{1.71} & \shadecell{0.1}{0.33} & \shadecell{6.0}{3.18} &  & \shadecell{0.6}{0.97$\pm$0.30} & \shadecell{0.6}{0.98} & \shadecell{0.0}{0.05} & \shadecell{1.9}{1.79} &  & \textbf{\shadecell{0.0}{0.00$\pm$0.00}} & \textbf{\shadecell{0.0}{0.00}} & \textbf{\shadecell{0.0}{0.00}} & \textbf{\shadecell{0.0}{0.00}} &  \\
    2D EUC-cube   & \shadecell{5.6}{2.88$\pm$0.41} & \shadecell{5.6}{2.90} & \shadecell{2.1}{1.60} & \shadecell{9.2}{3.92} &  & \shadecell{9.2}{3.69$\pm$0.55} & \shadecell{9.3}{3.73} & \shadecell{2.4}{1.69} & \shadecell{15.2}{5.05} &  & \textbf{\shadecell{4.5}{2.59$\pm$0.40}} & \textbf{\shadecell{4.6}{2.60}} & \textbf{\shadecell{1.4}{1.30}} & \textbf{\shadecell{7.6}{3.57}} &  & \shadecell{5.3}{2.81$\pm$0.62} & \shadecell{5.5}{2.85} & \shadecell{0.5}{0.81} & \shadecell{10.0}{4.09} &  \\
    3D EUC-cube   & \shadecell{7.9}{3.42$\pm$0.37} & \shadecell{8.0}{3.45} & \shadecell{3.5}{2.05} & \shadecell{11.7}{4.43} &  & \shadecell{13.2}{4.42$\pm$0.47} & \shadecell{13.3}{4.45} & \shadecell{6.7}{2.84} & \shadecell{19.2}{5.68} &  & \textbf{\shadecell{7.1}{3.25$\pm$0.37}} & \textbf{\shadecell{7.1}{3.26}} & \textbf{\shadecell{3.3}{1.99}} & \textbf{\shadecell{11.4}{4.36}} &  & \shadecell{8.6}{3.57$\pm$0.58} & \shadecell{8.7}{3.60} & \shadecell{2.1}{1.59} & \shadecell{16.3}{5.23} &  \\
    4D EUC-cube   & \shadecell{9.2}{3.70$\pm$0.33} & \shadecell{9.3}{3.72} & \shadecell{5.6}{2.61} & \shadecell{12.0}{4.48} &  & \shadecell{15.8}{4.85$\pm$0.43} & \shadecell{15.9}{4.87} & \shadecell{9.9}{3.47} & \shadecell{21.1}{5.95} &  & \textbf{\shadecell{8.7}{3.60$\pm$0.33}} & \textbf{\shadecell{8.8}{3.61}} & \textbf{\shadecell{5.6}{2.59}} & \textbf{\shadecell{11.9}{4.46}} &  & \shadecell{11.0}{4.04$\pm$0.58} & \shadecell{11.0}{4.05} & \shadecell{4.0}{2.20} & \shadecell{18.9}{5.62} &  \\
    5D EUC-cube   & \textbf{\shadecell{10.1}{3.88$\pm$0.32}} & \textbf{\shadecell{10.2}{3.90}} & \textbf{\shadecell{5.7}{2.62}} & \textbf{\shadecell{13.4}{4.74}} &  & \shadecell{17.5}{5.09$\pm$0.38} & \shadecell{17.6}{5.12} & \shadecell{11.2}{3.69} & \shadecell{21.6}{6.02} &  & \shadecell{9.8}{3.82$\pm$0.33} & \shadecell{9.9}{3.83} & \shadecell{5.7}{2.62} & \shadecell{13.3}{4.73} &  & \shadecell{12.4}{4.29$\pm$0.58} & \shadecell{12.5}{4.31} & \shadecell{4.2}{2.26} & \shadecell{20.1}{5.80} &  \\
    GS/cat        & \textbf{\shadecell{0.0}{0.00$\pm$0.00}} & \textbf{\shadecell{0.0}{0.00}} & \textbf{\shadecell{0.0}{0.00}} & \textbf{\shadecell{0.0}{0.00}} &  & \shadecell{28.3}{6.48$\pm$0.17} & \shadecell{28.5}{6.51} & \shadecell{28.4}{5.86} & \shadecell{28.3}{6.89} &  & \shadecell{0.0}{0.00$\pm$0.00} & \shadecell{0.0}{0.00} & \shadecell{0.0}{0.00} & \shadecell{0.0}{0.00} &  & \shadecell{24.4}{6.01$\pm$0.27} & \shadecell{24.3}{6.02} & \shadecell{22.3}{5.20} & \shadecell{27.3}{6.76} &  \\
    GS/bal        & \shadecell{6.6}{3.12$\pm$0.13} & \shadecell{6.6}{3.14} & \shadecell{5.6}{2.61} & \shadecell{6.8}{3.39} &  & \textbf{\shadecell{0.0}{0.00$\pm$0.00}} & \textbf{\shadecell{0.0}{0.00}} & \textbf{\shadecell{0.0}{0.00}} & \textbf{\shadecell{0.0}{0.00}} &  & \shadecell{0.0}{0.00$\pm$0.00} & \shadecell{0.0}{0.00} & \shadecell{0.0}{0.00} & \shadecell{0.0}{0.00} &  & \shadecell{6.6}{3.13$\pm$0.13} & \shadecell{6.6}{3.14} & \shadecell{5.7}{2.62} & \shadecell{6.9}{3.40} &  \\
    \bottomrule
  \end{tabular}
  }
\end{table*}
}

\iflong 
In \Cref{app:sec:experiments:preflib}, we also provide 
\fi 
a similar analysis of real-life elections from PrefLib~\cite{mat-wal:c:preflib}.

\toappendix{
\subsubsection{Real-life Elections: PrefLib}
\label{app:sec:experiments:preflib}
    We analyzed all elections with full strict ordering and at least~$3$ and at most~$15$ candidates from PrefLib~\cite{mat-wal:c:preflib}. There are $3,685$ such elections. Out of these, we were able to compute distances for $3,102$ of them. The rest timed out within $1$ hour per election time limit.

    \begin{figure}%
        \centering
        \includegraphics[width=\linewidth]{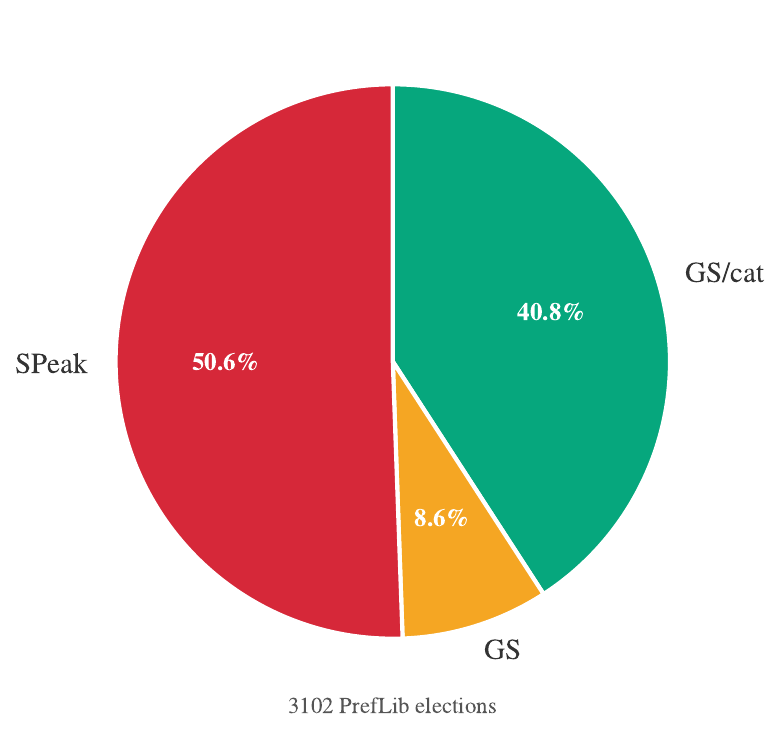}
        \caption{The overall fraction of elections where $\GScat$, $\GSbal$, $\GSbin$, or $\SPeak$ is the closest structured domain.}
        \label{fig:results:preflib:pie}
    \end{figure}

    \begin{figure}[bt]%
        \centering
        \includegraphics[width=\linewidth]{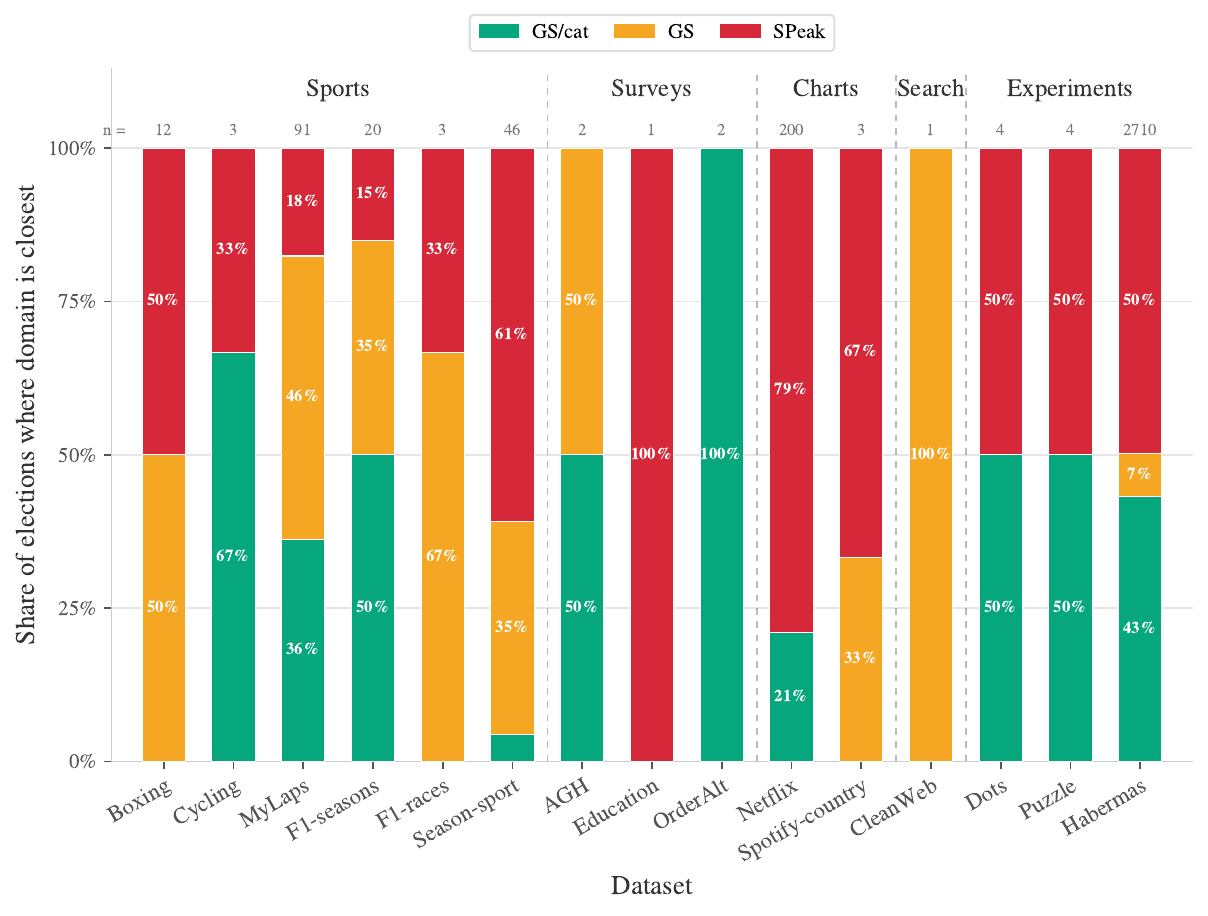}
        \caption{An overview of closest domains per dataset and category.}
        \label{fig:results:preflib:categories}
    \end{figure}

    \begin{figure}[bt]%
        \centering
        \includegraphics[width=0.95\linewidth]{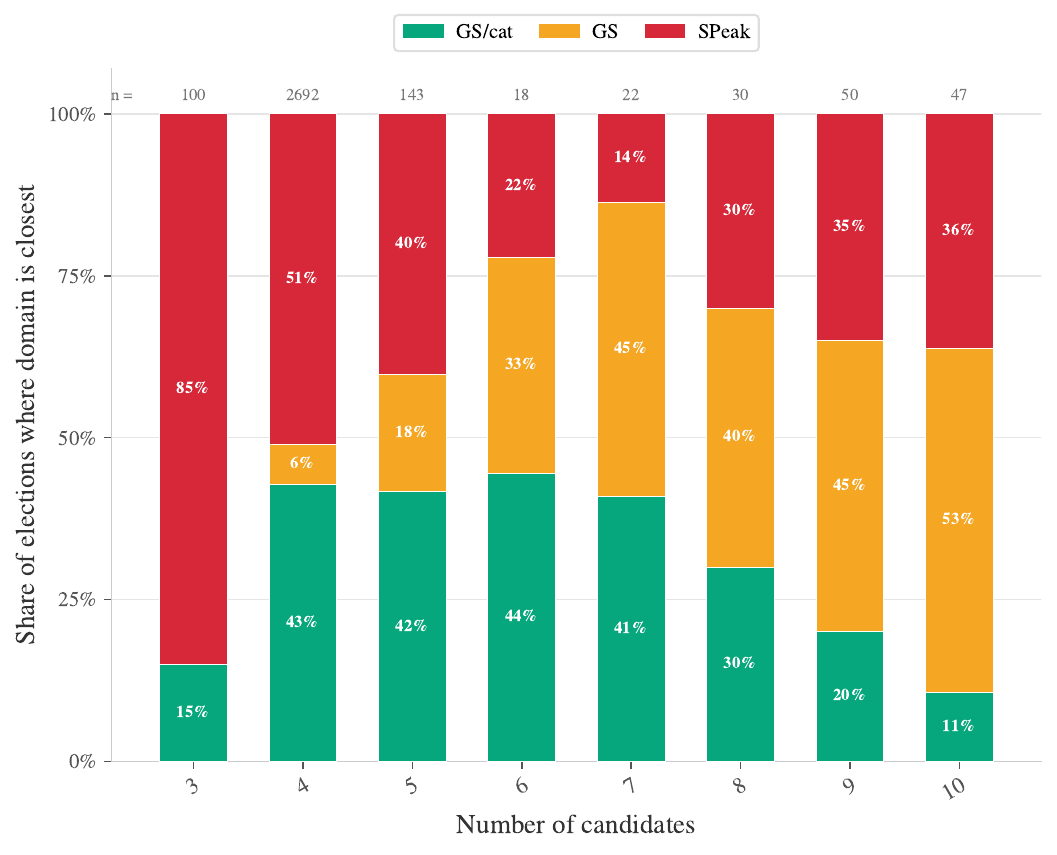}
        \caption{An overview of closest domains per the number of candidates within an election.}
        \label{fig:results:preflib:numCandidates}
    \end{figure}

    We present our results using three figures. First, \Cref{fig:results:preflib:pie} provides an overview of the entire dataset. Clearly, for roughly half of the elections satisfying the conditions above, the closest domain is $\SPeak$, and for roughly half of them, it is some variant of $\GS$. Surprisingly, for $40.8\%$ of all the elections, the closest domain is $\GScat$; however, it should be mentioned that this overview is highly affected by the dataset of \texttt{Habermas}, which on its own contains $2,710$ elections.

    Therefore, we also present an overview of the different datasets separately; see \Cref{fig:results:preflib:categories}. Moreover, we group different datasets based on their categories. Here, we can see that for datasets of sport events, $\GScat$, respectively $\GS$ domain is dominating. On the other hand, the elections based on charts are much more often close to the $\SPeak$ domain. For experiment-based datasets, roughly half of the elections are $\SPeak$ and the rest are $\GS$, with a substantial portion of them being $\GScat$.

    Finally, we also explore how the distance to the closest domain changes with the increasing number of candidates. These results are depicted in \Cref{fig:results:preflib:numCandidates}. The number of elections closest to the $\SPeak$ domain does not have any intuitive explanation. However, there is a clear indication that, with the increasing number of candidates, the proportion of elections where some $\GS$ domain is the closest is increasing at the expense of the $\GScat$ domain.
}

\toappendix{\FloatBarrier
}
\subsection{Retaining Winners}
\label{sec:experiments:winners}
\appendixsubsection{sec:experiments:winners}

\toappendix{

Regarding the winner-retaining experiment, we also evaluate the Single Transferable Vote (STV) rule; see \Cref{fig:stv}. Because our experimental setup strictly considers single-winner elections, the STV mechanism effectively operates as a standard iterative elimination process. The threshold for victory is established by the Droop quota, $q$, which in a single-seat scenario simplifies to an absolute majority requirement $q = \left\lfloor \frac{\numVoters}{2} \right\rfloor + 1$. If no candidate secures this quota, the candidate with the fewest first-preference votes is eliminated, and their ballots are transferred to the voters' next available preferences. Although our experiment utilizes the Weighted Inclusive Gregory (WIG) method, surplus vote redistribution is not triggered, as the election concludes immediately once the quota is met.

For scoring-based voting methods, we evaluate the impact of domain transformations on the original winners by analyzing their score differences. To ensure comparability, we first apply min-max normalization to the candidates' scores in each election. Next, when comparing an original election to its closest transformed counterparts, we identify the smallest score difference experienced by any of the original winning candidates. This minimal score difference is then averaged across all the closest transformed elections associated with that single original election. Finally, we compute the grand mean of these values across all sampled elections, alongside the standard error of the mean; see \Cref{fig:score-diff}.

\begin{figure}[h!]
    \centering
    \includegraphics[width=\linewidth]{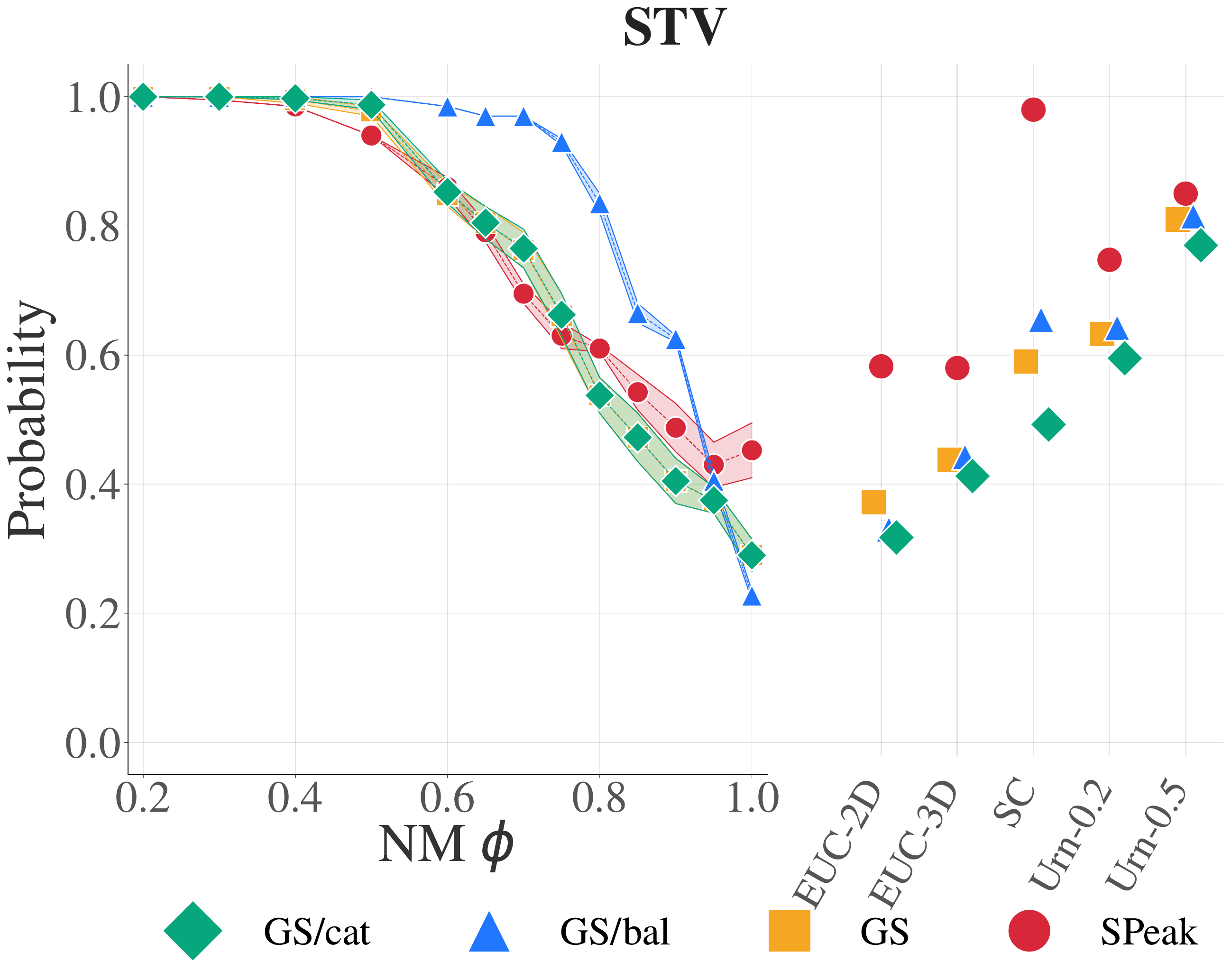}
    \caption{Probability of retaining original winners (success rate) under the STV voting rule for $\GScat$, $\GSbal$, $\GS$, and $\SPeak$ projections.}
    \label{fig:stv}
\end{figure}

\begin{strip}
\InsertBoxC{
\includegraphics[width=0.9\linewidth]{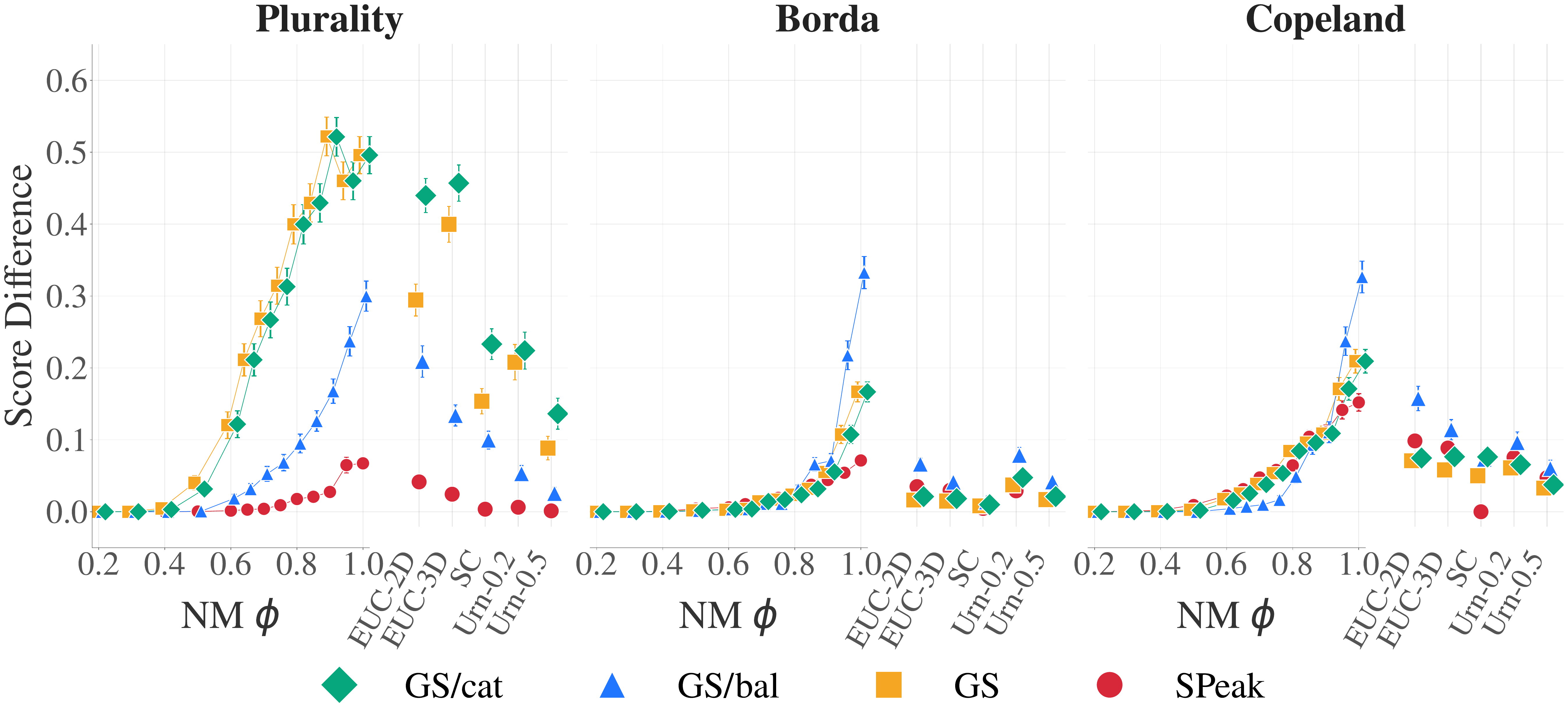}}
\captionof{figure}{{Minimal normalized score difference of original winners for $\GScat$, $\GSbal$, $\GS$, and $\SPeak$ projections.}\label{fig:score-diff}
}
\end{strip}

\FloatBarrier
}

\begin{figure}
    \centering

    \includegraphics[width=\linewidth]{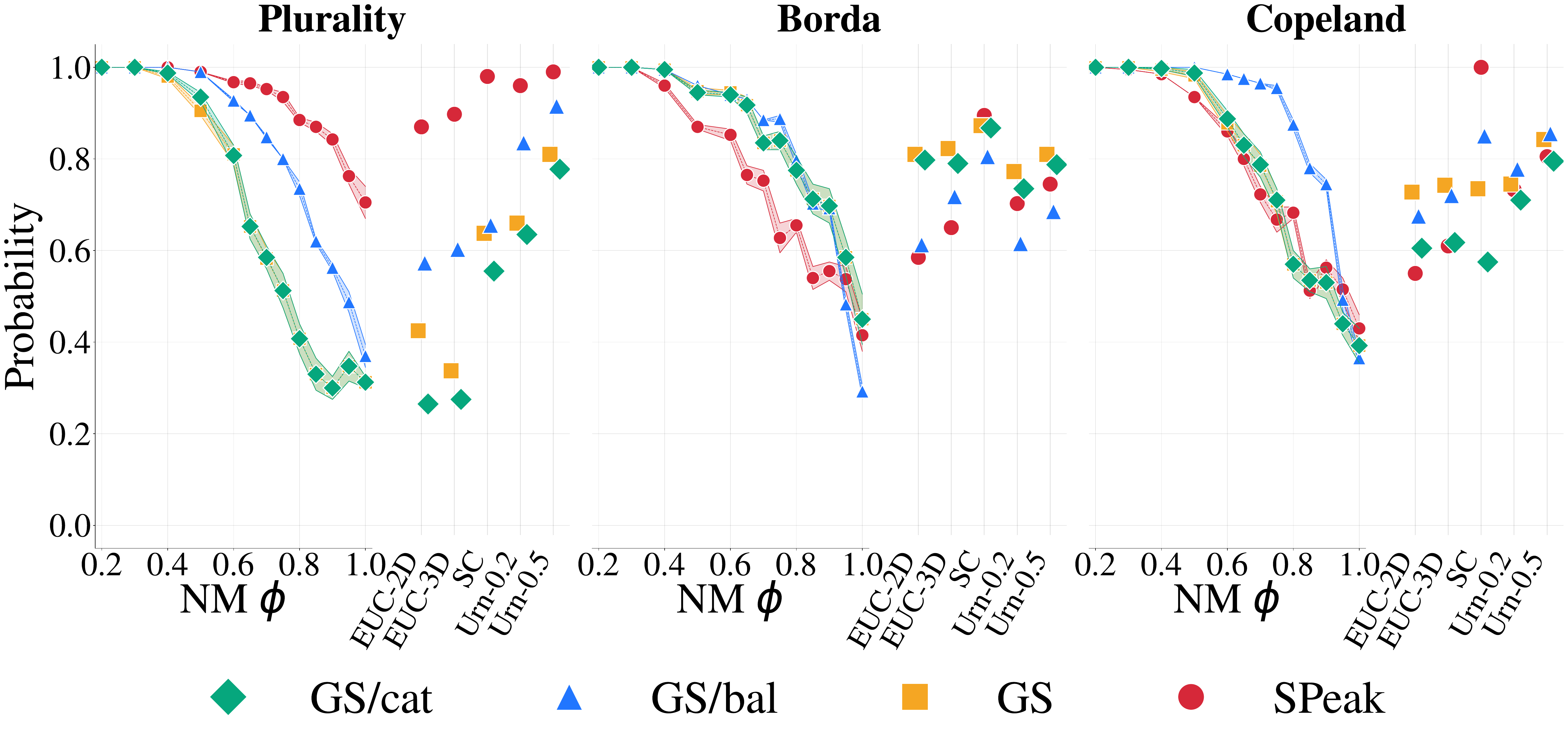}

    \caption{Probability of retaining winners (being successful) for $\GScat$,
      $\GSbal$, $\GS$, and $\SPeak$ projections.%
      }
    \label{fig:plurality-borda-copeland}
\end{figure}

Finally, we checked how likely the projections are to retain the same winners as the original elections. To maintain focus, we considered Plurality, Borda, and Copeland voting rules. For each of the considered models, i.e., for NM-$\phi$, EUC-$t$D, SC, and URN-$\alpha$, we have sampled $200$ elections (with $8$ candidates and $96$ voters), computed their closest $\GScat$, $\GSbal$, $\GSbin$, and single-peaked elections, and computed the winners according to each of our rules. Since the rules can output sets of winners, we refer to a projection as \emph{successful} if it shares at least one of the winners with the original election. Further, it is possible that there are several closest elections of a given type to a given original one, each with a possibly different set of winners. In \Cref{fig:plurality-borda-copeland}, we report the probabilities that the closest election of a given type is successful; additionally, the top line gives the probability that at least one of the projections is successful, and the bottom line gives the probability that all projections are successful.

Overall, our projections are successful with probability at least $50\%$, and the more structured the original election is, the higher the chance of success. That said, $\GScat$ projections are least likely to be successful (except for the Borda rule, where they typically are on par with other domains). $\GSbal$ projections perform best for Copeland, average for Plurality, and either average or worst for Borda. Single-peaked projections are the best for Plurality and perform well overall.

\section{Conclusions}

We have found that many (synthetic) elections (with 8 candidates and
96 voters) are close to being group-separable (GS).  Further, the less
structured they are, the more likely their closest GS elections are
caterpillar. Yet, it was the $\GSbal$ elections that maintained the
high-level features of the original elections best, maintaining the
same winners most frequently. To obtain these results, we have
designed a number of algorithms computing closest structured elections, solving Open Problem~5 of
\citet{elk-lac-pet:t:restricted-domains-survey}.

\section*{Acknowledgments}

This project has received funding from the European Research Council (ERC) under the European Union's Horizon 2020 research and innovation programme (grant agreement No 101002854), was co-funded by the European Union under the project Robotics and Advanced Industrial Production (reg. no. CZ.02.01.01/00/22\_008/0004590), and supported in part by the National Science Centre, Poland, grant number UMO-2025/58/A/ST6/00371.

\bibliography{references}

\iflong

\newpage

\toappendix{
\section{ILP for \probName{Nearly GS-Caterpillar Election}}\label{sec:ILP:caterpillar}
\FloatBarrier

We fix an arbitrary ordering $\triangleright$ on the candidate set $C$. This ordering is used only for indexing in constraints, and whenever we write $a \triangleright b$, we refer to it. 
For the relative order of two candidates in an input vote $\vote$ and in its corresponding vote $\vote'$ in the closest $\GScat$ election, we write $\succ_v$ and $\succ_{\vote'}$, respectively. Finally, for the mapping $\mapping = (\mapping_1, \ldots, \mapping_m)$ of the candidates to leaves and a pair of candidates $a, b \in C$, we write $a \succ_\mapping b$ to indicate that $a$ precedes $b$ in $\mapping$.  

We will use the following constants:
\begin{enumerate}
    \item \emph{Input vote position constants.} For every voter $\vote \in \voters$ and pair of candidates $a \triangleright b$, the value of $p^v_{ab}  \in \{0, 1\}$ indicates their relative order in the input vote $\vote$. Here, $p^v_{ab} = 1$ if $a \succ_\vote b$, and $p^v_{ab} = 0$ otherwise.
\end{enumerate}

We will use the following binary variables:
\begin{enumerate}
    \item \emph{Mapping variables}. For every pair of candidates $a \triangleright b$, the value of $s_{ab} \in \{0, 1\}$ indicates their relative order in the leaf mapping $\mapping$ as $s_{ab} = 1$ if $a \succ_\mapping b$ and $s_{ab} = 0$ otherwise.
    
    \item \emph{Flip variables}. For every voter $v \in V$ and candidate $a \in C$, the value of $f^v_a \in \{0, 1\}$ indicates whether the voter decides to prepend the candidate $a$ before the candidates that appear below in $\mapping$ in the resulting vote $v'$, in which case $f^v_a = 0$, or append the candidate $a$ to the rest of $\mapping$, in which case $f^v_a = 1$.
    
    \item \emph{Resulting votes position variables}. For every voter $v \in V$ and every pair of candidates $a \triangleright b$, the value of $r^v_{ab} \in \{0, 1\}$ indicates the relative order of $a$ and $b$ in the resulting vote~$\vote'$. Here, $r^v_{ab} = 1$ if $a \succ_{\vote'} b$ and $r^v_{ab} = 0$ otherwise.
\end{enumerate}

We will use two types of constraints:
\begin{enumerate}
    \item \emph{Mapping validity constraints} to ensure the leaf mapping variables encode a valid linear order. Antisymmetry is enforced by the variables encoding itself; hence we only need to ensure the transitivity by forbidding directed 3-cycles with the following two constraints:
    \begin{align*}
        s_{ab} + s_{bc} + (1 - s_{ac}) &\leq 2 \quad \forall a \triangleright b \triangleright c \in C \\
        (1 - s_{ab}) + (1 - s_{bc}) + s_{ac} &\leq 2 \quad \forall a \triangleright b \triangleright c \in C
    \end{align*}
    Note that it is enough to consider only these two constraints -- restricting to triples with $a \triangleright b \triangleright c$ does not exclude any case. For example, the 3-cycle $a \triangleleft b \triangleright c$ is the same 3-cycle as $b \triangleright a \triangleright c$ or $b \triangleright c \triangleright a$ (depending on whether $a \triangleleft c$ or $c \triangleleft a$).
    
    \item \emph{Group-separability constraints} to ensure that the resulting votes belong to the caterpillar domain $\domain(\Tcat, \mapping)$. This is enforced by the following constraints for all $\vote \in V$ and $a \triangleright b \in C:$
    \begin{align*}
        r^v_{ab} &= 1 - f^v_a ~~~~\text{ if } a \succ_\mapping b \\ 
        r^v_{ab} &= f^v_b ~~~~~~~~~~~\text{ if } b \succ_\mapping a
    \end{align*}
    These constraints are linearized in a standard way as follows:
    \begin{align*}
        r^v_{ab} &\le (1 - f^v_a) + (1 - s_{ab}) \quad \forall v \in V ~~ \forall a \triangleright b \in C \\
        r^v_{ab} &\ge (1 - f^v_a) - (1 - s_{ab}) \quad \forall v \in V ~~ \forall a \triangleright b \in C \\
        r^v_{ab} &\le f^v_b + s_{ab} \quad~~~~~~~~~~~~~~~~~~~~ \forall v \in V ~~ \forall a \triangleright b \in C \\
        r^v_{ab} &\ge f^v_b - s_{ab} \quad~~~~~~~~~~~~~~~~~~~~\forall v \in V ~~ \forall a \triangleright b \in C 
    \end{align*}
    Observe that when $s_{ab} = 1$ for $a \triangleright b$, the first two constraints yield $r^v_{ab} \leq (1 - f^v_a)$ and $r^v_{ab} \geq (1 - f^v_a)$, forcing $r^v_{ab} = 1 - f^v_a$, while the third and fourth inequalities hold regardless of the value of $f^v_b$; the reverse holds for $s_{ab} = 0$. Therefore, in the rest of the proof, we work directly with the compact formulation instead of the linearized inequalities.
\end{enumerate}

To compute the minimum swap distance between the input vote $\vote$ and the resulting vote $\vote'$, we minimize the number of inversions between them. Hence, as the objective function we set

\begin{align*}
    \min \sum_{v \in V} \sum_{a \triangleright b \in C} |p^v_{ab} - r^v_{ab}|.
\end{align*}
This can be linearized as
\begin{align*}
    \min \sum_{v \in V} \sum_{a \triangleright b \in C} p^v_{ab}(1 - r^v_{ab}) + (1 - p^v_{ab})r^v_{ab}
\end{align*} 

In total, the ILP formulation uses $\mathcal{O}(nm^2)$ variables and $\mathcal{O}((m+n)m^2)$ constraints. We now prove the correctness of the construction.

\emph{Forward implication.}
Assume that $\mapping = (\mapping_1, \ldots, \mapping_m)$ is a mapping to the leaves of the caterpillar tree $\Tcat$ and that $V' = (v'_1, \ldots, v'_n)$ is a list of votes from the $\domain(\Tcat, \mapping)$ domain. We show that then there exists an assignment to the ILP variables satisfying the constraints. We define the assignment as follows: for every pair of candidates $a \triangleright b$, we set $s_{ab} = 1$ if and only if $a \succ_\mapping b$, and for every $\vote \in V$ and $a \triangleright b$, set $r^v_{ab} = 1$ if and only if $a \succ_{\vote'} b$. Finally, based on the resulting vote $\vote'$, we set the flip variables as follows. For a vote $\vote \in V$ and every candidate $a \neq \mapping_m$, i.e. every candidate that is not the last in the mapping $\mapping$, we set $f^v_a = 0$ if and only if $a \succ_{\vote'} b$ for every $b$ with $a \succ_\mapping b$, i.e. for every candidate that follows $a$ in the leaf mapping $\mapping$. The value of the flip variable $f^v_{\mapping_m}$ for the last candidate $\mapping_m$ in $\mapping$ may be set arbitrarily.

We argue that this assignment of variables satisfies the ILP constraints. First, assume for contradiction that the mapping validity constraints are not satisfied. That is, there exists a triple of candidates $a \triangleright b \triangleright c$ such that $s_{ab} + s_{bc} + (1 - s_{ac}) = 3$ or $(1 - s_{ab}) + (1 - s_{bc}) + s_{ac} = 3$. In the first case, this implies that $s_{ab} = s_{bc} = 1$ and $s_{ac} = 0$. However, this would yield $a \succ_\mapping b$, $b \succ_\mapping c$ and $c \succ_\mapping a$. i.e., a directed 3-cycle in the relation induced by the $s$ variables. But, since $s_{ab} = 1$ if and only if $a \succ_\mapping b$, this would imply that $\succ_\mapping$ is not transitive, contradicting the fact that $\mapping$ is a valid linear order.
The second case would then imply $s_{ab} = s_{bc} = 0$ and $s_{ac} = 1$ and the existence of a reversed directed 3-cycle; again, a contradiction.

Next, we show that the assignment satisfies the group-separability constraints. Let $a$ and $b$ be a pair of candidates such that $a \triangleright b$. Based on their relative position in $\mapping$, we distinguish two cases.

\emph{Case A: $a \succ_\mapping b$}. In this case, the first constraint ${r^v_{ab} = 1 - f^v_a}$ applies. Based on the relative order of~$a$ and~$b$ in the resulting vote $\vote'$, we distinguish two subcases. If $a \succ_{\vote'} b$, then the resulting vote variable was set to $r^v_{ab} = 1$, and the flip variable was set to $f^v_a = 0$ 
(as $b$ is after $a$ in $\mapping$ and appears below $a$ in the resulting vote $\vote'$). 
If $b \succ_{\vote'} a$, then the variables were set to $r^v_{ab} = 0$ and $f^v_a = 1$. In both cases, the equality is satisfied.

\emph{Case B: $b \succ_\mapping a$}. In this case, the second constraint $r^v_{ab} = f^v_b$ applies. Again, based on the relative order of $a$ and $b$ in~$\succ_{\vote'}$, there are two possibilities: If $a \succ_{\vote'} b$, then the variables were set to $r^v_{ab} = 1$ and $f^v_b = 1$. If $b \succ_{\vote'} a$, then the variables were set to $r^v_{ab} = 0$ and $f^v_b = 0$. Again, in both cases, the equality holds.

\emph{Backward implication.}
Assume now that we have a valid assignment to the variables that satisfies the constraints of our ILP. We show that the assignment corresponds to a mapping~$\mapping$ and a list of votes $V' = (v'_1, \ldots, v'_n)$ from the caterpillar domain. For convenience, we extend the resulting vote variables antisymmetrically to all ordered pairs by setting $r^v_{ba} := 1 - r^v_{ab}$ for every $a \triangleright b$.

First, observe that by the definition of the mapping variables $s$, the induced relation $\mapping$ is irreflexive and antisymmetric (as we have exactly one variable defined per each unordered pair $a, b \in C$, indicating their relative order). The constraints forbidding directed 3-cycles then enforce transitivity, so $\mapping$ is a strict linear order on $C$.

Next, we argue that the resulting vote variables $r$ induce a strict linear order as well. Irreflexivity and antisymmetry follow by definition, by the same argument as for mapping variables. 
Transitivity is enforced by the transitivity of the $s$ variables together with the flip variables $f$ and vote group-separability constraints, as we now argue. Assume that the relation induced by $r$ is not transitive. Therefore, there exists a directed 3-cycle on some triple of candidates $a \triangleright b \triangleright c$. Similarly as in the description of the mapping validity constraints, this implies that either $r^v_{ab} + r^v_{bc} + (1 - r^v_{ac} )= 3$ or~$(1 - r^v_{ab}) + (1 - r^v_{bc}) + r^v_{ac}= 3$. We now argue that this leads to a contradiction. First assume that $r^v_{ab} + r^v_{bc} + (1 - r^v_{ac} )= 3$. This would mean that $r^v_{ab} = r^v_{bc} = 1$ and $r^v_{ac} = 0$. We proceed by case analysis based on how $a, b, c$ are ordered by~$\succ_\mapping$. If $a \succ_\mapping b \succ_\mapping c$, then from the first group-separability constraint and $r^v_{ab} = r^v_{bc} = 1$ we know that $f^v_a = f^v_b = 0$. But $r^v_{ac} = 0$ and the first group-separability constraint imply that $f^v_a = 1$, a contradiction. Next, consider the case when  $b \succ_\mapping a \succ_\mapping c$, then $r^v_{ab} = 1$ together with the second group-separability constraint yield that $f^v_b = 1$, and $r^v_{bc} = 1$ together with the first group-separability constraint imply $f^v_b = 0$; a contradiction. For other possible orderings of $a, b, c$ by $\succ_\mapping$ one can argue analogously. Dual arguments then hold for the case when $(1 - r^v_{ab}) + (1 - r^v_{bc}) + r^v_{ac}= 3$, with the values of $f^v$ flipped from $0$ to $1$ and vice-versa.

Finally, we show that the votes induced by the $r$ variables belong to the caterpillar domain $\domain(\Tcat, \mapping)$.
Let $a$ be an arbitrary candidate. We show that in every vote $\vote'$, $a$ must be ranked either above all $b$ such that $a \succ_\mapping b$, or below. We proceed by case analysis on the $\triangleright$-order of $a$ and $b$.
First, in case $a \triangleright b$, the relative order of the candidates is determined by the constraint $r^v_{ab} = 1 - f^v_{a}$. We distinguish two subcases. If $f^v_{a} = 0$, then $r^v_{ab} = 1$, yielding that $a$ is ranked above all $b$ such that $a \succ_\mapping b$. If $f^v_{a} = 1$, then $r^v_{ab} = 0$ , yielding that $a$ is ranked below all $b$ such that $a \succ_\mapping b$.
Second, in case $b \triangleright a$, the relative order of $a$ and $b$ is determined by the second constraint $r^v_{ba} = f^v_a$. Again, we have two possibilities. If~$f^v_a = 0$, then $r^v_{ba} = 0$, yielding that $a$ is ranked above $b$. If~$f^v_{a} = 1$, then $r^v_{ba} = 1$, yielding that $a$ is ranked below~$b$.
Since $b$ was an arbitrary candidate such that $a \succ_\mapping b$, it follows that $a$ is ranked either above or below all candidates following in $\mapping$, and hence the votes $V' = (v'_1, \ldots, v'_n)$ belong to $\domain(\Tcat, \mapping)$.

Since in both directions the objective is equal to the total swap distance between the corresponding votes, the optimal ILP solutions correspond to the closest $\GScat$ elections.
}

\toappendix{
\section{ILP for \probName{Nearly GS-Balanced Election}}\label{sec:ILP:balanced}
Given an integer $k$, we write $[k]_0$ to denote $[k] \cup \{0\}$. We assume that the number of candidates $\numCandidates$ is a power of two. Throughout this section, we denote by $\Tbal$ a perfect binary \typeQ-tree with $m$ leaves. To simplify the notation, we identify each candidate $a \in C$ with the leaf $\mapping(a)$ it is mapped to. In particular, we speak of candidates lying in the subtrees of~$\Tbal$ and write $\lca(a, b)$ for the lowest common ancestor of~$a$ and~$b$.

    We use the following constants:
    \begin{enumerate}
        \item \emph{Input vote position constants}. For every voter $\vote \in \voters$ and every ordered pair of candidates $a \neq b \in C$, the value of~$p^v_{ab}\in \{0, 1\}$ encodes the position of $a$ and $b$ in the input vote. The value is $1$ if $a \succ_\vote b$ and $0$ otherwise.
        \item $L = \log_2{m} \in \mathbb{N}$ to denote the depth of the tree $\Tbal$. Levels are indexed by $\ell \in [L - 1]_0$, i.e. the root vertex is on level $\ell = 0$.
    \end{enumerate}

    \noindent
    We introduce the following variables:
    \begin{enumerate}
        \item \emph{The tree variables }$T_a^{\ell} \in \{0, 1\}$ to encode the positions of candidates in the leaf mapping $\mapping$ as a binary string. Equal to $0$ if the candidate $a$ is in the left subtree of its ancestor node at level $\ell$, and $1$ otherwise.
        
        \item \emph{The flip variables} $F_{va}^\ell \in \{0 ,1\}$. Intuitively, at each inner node $u$ of $\Tbal$, the voter $\vote$ decides whether the candidates in the left subtree of $u$ should be ranked above the candidates in the right subtree of $u$, following the order of the mapping $\mapping$, or below them (in which case a flip occurs). The variable is equal to $1$ if voter $\vote$ flips the order at candidate $a$'s unique ancestor on level $\ell$, and $0$ otherwise.

        \item \emph{The resulting vote variables}. For every voter $\vote \in \voters$ and every ordered pair of candidates $a \neq b \in C$, the value of~$r^v_{ab} \in \{0, 1\}$ indicates the relative order of $a$ and $b$ in the resulting vote $\vote'$. Equals to $1$ if $a \succ_{\vote'} b$ in the resulting vote $\vote'$, and $0$ otherwise.

        \item \emph{Auxiliary variables}. In the construction, we further use binary variables $t^\ell_{ab}, S^{<\ell}_{ab}, A^{\ell}_{ab}, \alpha^{\ell}_{va}$ and ${\beta^{\ell}_{vab} \in \{0, 1\}}$, whose meaning we specify with the constraints that use them.
    \end{enumerate}

    \noindent
    We introduce the following constraints:
    \begin{enumerate}
        \item \emph{The tree constraints.} In order to have a valid mapping~$\mapping$ of the candidates to the leaves of $\Tbal$, there need to exist $m$ unique binary encodings via the tree variables. This means that, for every pair of candidates $a \neq b \in C$, there must exist at least one level $\ell \in [L-1]_0$ such that $T^\ell_a \neq T^\ell_b$. Let $t_{ab}^\ell := |T^\ell_a - T^\ell_b|$ denote the difference of the variables on level $\ell$. Then we can express this requirement through the constraints as follows:
        \begin{align*}
            t^\ell_{ab} &= |T^\ell_a - T^\ell_b| && \forall a \neq b \in C && \forall \ell \in [L - 1]_0 \\
            \sum_{\ell = 0}^{L- 1}& t^\ell_{ab} \geq 1 && \forall a \neq b \in C &&
        \end{align*}

        The first constraint can be linearized as:
        \begin{align*}
            t^\ell_{ab} &\ge T^\ell_a - T^\ell_b \\
            t^\ell_{ab} &\ge T^\ell_b - T^\ell_a \\
            t^\ell_{ab} &\le T^\ell_a + T^\ell_b \\
            t^\ell_{ab} &\le 2 - (T^\ell_a + T^\ell_b)
        \end{align*}

        \item \emph{The flip constraints.} Let $a \neq b \in C$ be a pair of candidates. Observe that the relative position of $a$ and $b$ in the resulting vote $\vote'$ is decided at the flip decision at their lowest common ancestor $\lca(a, b)$ node in $\Tbal$; in all levels strictly above, the candidates belong to the same subtrees. Therefore, we need to enforce that in all inner nodes on the path from the root to $\lca(a, b)$ (including), the flip variables for $a$ and $b$ must have the same values.

        To do so, we use auxiliary variable $S^{<\ell}_{ab} \in \{0, 1\}$ that is equal to $1$ if $a$ and $b$ have the same ancestor in all levels $\ell' < \ell$, and $0$ otherwise.  Trivially, strictly before the root level $\ell= 0$, all candidates have the same ancestor and therefore we set
        $S^{<0}_{ab} = 1$ for all $a \neq b \in C$.
        For every other level $\ell \geq 1$, we require that $S^{<\ell}_{ab} = 1$ if and only if the binary addresses of $a$ and $b$ (encoded via the tree variables) agree on all levels strictly below~$\ell$. That is, using the auxiliary variable $t^{\ell'}_{ab} := |T^{\ell'}_a - T^{\ell'}_b|$ from the tree constraints, we require that 
        $S^{<\ell}_{ab} = 1 \iff t^{\ell'}_{ab} = 0$ for all $\ell'~<~\ell$, linearized as:
        \begin{align*}
            S^{<\ell}_{ab} &\leq 1 - t^{\ell'}_{ab} ~~~ \forall \ell' < \ell \\
            S^{<\ell}_{ab} &\geq 1 - \sum_{\ell' = 0} ^{\ell-1} t^{\ell'}_{ab}
        \end{align*}

        Finally, for any level $\ell \in [L - 1]_0$ where $a$ and $b$ have the same ancestor, we require that the flip variables for $a$ and $b$ must be equal. Formally, for every pair $a \neq b \in C$, every voter $v \in V$ and every level $\ell \in [L - 1]_0$ we require that:
        \begin{align*}
            S^{<\ell}_{ab} = 1 \Rightarrow F^\ell_{va} = F^\ell_{vb}
        \end{align*}
        Linearized as:
        \begin{align*}
            1 - S^{<\ell}_{ab} \geq F^\ell_{va} - F^\ell_{vb} \\
            1 - S^{<\ell}_{ab} \geq F^\ell_{vb} - F^\ell_{va}
        \end{align*}
        
        \item \emph{The group-separability constraints.} Finally, we tie the leaf mapping $\mapping$ and the flip variables and enforce that the resulting votes belong to $\domain(\Tbal, \mapping)$. Recall that the relative position of two candidates $a$ and $b$ is determined by the decision (flip) at $\lca(a, b)$.
        To encode the lowest common ancestor level, we use an auxiliary variable $A^{\ell}_{ab}$, which we require to equal $1$ if $a$ and $b$ agree in all levels strictly below $\ell$ and disagree in their binary addresses at $\ell$. Formally, we enforce that $A^{\ell}_{ab} = 1$ if and only if $S^{<\ell}_{ab} = 1$ and~$t^{\ell}_{ab} = |T^\ell_a - T^\ell_b| =  1$. Additionally, we require that exactly one $\ell \in [L - 1]_0$ satisfies $A^\ell_{ab} = 1$. We express these two requirements with the following four constraints:
        \begin{align*}
            A^\ell_{ab} &\leq S^{<\ell}_{ab} && \forall a \neq b \in C && \forall \ell \in [L - 1]_0 \\ 
            A^\ell_{ab} &\leq t^{\ell}_{ab} && \forall a \neq b \in C && \forall \ell \in [L - 1]_0 \\ 
            A^\ell_{ab} &\geq S^{<\ell}_{ab} + t^{\ell}_{ab} - 1 && \forall a \neq b \in C && \forall \ell \in [L - 1]_0 \\ 
            \sum_{\ell = 0}^{L - 1}&A^{\ell}_{ab} = 1 && \forall a \neq b \in C &&
        \end{align*}

        To encode the relative position of $a$ and $b$ in the resulting vote $\vote'$, we further use an auxiliary variable $\alpha^\ell_{va} \in \{0, 1\}$. Let $u$ be the ancestor of $a$ on level $\ell$. The $\alpha$ variable expresses whether on level $\ell$, the candidate $a$ should be ranked above the candidates in the subtree of $u$ not containing $a$, or below. In particular, if $T^\ell_a = 0$ and $F^\ell_{va} = 0$ (or $T^\ell_a = 1$ and $F^\ell_{va} = 1$), voter $v$ ranks $a$ above the other subtree, in which case we require $\alpha^\ell_{va} = 0$. Symmetrically, if $T^\ell_a = 0$ and $F^\ell_{va} = 1$ (or $T^\ell_{a} = 1$ and $F^\ell_{va} = 0$), $a$ should be ranked below all candidates in the other subtree, in which case we require $\alpha^\ell_{va} = 1$.
        Formally, we require that 
        \begin{align*}
            \alpha^\ell_{va} = |T^\ell_a - F^\ell_{va}| ~~~ \forall \vote \in V ~~ \forall a \in C ~~ \forall \ell \in [L - 1]_0.
        \end{align*}

        This constraint can be linearized analogously to the tree constraint $t^\ell_{ab} = |T^\ell_a - T^\ell_b|$.

        To connect the resulting vote variables, we additionally introduce auxiliary variable $\beta^\ell_{vab} \in \{0, 1\}$. For a pair of candidates $a$ and $b$ and level $\ell \in [L - 1]_0$, we want $\beta^\ell_{vab} = 1$ if and only if $\ell$ is the level containing $\lca(a, b)$ and $a$ is supposed to be ranked above $b$, i.e., $A^\ell_{ab} = 1$ and $\alpha^\ell_{va} = 0$. Formally, for all $v \in V$, $a \neq b \in C$ and $\ell \in [L - 1]_0$ we enforce that
        \begin{align*}
            \beta^\ell_{vab} = A^\ell_{ab} (1 - \alpha^\ell_{va}).
        \end{align*}
        This can be linearized for every $v \in V$,  every $a \neq b \in C$, and every $\ell \in [L - 1]_0$ as
        \begin{align*}
            \beta^\ell_{vab} &\leq A^\ell_{ab} \\
            \beta^\ell_{vab} &\leq 1 - \alpha^\ell_{va} \\
            \beta^\ell_{vab} &\geq A^\ell_{ab} - \alpha^\ell_{va} \\
            \beta^\ell_{vab} &\geq 0
        \end{align*}

        Summing across all levels $\ell \in [L - 1]_0$, we then determine the final order of $a$ and $b$ as follows:
        \begin{align*}
            r^v_{ab} = \sum_{\ell = 0}^{L - 1} \beta^\ell_{vab} ~~~ \forall v \in V ~~~ \forall a \neq b \in C
        \end{align*}
    \end{enumerate}

    \proofsubparagraph{Objective Function.}
    To compute the minimal swap distance between the input and resulting votes, we minimize the number of inversions between the original vote variables and resulting vote variables:
    \begin{align*}
        \min \sum_{v \in V}\sum_{a \triangleright b \in C} |p^v_{ab} - r^v_{ab}|
    \end{align*}

    Here, $\triangleright$ is an arbitrary fixed ordering on the candidate set~$C$, used so each pair of candidates is counted exactly once. The objective function can be linearized as:
    \begin{align*}
        \min \sum_{v \in V} \sum_{a \triangleright b \in C}  p^v_{ab}(1 - r^v_{ab}) + (1 - p^v_{ab})r^v_{ab}.
    \end{align*}
    
    In total, the ILP formulation uses $\mathcal{O}(nm^2 \log m)$ variables and $\mathcal{O}((n + \log m ) m^2 \log m)$ constraints.
    We now prove the correctness. 
    
    \proofsubparagraph{Forward implication.} Let $\mapping = (\mapping_1, \ldots, \mapping_m)$ be a mapping of the candidates to the leaves of $\Tbal$, ordered from the left-most leaf to the right-most, and let $V' = (\vote'_1, \ldots, \vote'_n)$ be a list of votes from $\domain(\Tbal, \mapping)$.
    We show that there exists an assignment of the variables that satisfies the ILP constraints.

    First, we assign the values to the tree variables $T$ according to $\mapping$. For the $i$-th candidate~$\mapping_i$, we express its position $i$ in binary, with the position of the first candidate $\mapping_1$ being encoded as all zeros. Then, we set $T^\ell_{\mapping_i}$ to be the value of the $\ell$-th position of the binary string encoding $\mapping_i$'s position.
    Next, given a vote $\vote'$, we assign the values of the flip variables~$F$ as follows. Since $\vote'  \in \domain(\Tbal, \mapping)$, the candidates mapped by $\mapping$ to the leaves of every subtree form a contiguous block in $\vote'$. For each level $\ell \in [L - 1]_0$ and each inner node~$u$ at level $\ell$, we set $F^\ell_{va} := 0$ for every candidate $a$ in the subtree of $u$ if the candidates in the left subtree of $u$ lie in the first half of $u$'s block in vote $\vote'$, and $F^\ell_{va} := 1$ otherwise.
    Further, given a vote $v'$ and pair of candidates $a \neq b \in C$, we set $r^v_{ab} = 1$ if and only if $a \succ_{\vote'} b$. Finally, we set the values of the auxiliary variables $t, S, A, \alpha$ and $\beta$ according to their defining equations.
    
    We now argue that this assignment of variables satisfies the constraints of our ILP. The tree constraints are clearly satisfied, as the values of the $T$ variables correspond to the binary encoding of a position of a candidate $a$ in the leaf ordering $\mapping$, and the encodings are unique. 
    Next, we verify the flip constraints. Let $\vote \in V$ be any voter, $a \neq b \in C$ an arbitrary pair of candidates and $\ell$ the level containing their lowest common ancestor node. We argue that $F^{\ell'}_{va} = F^{\ell'}_{vb}$ for every $\ell' \leq \ell$. By the definition of the lowest common ancestor, $a$ and $b$ share the path from $\ell$ up to the root. That is, for every $\ell' \leq \ell$, they lie in a subtree of the same inner node~$u$ at level $\ell'$. As our assignment gives the same value to all candidates that lie in the subtree of $u$, we get that $F^{\ell'}_{va} = F^{\ell'}_{vb}$. Next, observe that there is indeed exactly one $\ell \in [L-1]_0$ such that $A^{\ell}_{ab} = 1$: for any pair of candidates $a \neq b \in C$, it is the smallest $\ell$ such that $T^{\ell}_a \neq T^{\ell}_b$, i.e. the first bit, where their binary addresses differ. For every pair, this $\ell$ is unique.

    Finally, we verify the group-separability constraints.
    Let $a \neq b \in C$ be a pair of candidates, $u$ their lowest common ancestor, and $\ell$ the level containing $u$. First, observe that $a$ and $b$ lie in different subtrees of $u$. Second, since $\vote'$ is in the domain $\domain(\Tbal, \mapping)$, the left and right subtrees of $u$ form two contiguous blocks. In other words, the relative position of $a$ and $b$ is determined precisely at $u$ and the level $\ell$. To show the correctness, we need to show that the constraint $r^v_{ab} = \sum_{\ell = 0}^{L - 1} \beta^\ell_{vab}$ holds, where $\beta^\ell_{vab} = A^\ell_{ab} (1 - \alpha^\ell_{va})$. Observe that by our construction, $A^{\ell'}_{ab} = 1$ precisely (and only) at $\ell' = \ell$, and hence the constraint collapses to $r^v_{ab} = 1 - \alpha^\ell_{va}$.
    
    We proceed by case analysis based on the relative order of $a$ and $b$ in $\vote'$. First, assume that $a \succ_{\vote'} b$; then $r^v_{ab} = 1$ by our construction. Based on the relative order of $a$ and $b$ in~$\succ_\mapping$, we have two subcases. If $a \succ_\mapping b$, then $a$ lies in the left subtree and $b$ lies in the right subtree of $u$, i.e. $T^\ell_a = 0$ and~$T^\ell_b = 1$. Since $a \succ_{\vote'} b$, in the vote $\vote'$, the block of candidates of the left subtree lies in the first half of $u$'s block, and our construction sets  $F^\ell_{va} = F^\ell_{vb} = 0$, informally meaning the voter did not decide to flip the order at  node~$u$. It follows that 
    $\alpha^\ell_{va}  = |T^\ell_{a} - F^\ell_{va}| = 0 = 1 - r^v_{ab}$ as desired. Symmetrically, if $b \succ_\mapping a$, then $a$ lies in the right subtree and $b$ lies in the left subtree of $u$, i.e. $T^\ell_a = 1$ and $T^\ell_b = 0$. Then the block of candidates of the left subtree lies behind the block of the candidates of the right subtree of $u$, and our construction sets $F^\ell_{va} = 1$. In this case, we again obtain that $\alpha^\ell_{va} = |T^\ell_a - F^\ell_{va}| = 0 = 1 - r^v_{ab}$.
    The case when $b \succ_{\vote'} a$ is symmetric, giving $\alpha^\ell_{va} = 1$ and $r^v_{ab} = 0$ instead.

    \proofsubparagraph{Backward implication.} Assume we have an  assignment to the variables that satisfies the ILP constraints. We show that this assignment corresponds to a valid mapping $\mapping$ of candidates to the leaves of $\Tbal$ and to a list of votes $V'~=~(\vote'_1, \ldots, \vote'_n)$ in $\domain(\Tbal, \mapping)$.

    For a candidate $a \in C$, consider the values of $T^\ell_{a}$ for each $\ell \in [L-1]_0$, ordered from $0$ to $L-1$, and concatenate them into a binary string. This gives us a binary encoded position for every candidate. Since we have $L = \log_2 m$ levels, there are $2^{L} = m$ such strings. Moreover, these strings are guaranteed to be unique: our constraints ensure that for every pair of distinct candidates $a \neq b \in C$, there exists at least one $\ell$ for which $T^\ell_a \neq T^\ell_b$. Consequently, the values of~$T$ define a valid mapping $\mapping$ of the candidates to the leaves of the tree.

    First, we argue that the relation induced by the resulting vote variables $r$ is antisymmetric. Fix a pair of candidates $a \neq b \in C$. By our constraints, there exists precisely one $\ell \in [L-1]_0$ such that $A^\ell_{ab} = 1$. Observe that then the constraint $r^v_{ab} = \sum_{\ell = 0}^{L - 1} A^\ell_{ab} (1 - \alpha^\ell_{va})$ collapses to 
    $r^v_{ab} = A^\ell_{ab}(1 - \alpha^\ell_{va}) = 1 - \alpha^\ell_{va}$, as for every other $\ell' \neq \ell$ the value of $A^{\ell'}_{ab} = 0$.
    Without loss of generality, assume that $T^\ell_a = 0$ and $T^\ell_b = 1$. By the flip constraints, we have $F^\ell_{va} = F^\ell_{vb}$. Since $\alpha^\ell_{vc} = 0 \iff T^\ell_c = F^\ell_{vc}$ for any candidate $c$, it follows that exactly one $c \in \{a, b\}$ has $\alpha^\ell_{vc} = 0$, while the other has $\alpha^\ell_{vc} = 1$.
    Consequently, exactly one of $r^v_{ab}$ and $r^v_{ba}$ equals to $1$, while the other equals to $0$, and hence the relation is antisymmetric.

    Next, we argue that the values of $r^v_{ab}$ induce a transitive relation.
    Let $a, b, c \in C$ be any triple of candidates and $v \in V$ an arbitrary voter. We show that if $r^v_{ab} = 1$ and $r^v_{bc} = 1$, then also $r^v_{ac} = 1$.  
    Let $\ell_{ab}, \ell_{ac}$ and $\ell_{bc}$ denote the levels of the lowest common ancestors of the corresponding pairs of candidates. Let $\ell = \min \{\ell_{ab}, \ell_{ac}, \ell_{bc}\}$. Note that two of these three values are equal. Let $u$ be the lowest common ancestor of all three candidates $a, b, c$. Since the tree is binary, one subtree of $u$ needs to contain exactly one of the three candidates, while the other needs to contain the remaining two. The two pairs involving the separated candidate then need to have $u$ as their lowest common ancestor, while the lowest common ancestor of the remaining pair lies in a level strictly below $u$.
    Based on which two levels among $\ell_{ab}, \ell_{ac}$ and $\ell_{bc}$ are equal, we therefore distinguish three cases.

    First, if $\ell_{ab} = \ell_{ac} = \ell$, the candidate $a$ is separated at level~$\ell$ and $T^\ell_a \neq T^\ell_c = T^\ell_b$. Since $r^v_{ab} = 1$ by our assumption, it follows from the constraint $r^v_{ab} = A^\ell_{ab} (1 - \alpha^\ell_{va})$ that ${\alpha^\ell_{va} = 0}$. As $\ell$ is the level containing $\lca(a, c)$, it follows that $r^v_{ac} = A^\ell_{ac}(1 - \alpha^\ell_{va}) = 1$ as desired. For the second case, assume $\ell = \ell_{ab} = \ell_{bc}$, i.e, the candidate $b$ is separated at level~$\ell$ and $T^\ell_b \neq T^\ell_a = T^\ell_c$. Then, analogously to the previous case, $r^v_{ab} = 1$ yields $\alpha^\ell_{va} = 0$, and hence by the definition of~$\alpha^\ell_{va}$, we get that $T^\ell_a =  F^\ell_{va}$. By the same argument, $r^v_{bc} = 1$ implies $\alpha^\ell_{vb} = 0$ and $T^\ell_b = F^\ell_{vb}$. We know that up to the level containing the lowest common ancestor (including), the flip variables need to be equal, and in particular it must hold that $F^\ell_{va} = F^\ell_{vb} = F^\ell_{vc}$. But then, by combining the facts that $T^\ell_a =  F^\ell_{va}$ and $T^\ell_b = F^\ell_{vb}$, we obtain  $T^\ell_a = F^\ell_{va} = F^\ell_{vb} = T^\ell_{b}$,  a contradiction with our assumption that $T^\ell_{a} \neq T^\ell_{b}$. The second case is, therefore, not possible. The third case is then analogous to the first. Here, assume that $\ell = \ell_{ac} = \ell_{bc}$, i.e., $c$ is the separated candidate and $T^\ell_c \neq T^\ell_a = T^\ell_b$. Again, from $r^v_{bc} = 1$ we obtain that $\alpha^\ell_{vb} = 0$ and hence $T^\ell_b = F^\ell_{vb}$. From our assumption that $T^\ell_{a} = T^\ell_{b}$ we then get that $F^\ell_{va} = F^\ell_{vb}$, and therefore $\alpha^\ell_{va} = 0$. Finally, since $A^\ell_{ac} = 1$, we obtain that $r^v_{ac} = 1 - \alpha^\ell_{va} = 1$, as desired.

    Finally, we argue that the constructed votes belong to $\domain(\Tbal, \mapping)$. Let $\vote \in V$ be an arbitrary voter and $\vote'$ the vote induced by the $r$ variables. We show that for every inner node $u$ in the tree $\Tbal$, all candidates in the left subtree of~$u$ are either ranked above all candidates in the right subtree of~$u$, or below. 

    Let $\ell$ be the level containing $u$. By our constraints, any pair of candidates $a, b$ in the subtree of $u$ shares the same value of the flip variable at $\ell$: $a$ and $b$ agree on all levels $\ell' < \ell$, so $S^{<\ell}_{ab} = 1$ and the constraint $S^{<\ell}_{ab} = 1 \implies F^\ell_{va} = F^\ell_{vb}$ forces the flip variables to be equal. Assume now that $a$ is a candidate in the left subtree of $u$, and $b$ is a candidate in the right subtree of $u$, and that $\lca(a, b) = u$. In this case, $T^\ell_{a} = 0$, $T^\ell_{b} = 1$ and $A^{\ell}_{ab} = 1$. The constraints then yield $r^v_{ab} = 1 - \alpha^\ell_{va} = 1 - |T^\ell_{a} - F^\ell_{va}|$. Hence, if $F^\ell_{va} = 0$, i.e., the voter does not flip the order at $u$, $r^v_{ab} = 1$ and $a$ is ranked above $b$; if $F^\ell_{va} = 1$, we get that $r^v_{ab} = 0$ and $a$ is ranked below~$b$ in the vote $\vote'$. 
    Since all candidates contained in the subtree rooted at $u$ share the same flip value at $u$, say $f$, it follows that every candidate of the left subtree is ranked above every candidate in the right subtree (in case $f = 0$), or below every candidate in the right subtree (in case $f = 1$). As this holds at every inner node of $\Tbal$, it follows that the votes obtained from the $r$ variables belong to $\domain(\Tbal, \mapping)$. 

    Since in both directions the objective equals the total swap distance between the input and resulting votes, the optimal ILP solutions correspond to the closest $\GSbal$ elections.

}

\clearpage
\appendix
\crefalias{section}{appendix}
\crefalias{subsection}{appendix}
\crefalias{subsubsection}{appendix}

\appendixtext 
\fi 

\end{document}